\documentclass[12pt,reqno]{article}
\pdfoutput=1
\usepackage[a4paper,width=170mm,top=25mm,bottom=25mm]{geometry}
\usepackage{hyperref}
\hypersetup{
    colorlinks=false,                       
    linkcolor=red,                          
    citecolor=blue
}
\usepackage{setspace}
\usepackage{fancyhdr}
\usepackage{amsmath,amsthm,amsfonts,amssymb,braket,mathtools,mathrsfs}
\usepackage{url}
\usepackage{bbm, dsfont}
\usepackage{bbold}
\usepackage{enumerate}
\usepackage{multirow}
\usepackage[all]{xy}
\usepackage{color}
\usepackage{graphicx}
\usepackage{caption}
\usepackage{subcaption}

\numberwithin{equation}{section}
\theoremstyle{plain}
\newtheorem{theorem}{Theorem}[section]
\newtheorem{corollary}[theorem]{Corollary}
\newtheorem{lemma}[theorem]{Lemma}
\newtheorem{proposition}[theorem]{Proposition}
\newtheorem{main theorem}{Main Theorem}
\newtheorem{thmalpha}{Theorem}

\theoremstyle{definition}

\newtheorem{remark}[theorem]{Remark}

\newtheorem{definition}[theorem]{Definition}

\newtheorem*{acknowledgements}{Acknowledgements}

\usepackage{tikz}
\usepackage{xy}
\usepackage{tikz-cd}
\usepackage{xcolor}
\usetikzlibrary{decorations.pathmorphing}

\tikzset{snake/.style={decorate, decoration=snake}}
\usetikzlibrary {arrows.meta,bending,positioning,shapes.geometric,decorations.text,decorations.markings}
\tikzset{
    partial ellipse/.style args={#1:#2:#3}{
        insert path={+ (#1:#3) arc (#1:#2:#3)}
    },
    ->-/.style={
        decoration={markings,mark=at position #1 with {\arrow{>}}},
        postaction={decorate}
    }
}
\newcommand{\VeryThinLine}{\draw[line width=0.5pt]}

\tikzset{block/.style={rectangle, draw, fill=blue!20, text width=5em, 
    text centered, rounded corners, minimum height=4em},
    cloud/.style={draw, ellipse,fill=red!20, node distance=3cm, minimum height=2em},
    line/.style={draw, -latex'},
    Cbox/.style = {circle, draw, thick, fill=white, opaque}, 
    midarrow/.style={
        decoration={
            markings,
            mark=at position 0.5 with {\arrow{>}}
        },
        postaction={decorate}
    }
    }

\newcommand{\cM}{\mathcal{M}}
\newcommand{\cN}{\mathcal{N}}
\newcommand{\coev}{\textbf{coev}}
\newcommand{\End}{\text{End}}
\newcommand{\ev}{\textbf{ev}}
\newcommand{\Hom}{\text{Hom}}

\newcommand{\image}{\text{Image }}

\newcommand{\Tr}{\mathrm{Tr}}
\newcommand{\tr}{\mathrm{tr}}

\newcommand{\Irr}{\mathrm{Irr}}

\NewDocumentCommand{\tens}{t_}
 {%
  \IfBooleanTF{#1}
   {\tensop}
   {\otimes}%
 }
\NewDocumentCommand{\tensop}{m}
 {%
  \mathbin{\mathop{\otimes}\displaylimits_{#1}}%
 }

\begin{document}
\title{Axiomatization of the Levin--Wen Wave Function}
\author{
  Zhengwei Liu$^{1,2,3}$\thanks{liuzhengwei@mail.tsinghua.edu.cn}, Zishuo Zhao$^{1}$\thanks{zishuozhao0602@gmail.com} \\
  \\
  \small $^{1}$Yau Mathematical Sciences Center, Tsinghua University, Beijing 100084, China \\
  \small $^{2}$Department of Mathematical Sciences, Tsinghua University, Beijing 100084, China \\
  \small $^{3}$Yanqi Lake Beijing Institute of Mathematical Sciences and Applications, Beijing 100407, China 
}
\date{\today}
\maketitle
\begin{abstract}
The Levin--Wen model provides a lattice realization of topological orders associated with a given unitary fusion category. A longstanding open problem is to characterize Levin--Wen ground-state wave functions intrinsically, without assuming a priori categorical symmetry data or a Hamiltonian. We address this by proposing six axioms on a family of wave functions defined on lattices at multiple scales. These axioms allow us to reconstruct the underlying unitary fusion category and prove that the resulting wave functions map to nonzero Levin--Wen ground-state vectors of the emergent category. 
\end{abstract}
\tableofcontents
\section{Introduction}

The Landau paradigm organizes phases of matter by patterns of spontaneous symmetry breaking, but topological order lies beyond this framework: gapped quantum systems can belong to distinct phases without being distinguished by a local order parameter or broken microscopic symmetry \cite{Wen1989,Wen1990}.
In $(2+1)$ dimensions, the Levin--Wen string-net models provide a paradigmatic class of exactly solvable bosonic topological phases with gapped boundaries \cite{LevinWen2005}.
Starting from a unitary fusion category (UFC) $\mathcal C$, the construction produces a commuting-projector Hamiltonian with a unique ground state on the sphere; its amplitudes are determined by the graphical calculus of $\mathcal C$, and its bulk anyon theory is encoded by the Drinfeld center of $\mathcal C$ \cite{kitaev_kong_2012}.

It is widely believed that two-dimensional topologically ordered liquid phases with gappable boundaries are exhausted by generalized string-net models \cite{LinLevinBurnell2021}. 
In the renormalization-group (RG) formulation, a gapped quantum liquid phase is described by a stable family of gapped ground states across system sizes \cite{ZengWen2015,SwingleMcGreevy2016}. 
Generalized local unitary circuits, which also allow the addition or removal of unentangled local degrees of freedom, implement the corresponding equivalence relation \cite{ChenGuWen2010,QImeetsQM2019}. 
Moreover, it is known that Levin--Wen wave functions are exact fixed points of explicit entanglement-renormalization transformations \cite{KoenigReichardtVidal2009}. 
Recently, in a simplified zero-correlation-length setting based on a slight generalization of the entanglement-bootstrap axioms \cite{ShiKatoKim2020}, Kim and Ranard \cite{RanardKim2024} showed that every state satisfying these axioms can be mapped to the ground state of a generalized Levin--Wen model by a constant-depth quantum circuit. 
What remains is an intrinsic characterization formulated directly for a scale-compatible family of wave functions, from which one can derive the equivalence of the family to a family of Levin--Wen ground-state wave functions. 

The major obstructions in pursuing such a proof are twofold. 
First, one needs to construct the emergent unitary fusion category from wave functions defined on lattices at different scales. 
Second, one needs a condition that captures the fact that the family of states is at an RG fixed point. 
Generalizing Jones's planar-algebra \cite{PlanarAlgebra2021}, the theory of $\mathbb S^n$-functionals developed by the first author in \cite{Liu2024} resolves the first obstruction. 
An $\mathbb S^n$-functional is a linear functional on the space of labelled regular stratifications of $\mathbb S^n$.
For a fixed embedded lattice, Riesz representation turns this functional into a wave function. 
A single $\mathbb S^n$-functional therefore describes a family of wave functions parameterized by embedded lattices in $\mathbb S^n$. 
In two dimensions, an $\mathbb S^2$-functional satisfying homeomorphism invariance, superposed reflection positivity, multiplicativity, and complete finiteness gives rise to a category $\mathcal{C} = \mathcal C(Z)$; multiplicativity ensures that the tensor unit is simple, so that $\mathcal C$ is a unitary fusion category \cite{Liu2024}. 
This construction produces the categorical data, but does not identify the wave functions with the ones of the Levin--Wen model by itself.

In this paper, we first observe that the remaining freedom is the one-site reduced density matrix $\mathfrak{D}$, which is $\hat{\mathbb{T}}$, the Fourier transform of a morphism $\mathbb{T}\in \Hom_{\mathcal C \boxtimes \mathcal C^{op} }(A\boxtimes A^{op} \otimes A\boxtimes A^{op} ,A\boxtimes A^{op} ) $, for the emergent generating object $A$ in $\mathcal C$ labeled on a single edge.
We introduce three additional conditions---local non-degeneracy, topological connectedness, and commutativity---that fix $\mathfrak{D}$. 
Together with homeomorphism invariance, superposed reflection positivity, and multiplicativity, these conditions characterize the Levin--Wen wave function among $\mathbb S^2$-functionals.

\begin{thmalpha}[Theorem ~\ref{thm:: reconstructed state is Levin-Wen ground state}, informal]
    Let $Z$ be an $\mathbb S^2$-functional with the following properties:
    \begin{enumerate}
        \item Homeomorphism invariance \ref{def:: HI}; 
        \item Superposed reflection positivity ~\ref{def:: reflection positivity};
        \item Multiplicativity ~\ref{def:: multiplicativity};
        \item Topological connectedness ~\ref{def:: topological connectedness};
        \item Local non-degeneracy ~\ref{def:: local non degeneracy};
        \item Commutativity ~\ref{def:: commutativity}
    \end{enumerate}
    Then the category $\mathcal{C}$ reconstructed from $Z$ is a unitary fusion category. 
    Moreover, for every stratification $\cM$ of $\mathbb S^2$ with connected $\cM^1$, a product of local unitary identifications maps the induced vector $\ket{\tilde{\Psi}_{\cM}}$ to a nonzero ground-state vector of the Levin--Wen model with input $\mathcal{C}$.
\end{thmalpha}

The three additional conditions have both physical and mathematical interpretations.
Local non-degeneracy requires the one-site reduced density matrix $\mathfrak{D}>0$; otherwise a local projection could remove redundant degrees of freedom without changing the long-range entanglement. 
Commutativity is the condition that excludes multiplicities of the edge labels. 
The topological connectedness imposes two conditions which informally say: 
\begin{enumerate}
    \item TC1: after the canonical rescaling by $\tau$ per interior face, tracing out a connected lattice in a disk gives a reduced positive operator independent of that lattice;
    \item TC2: after the same face rescaling, tracing out a disconnected lattice gives the compression of the reduced positive operator associated with a connected lattice. 
\end{enumerate}
It is these conditions that allow us to connect the wave functions on different lattices. 

Mathematically, Axioms 1-4 ensure complete finiteness as shown in Corollary~\ref{corollary:: A is sum of simple objects}, so we obtain an emergent unitary fusion category $\mathcal{C}$. 
In addition, Axioms 1-4 imply that $\mathbb{T}$ satisfies the Frobenius identity and that $\mathfrak{D}=\hat{\mathbb{T}}\geq 0$, namely the $\mathcal{F}$-positivity of $\mathbb{T}$ \ref{def:: F-positivity}. 
Axiom 5 means that $\mathfrak{D}>0$, which implies the object $A\in \mathcal{C}$ associated with an edge is regular (contains all simple objects). 
By commutativity, we then have $A = \bigoplus_{x\in \Irr(\mathcal{C})}x$.
Moreover, we prove that  $\mathbb{T}\in \Hom(\gamma^2,\gamma)$, where $\gamma = \bigoplus_{x\in \Irr(\mathcal{C})}x\boxtimes x^{\mathrm{op}}$.
The uniqueness of $\mathfrak{T}$ follows from the following surprising rigidity theorem. 

\begin{thmalpha}
    Let $\mathcal{C}$ be a unitary fusion category. 
    Set $\gamma = \bigoplus_{x\in \Irr(\mathcal{C})}x\boxtimes x^{\mathrm{op}}$, and let $\mathbb{T}\in \Hom(\gamma^2,\gamma)$.
    Suppose that $\mathbb{T}$ is invariant under the one-click rotation and modular conjugation, strictly $\mathcal{F}$-positive, and satisfies the Frobenius identity \eqref{eqn:: Frobenius identity}.
    Then there exists $\lambda>0$ such that
    \begin{equation*}
        \mathbb{T} = \bigoplus_{i,j,k\in \Irr(\mathcal{C})}\lambda\sum_{t^{ij}_k} \sqrt{d_id_jd_k} P_k \left( t^{ij}_{k}\boxtimes (t^{ij}_{k})^* \right) P_i\otimes P_j,
    \end{equation*}
    where $t^{ij}_{k}$ ranges over an orthonormal basis of $\Hom(ij,k)$ with respect to the trace inner product, and $P_i$ is the projection onto the diagonal component $i\boxtimes i^{\mathrm{op}}$.
\end{thmalpha}
It is intriguing to see that the above morphism is proportional to the multiplication map of the canonical Frobenius algebra in $\mathcal{C}\boxtimes\mathcal{C}^{\mathrm{op}}$, which implements the weak Morita equivalence with the Drinfeld center of $\mathcal{C}$ \cite{Longo1995Nets,Muger2003b}.

Removing any one of the six axioms will produce different types of models, which deserve future study. 
Meanwhile, TC is closely related to the area law and topological entanglement entropy \cite{KitaevPreskill2006,LevinWen2006} and the entanglement bootstrap axioms \cite{ShiKatoKim2020}, and we expect to develop further connections between the entanglement bootstrap approach and the one developed in this paper. 
Furthermore, a perturbation of the one-site density matrix will produce perturbed models beyond RG fixed points.
Finally, an extension of this axiomatization to higher dimensions remains a substantial challenge. 

The paper is organized as follows.
Section~\ref{section:: String-net wave function} constructs the Levin--Wen model and its ground-state wave function on $\mathbb S^2$. 
Section~\ref{sec:: S^2-functional} recalls $\mathbb S^2$-functionals and constructs their basic examples from string-net ground states.
Section~\ref{sec:: axioms} formulates the axioms above and verifies them for Levin--Wen wave functions.
Section~\ref{sec:: uniqueness of Levin-Wen wave function} develops their consequences and proves the uniqueness theorem. Appendix~\ref{app:: reconstruction from S2 functional} reviews the reconstruction of the spherical unitary tensor category $\mathcal C(Z)$ from an $\mathbb S^2$-functional.

\begin{acknowledgements}
Zhengwei Liu was supported by Beijing Natural Science Foundation Key Program (Grant No. Z220002). 
All authors were supported by Beijing Natural Science Foundation (Grant No. Z221100002722017).
\end{acknowledgements}

\section{Levin--Wen String-net Model}\label{section:: String-net wave function}

In this section, we construct the Levin--Wen model and its ground-state wave function on the $2$-dimensional sphere $\mathbb{S}^2$.
The input is a (strict) spherical unitary fusion category (UFC) $\mathcal{C}$.
We fix a set of representatives $\Irr(\mathcal{C})$ of simple objects such that $\bar{a}\in \Irr(\mathcal{C})$ for every non-self-dual object $a\in \Irr(\mathcal{C})$.
This defines an involution $(\cdot)^{-}$ on $\Irr(\mathcal{C})$ that maps each non-self-dual object $a$ to $\bar{a}$ and fixes every self-dual object.
For each self-dual object $b\in \Irr(\mathcal{C})$, we also fix a unitary isomorphism $\phi_b\in \Hom(b,\bar{b})$.

For an object $x$, denote by $\mathrm{Id}_x$ the identity morphism of $x$ and by $d_x = \tr_{\mathcal{C}}(\mathrm{Id}_x)$ its quantum dimension. 
The evaluation and coevaluation maps for an object $x$ are denoted by $\ev_x\in \Hom(\bar{x} x,\mathbb{1})$ and $\coev_x\in \Hom(\mathbb{1},x\bar{x})$, respectively.
All diagrams are read from top to bottom, with upward arrows representing dual objects. 
For a self-dual simple object $b$, we represent the unitary isomorphism $\phi_b$ by a colored dot on a $b$-labelled strand as in $\vcenter{\hbox{\begin{tikzpicture}
    \VeryThinLine[midarrow] (0,0.6) -- (0,0.05);
    \VeryThinLine[midarrow] (0,-0.6) -- (0,-0.05);
    \node[font=\scriptsize] at (0.25,0.45) {$b$};
    \node[font=\scriptsize] at (0.25,-0.45) {$b$}; 
    \draw[draw=black, fill=yellow] (0,0) circle (2pt);
\end{tikzpicture}}}$, with $\phi^{-1}_b$ as $\vcenter{\hbox{\begin{tikzpicture}
    \VeryThinLine[midarrow] (0,0.05) -- (0,0.6); 
    \VeryThinLine[midarrow] (0,-0.05) -- (0,-0.6) ;
    \node[font=\scriptsize] at (0.25,0.45) {$b$};
    \node[font=\scriptsize] at (0.25,-0.45) {$b$}; 
    \draw[draw=black, fill=yellow] (0,0) circle (2pt);
\end{tikzpicture}}}$. 

Given simple objects $a_1,a_2, \dots, a_n$ and $b_1,b_2,\dots,b_m$ in $\mathcal{C}$, we define an inner product on $\Hom(a_1 a_2 \cdots a_n, b_1 b_2 \cdots b_m)$ as follows:
\begin{equation}\label{eqn:: skein module inner product}
    \braket{g,f} = \frac{1}{\sqrt{d_{a_1}d_{a_2}\cdots d_{a_n}}}\frac{1}{\sqrt{d_{b_1}d_{b_2}\cdots d_{b_m}}}\tr_{\mathcal{C}}(g^*f),\quad f,g\in \Hom(a_1 a_2 \cdots a_n, b_1 b_2 \cdots b_m). 
\end{equation}
We will equip the hom spaces with the above inner product throughout this section. 
By $\mathrm{ONB}(ab,c)$, we mean an orthonormal basis of $\Hom(ab,c)$ with respect to this inner product. 

\subsection{Local Hamiltonian}

We now describe the local Hilbert space of the Levin--Wen string-net model on the two-dimensional sphere $\mathbb{S}^2$.
Let $\Gamma = (V,E)$ be a finite connected \emph{directed} graph embedded in $\mathbb{S}^2$, drawn in blue.
For each vertex $v\in V$, denote its valence by $|v|$.
Choose a \emph{counterclockwise} ordering of the edges incident to $v$, and denote the $i$-th edge by $e_v(i)$.
This choice will be indicated by placing a dollar sign between the first edge $e_v(1)$ and the last edge $e_v(|v|)$. 

\begin{remark}
    We use directed graphs in fixing the FS indicators, which is similar to the approach taken in \cite{Simon2022StraighteningFSindicator}. 
    When $\mathcal{C}$ contains no self-dual objects with nontrivial Frobenius-Schur indicators, undirected graphs suffice. 
\end{remark}

For a vertex $v\in V$, a boundary label is a map from the edges incident to $v$ to $\Irr(\mathcal{C})$, denoted by $\vec{a} = (a_1,a_2,\cdots, a_{|v|})$.
For a labelling $\vec{a}$, we define the Hilbert space $\mathcal{H}_v(\vec{a}) = \Hom(\mathbb{1}, a_1a_2\cdots a_{|v|})$. 
The local Hilbert space at $v$ is the orthogonal direct sum $\mathcal{H}_v = \bigoplus_{\vec{a}} \mathcal{H}_v(\vec{a})$ over all boundary labels $\vec{a}$.
The local Hilbert space of $\Gamma$ is $\mathcal{H}_{\Gamma} = \bigotimes_{v\in V} \mathcal{H}_v$.
Let $A = \bigoplus_{a\in \Irr(\mathcal{C})}a$ be the direct sum of all simple objects in $\mathcal{C}$.
Then $\mathcal{H}_v$ is unitarily equivalent to $\Hom(\mathbb{1}, A^{|v|})$, with the equivalence depending on the ordering and orientations of the edges incident to $v$.

The Levin--Wen Hamiltonian consists of two types of local interactions. 
For each edge $e\in E$ with $\partial(e) = \{v_1,v_2\}$, define a projection $Q_e$ on $\mathcal{H}_{v_1}\otimes \mathcal{H}_{v_2}$ which projects onto the span of all configurations on $v_1$ and $v_2$ that assign dual simple objects to $e$. 
The image of the product of all $Q_e$ is thus spanned by all configurations with consistent labellings, often called the string-net subspace. 

We next define the plaquette operators.
First, introduce the following local operators.
For objects $x,y$ in $\mathcal{C}$ and a morphism $f\in \Hom(x,y)$, the modular conjugation of $f$ is the following morphism in $\Hom(\bar{x},\bar{y})$:
\begin{equation}
    \overline{f} = (\mathbf{ev}_y\otimes \mathrm{Id}_{\bar{x}}) \circ \left( \mathrm{Id}_{\bar{y}}\otimes f^*\otimes \mathrm{Id}_{\bar{x}} \right)\circ (\mathrm{Id}_{\bar{y}}\otimes \mathbf{coev}_x) = \vcenter{\hbox{\begin{tikzpicture}[scale=0.75]
    \VeryThinLine (0,0.75) -- (0,0.25);
  \VeryThinLine[midarrow]  (1,0.75) arc (0:180:0.5 and 0.25);
  \VeryThinLine[midarrow] (1,-0.75) -- (1,0.75);
  \node at (1,-1) {$x$};
  \VeryThinLine(0,-0.25) -- (0,-0.75);
  \VeryThinLine[midarrow] (0,-0.75) arc (0:-180:0.5 and 0.25);
  \VeryThinLine[midarrow] (-1,-0.75) -- (-1,0.75);
  \node at (-1,1) {$y$};
  \node [draw, minimum width=0.5cm, minimum height=0.25cm, fill = white] at (0,0) {$f^*$};
  \end{tikzpicture}}}.
\end{equation}
For $a\in\Irr(\mathcal{C})$, define $\iota_a\colon a^-\to\bar a$ to be the identity if $a$ is non-self-dual and $\phi_a$ if $a$ is self-dual.
For $s,t,j\in \Irr(\mathcal{C})$ and $f\in \Hom(sj,t)$, define $\widehat{f}\in \Hom(\bar{j}s^{-},t^{-})$ by
\begin{equation}
    \widehat{f}=\iota_t^{-1}\circ\overline{f}\circ(\mathrm{Id}_{\bar{j}}\otimes\iota_s).
\end{equation}

Given simple objects $s_1,s_2,t_1,t_2,j\in \Irr(\mathcal{C})$ and morphisms $\xi\in \Hom(s_1j,t_1)$ and $\eta\in \Hom(s_2j,t_2)$, define the following morphism in $\Hom(s_1s^{-}_2,t_1t^{-}_2)$:
\begin{equation}\label{equation:: f(xi,eta)}
    \begin{aligned}
        &f(\xi,\eta) = \vcenter{\hbox{\begin{tikzpicture}
    \VeryThinLine[midarrow] (0.5,0.65) arc (60:120:1 and 1);
    \node at (0,0.5) {$j$};
    \VeryThinLine[midarrow] (-0.5,1.15) -- (-0.5,0.65);
    \VeryThinLine[midarrow] (-0.5,0.65) -- (-0.5,0);
    \VeryThinLine[midarrow] (0.5,1.15) -- (0.5,0.65);
    \VeryThinLine[midarrow] (0.5,0.65) -- (0.5,0);
    \node at (-0.75,0.65) {$\xi$};
    \node at (0.75,0.65) {$\widehat{\eta}$};
    \node at (-0.5,1.4) {$s_1$};
    \node at (-0.5,-0.25) {$t_1$};
    \node at (0.5,1.4) {$s^{-}_2$};
    \node at (0.5,-0.25) {$t^{-}_2$};
\end{tikzpicture}}}
    \end{aligned}
\end{equation}
Let $v$ be a vertex adjacent to a plaquette $p$, and suppose that $e_v(k)$ and $e_v(k+1)$ lie on the boundary of $p$ for some $1\leq k\leq |v|-1$.
Define the local operator $B_{v,p}(\xi,\eta)$ on $\mathcal{H}_v$ by
\begin{equation}
    B_{v,p}(\xi,\eta)\ket{g} = \Ket{ \left( \mathrm{Id}^{k-1}_{A}\otimes f(\xi,\eta)\otimes \mathrm{Id}^{|v|-k-1}_{A} \right)g },\quad g\in \Hom(\mathbb{1}, A^{|v|}).
\end{equation}
The operator $B_{v,p}(\xi,\eta)$ maps summands of the form $\mathcal{H}_v(a_1,\dots,a_{k-1},s_1,s^{-}_2,a_{k+2},\dots,a_{|v|})$ to $\mathcal{H}_v(a_1,\dots,a_{k-1},t_1,t^{-}_2,a_{k+2},\dots,a_{|v|})$ and acts by $0$ on all other summands.
When $k=|v|$, let $R_v:\Hom(\mathbb{1},A^{|v|})\to\Hom(\mathbb{1},A^{|v|})$ be the one-click rotation induced by the chosen duality identifications, and define $B_{v,p}(\xi,\eta)$ by
\begin{equation}
    B_{v,p}(\xi,\eta)\ket{g}
    =\Ket{R_v^{-1}\left(\left(f(\xi,\eta)\otimes \mathrm{Id}^{|v|-2}_{A}\right)R_v(g)\right)},\quad g\in \Hom(\mathbb{1}, A^{|v|}).
\end{equation}
Thus every composition in this cyclic case is interpreted after the displayed rotation, where the adjacent factors have the domains and codomains specified in the definition of $f(\xi,\eta)$.

Since $\Gamma$ is connected, the components of $\mathbb{S}^2\setminus \Gamma$ are simply connected, which we call plaquettes.
For a plaquette $p$, fix a \emph{counterclockwise} ordering $[e_1,v_1,e_2,v_2,\dots,e_{|\partial p|},v_{|\partial p|}]$ of the edges and vertices on $\partial p$, and let $\partial p^{\Irr(\mathcal{C})}$ be the set of maps from these edges to $\Irr(\mathcal{C})$.
For $\vec{s}\in \partial p^{\Irr(\mathcal{C})}$, define the operator $B^j_p(\vec{s})$ by
\begin{equation}
    B^j_p(\vec{s}) = \sum_{\vec{t}\in \partial p^{\Irr(\mathcal{C})}}\sum_{\xi_i\in \mathrm{ONB}(s_i j,t_i)} \bigotimes^{|\partial p|}_{i=1}\left( \frac{d_{t_i}}{d_{s_i}d_j} \right)^{1/4}\left( \frac{d_{t_{i+1}}}{d_{s_{i+1}}d_j} \right)^{1/4} \sigma(\vec{t})_{i+1} \sigma(\vec{s})_{i+1} B_{v_i,p}(\xi_{i+1},\xi_{i}),
\end{equation}
where the tensor product is taken in cyclic order. 
The factor $\sigma(\vec{t})_i$ is defined as follows: when $e_i$ is directed from $v_{i+1}$ to $v_i$ and $t_i$ is self-dual, set $\sigma(\vec{t})_i = \nu_{t_i}$, the Frobenius-Schur indicator of $t_i$; otherwise set $\sigma(\vec{t})_i = 1$.
We then define $B^j_p = \sum_{\vec{s}\in \partial p^{\Irr(\mathcal{C})}} B^j_p(\vec{s})$, which is independent of the ordering of $\partial p$. 
Finally, define the plaquette operator $B_p$ by
\begin{equation}
    B_p = \frac{1}{\mu^2}\sum_{j\in \Irr(\mathcal{C})} d_jB^j_p,
\end{equation}
where $\mu= \sqrt{\sum_{a\in \Irr(\mathcal{C})} d^2_a}$ is the global dimension of $\mathcal{C}$. 
The Levin--Wen Hamiltonian is defined as
\begin{equation}
    H = \sum_{e} (\mathrm{I} - Q_e) + \sum_{p} (\mathrm{I} - B_p).
\end{equation}

\subsection{The ground state wave function}

Given a simply connected region $D$ in $\mathbb{S}^2$, the local Hilbert space $\mathcal{H}_D$ is defined to be the tensor product of $\mathcal{H}_v$ for all vertices $v$ in $D$. 
The local Hamiltonian $H_D$ is defined as the sum of all interactions supported in $D$. 
Since $H_D$ is a sum of commuting projections, its ground-state space is the common eigenspace of the $Q_e$ and $B_p$ with eigenvalue $1$.
We describe this local ground-state subspace using the evaluation map \cite{kitaev_kong_2012,Green2024enrichedstringnet}.

\begin{definition}\label{def:: Levin-Wen evaluation map}
    Consider a state $\displaystyle \ket{\psi} = \ket{\otimes_{v\in D}\psi_v}$ in $\mathcal{H}_D$, where each $\psi_v$ belongs to a direct summand $\mathcal{H}_v(\vec{l}_v)$ of $\mathcal{H}_v$.
    Fix a \emph{counterclockwise} ordering of the edges in $\Gamma\cap \partial D$, and let $|\partial D|$ be the number of edges incident with $\partial D$.
    Construct a morphism in $\Hom(\mathbb{1}, A^{|\partial D|})$ as follows:
\begin{enumerate}
    \item For every vertex $v\in D$, display $\psi_v$ with the dollar sign above $v$;
    \item For every interior edge $e$ directed from $v$ to $w$, connect $\psi_v$ and $\psi_w$ by an $a$-labelled string directed from $v$ to $w$, where $a$ is the simple object assigned to $e$ by $\psi_v$; if $a$ is self-dual, place an additional $\phi_a$ at the end connected to $\psi_w$;
    \item For every vertex $v$, connect the remaining edges of $\psi_v$ to the bottom of the diagram by applying the evaluation followed by the coevaluation map.
\end{enumerate}
\end{definition}
The following example illustrates the definition, with $a,c$ non-self-dual and $b$ self-dual:
\begin{equation}
    \begin{aligned}
        \vcenter{\hbox{\begin{tikzpicture}[scale=0.75]
        \draw[fill=black] (-1.75,0) circle (1.5pt);
        \draw[fill=black] (-1,-1.5) circle (1.5pt);
        \draw[fill=black] (0,0) circle (1.5pt);
        \VeryThinLine[midarrow,blue] (-1.75,0) -- (-1,-1.5);
        \VeryThinLine[midarrow,blue] (-1,-1.5) -- (0,0);
        \VeryThinLine[midarrow,blue] (0,0) -- (-1.75,0);
        \node at (-1.65,-0.75) {$a$};
        \node at (-0.875,0.25) {$b$};
        \node at (-0.25,-0.75) {$c$};
        \VeryThinLine[blue] (-1.75,0) -- +(-0.5,0.5);
        \VeryThinLine[blue] (0,0) -- +(0.65,0.5);
        \VeryThinLine[blue] (-1,-1.5) -- +(0,-0.5);
        \node at (-2.25,0.75) {$1$};
        \node at (-1,-2.5) {$2$};
        \node at (0.65,0.75) {$3$};
        \node at (-1.75,0.25) {$\mathdollar$};
        \node at (-1.25,-1.5) {$\mathdollar$};
        \node at (0.25,-0.15) {$\mathdollar$};
    \end{tikzpicture}}}&\rightarrow 
    \vcenter{\hbox{\begin{tikzpicture}
        \coordinate (v1) at (-2,0);
        \coordinate (v2) at (0,0);
        \coordinate (v3) at (2,0);
        \draw[fill=black] (v1) circle (1.5pt);
        \draw[fill=black] (v2) circle (1.5pt);
        \draw[fill=black] (v3) circle (1.5pt);
        \VeryThinLine[midarrow] (v1) -- (-2,-0.75);
        \VeryThinLine (v1) -- (-2.5,-0.75);
        \VeryThinLine[midarrow] (v1) -- (-1.5,-0.75);
        \node at (-2,-1) {$a$};
        \node at (-1.5,-1) {$b$};
        \VeryThinLine[midarrow] (v2) -- (0,-0.75);
        \VeryThinLine (v2) -- (-0.5,-0.75);
        \VeryThinLine[midarrow] (v2) -- (0.5,-0.75);
        \node at (0,-1) {$b$};
        \node at (0.5,-1) {$c^{-}$};
        \VeryThinLine[midarrow] (v3) -- (2,-0.75);
        \VeryThinLine (v3) -- (1.5,-0.75);
        \VeryThinLine[midarrow] (v3) -- (2.5,-0.75);
        \node at (2,-1) {$c$};
        \node at (2.5,-1) {$a^{-}$};
        \node at (-2.5,-1) {$1$};
        \node at (-0.5,-1) {$3$};
        \node at (1.5,-1) {$2$};
    \end{tikzpicture}}} \rightarrow \vcenter{\hbox{\begin{tikzpicture}[scale=0.9]
        \coordinate (v1) at (-2,0);
        \coordinate (v2) at (0,0);
        \coordinate (v3) at (2,0);
        \draw[fill=black] (v1) circle (1.5pt);
        \draw[fill=black] (v2) circle (1.5pt);
        \draw[fill=black] (v3) circle (1.5pt);
        \VeryThinLine (v1) -- (-2,-0.75);
        \VeryThinLine (v1) -- (-2.5,-0.75);
        \VeryThinLine[midarrow]  (v1) -- (-1.5,-0.75);
        \VeryThinLine (v2) -- (-0,-0.75);
        \VeryThinLine (v2) -- (-0.5,-0.75);
        \VeryThinLine (v2) -- (0.5,-0.75);
        \VeryThinLine (v3) -- (2,-0.75);
        \VeryThinLine (v3) -- (1.5,-0.75);
        \VeryThinLine (v3) -- (2.5,-0.75);
        \node at (-2.5,-1) {$1$};
        \node at (-0.5,-1) {$3$};
        \node at (1.5,-1) {$2$};
        \node at (-1.45,-0.35) {$b$};
        \VeryThinLine[midarrow] (0,-0.75) arc (0:-180:0.75 and 0.75);
        \node at (-0.75,-1.75) {$b$};
        \draw[draw=black, fill=yellow] (-1.5,-0.75) circle (2pt);
        \VeryThinLine[midarrow] (2,-0.75) arc (0:-180:0.75 and 0.75);
        \node at (1.25,-1.75) {$c$};
        \VeryThinLine[midarrow] (-2,-0.75) arc (180:360:2.25 and 1.5);
        \node at (0.25,-2.5) {$a$};
    \end{tikzpicture}}}\\
    &\rightarrow \vcenter{\hbox{\begin{tikzpicture}[scale=0.9]
        \coordinate (v1) at (-2,0);
        \coordinate (v2) at (0,0);
        \coordinate (v3) at (2,0);
        \draw[fill=black] (v1) circle (1.5pt);
        \draw[fill=black] (v2) circle (1.5pt);
        \draw[fill=black] (v3) circle (1.5pt);
        \VeryThinLine (v1) -- (-2,-0.75);
        \VeryThinLine (v1) -- (-2.5,-0.75);
        \VeryThinLine[midarrow] (v1) -- (-1.5,-0.75);
        \VeryThinLine (v2) -- (-0,-0.75);
        \VeryThinLine (v2) -- (-0.5,-0.75);
        \VeryThinLine (v2) -- (0.5,-0.75);
        \VeryThinLine (v3) -- (2,-0.75);
        \VeryThinLine (v3) -- (1.5,-0.75);
        \VeryThinLine (v3) -- (2.5,-0.75);
        \node at (-2.5,-3) {$1$};
        \node at (4,-3) {$3$};
        \node at (3,-3) {$2$};
        \node at (-1.45,-0.35) {$b$};
        \VeryThinLine[midarrow] (0,-0.75) arc (0:-180:0.75 and 0.75);
        \node at (-0.75,-1.75) {$b$};
        \draw[draw=black, fill=yellow] (-1.5,-0.75) circle (2pt);
        \VeryThinLine[midarrow] (2,-0.75) arc (0:-180:0.75 and 0.75);
        \node at (1.25,-1.75) {$c$};
        \VeryThinLine[midarrow] (-2,-0.75) arc (180:360:2.25 and 1.5);
        \node at (0.25,-2.5) {$a$};
        \VeryThinLine (-2.5,-0.75) -- (-2.5,-2.75);
        \VeryThinLine (1.5,-0.75) arc (0:-180:0.25 and 0.25) -- +(0,0.75) arc (180:0:1 and 0.5) -- +(0,-2.75);
        \VeryThinLine (-0.5,-0.75) arc (0:-180:0.25 and 0.25) -- +(0,0.75) arc (180:0:2.5 and 1.5) -- +(0,-2.75);
    \end{tikzpicture}}}
    \end{aligned}
\end{equation} 
Thus $\mathrm{eval}_D: \mathcal{H}_D\rightarrow \Hom(\mathbb{1}, A^{|\partial D|})$ defines a linear map that vanishes on the complement of the common eigenspace of the edge operators $Q_e$ in $D$ with eigenvalue $1$.

\begin{lemma}\label{lemma:: eval intertwines partial Bp}
    Let $D$ be a simply connected region in $\mathbb{S}^2$ with a counterclockwise ordering of its boundary.
    Consider a plaquette $p$ such that $\partial p\cap D$ is a path $\gamma = [e_{k+1},v_1,f_1,\dots,f_{n-1},v_n,e_{k}]$, with $e_{k}$ and $e_{k+1}$ in $\partial D$. 
    Then, for any $s_1,s_2,t_1,t_2,j\in \Irr(\mathcal{C})$ and morphisms $\xi\in \Hom(s_1 j,t_1)$ and $\eta\in \Hom(s_2 j,t_2)$, we have
    \begin{equation}\label{eqn:: evaluation map and plaquette operator}
        \begin{aligned}
            &\left( \mathrm{Id}^{k-1}_{A}\otimes f(\xi,\eta)\otimes \mathrm{Id}^{|\partial D|-k-1}_{A} \right)\mathrm{eval}_D = \mathrm{eval}_D \sum_{\vec{t},\vec{s}} \sum_{\zeta_i\in \mathrm{ONB}(s_i j,t_i)}\prod^{n-1}_{i=1}\sqrt{{\frac{d_{t_i}}{d_{s_i}d_j}}} \sigma(\vec{t})_i\sigma(\vec{s})_i\\
            & B_{v_n,p}(\xi, \zeta_{n-1})\otimes B_{v_{n-1},p}(\zeta_{n-1},\zeta_{n-2})\otimes\cdots \otimes B_{v_2,p}(\zeta_2,\zeta_1)\otimes B_{v_1,p}(\zeta_1,\eta),
        \end{aligned}
    \end{equation}
    where the summation is over all $\vec{t} = (t_0,\dots,t_{n+1})$ and $\vec{s} = (s_0,\dots,s_{n+1})$ in $\gamma^{\Irr(\mathcal{C})}$ subject to the constraints $t_0=s_0 = s_1$ and $t_{n+1} = s_{n+1} = s_2$. 
\end{lemma}
\begin{proof}
    Multiplication by $\mathrm{Id}^{k-1}_{A}\otimes f(\xi,\eta)\otimes \mathrm{Id}^{|\partial D|-k-1}_{A}$ changes only the label along the path $\gamma$.
    Therefore, it suffices to verify the equality when $D \cap \Gamma = \gamma$.
    Let $\ket{\psi_i}\in \mathcal{H}_{v_i}$.
    The image of $\ket{\psi_n\otimes\cdots\otimes\psi_1}$ under the linear map on the left-hand side of Equation~\eqref{eqn:: evaluation map and plaquette operator} is
    \begin{equation}
        \vcenter{\hbox{\begin{tikzpicture}
            \begin{scope}[shift={(-0.87,1.25)}]
                \VeryThinLine (0:0) -- (-90:0.75);
                \VeryThinLine (0:0) -- (150:0.5);
                \VeryThinLine (0:0) -- (30:0.5);
                \draw[fill=black] (0:0) ellipse (0.05 and 0.05);
                \node at (-0.25,-0.125) {$\psi_n$};
            \end{scope}
            \node at (-0.55,1.65) {$s_{n-1}$};
            \begin{scope}[shift={(0,1.75)}]
                \node at (0.1,0) {$\cdots$};
            \end{scope}
            \node at (0.55,1.65) {$s_1$};
            \begin{scope}[shift={(0.87,1.25)}]
                \VeryThinLine (0:0) -- (-90:0.75);
                \VeryThinLine (0:0) -- (30:0.5);
                \VeryThinLine (0:0) -- (150:0.5);
                \draw[fill=black] (0:0) ellipse (0.05 and 0.05);
                \node at (-0.25,-0.125) {$\psi_1$};
            \end{scope}
            \VeryThinLine[midarrow] (-0.87,0.5) -- (-0.87,0);
            \VeryThinLine[midarrow] (0.87,0.5) -- (0.87,0);
            \VeryThinLine[midarrow] (0.87,0.5) -- (-0.87,0.5);
            \node at (0,0.2) {$j$};
            \node at (-1.2,0.5) {$\xi$};
            \node at (1.2,0.5) {$\widehat{\eta}$};
            \draw[dashed] (-2,0) -- (2,0);
            \node at (-1.75,0.25) {$D$};
        \end{tikzpicture}}}
    \end{equation}
    To obtain the right-hand side, bring the $j$-string close to each $s_i$-string and apply the partition of unity.
    When $s_i$ and $t_i$ are non-self-dual, this gives the following decomposition:
    \begin{equation}\label{eqn:: partition of unity}
        \vcenter{\hbox{\begin{tikzpicture}
    \VeryThinLine[midarrow] (-0.35,0) -- (-0.35,-0.5) arc (0:-180:0.25 and 0.25) -- (-0.85,0);
    \VeryThinLine[midarrow] (0.35,0) -- (0.35,-0.5) arc (0:-180: 0.95 and 0.95) -- +(0,0.5);
    \node at (-0.35,0.25) {$s_i$};
    \node at (0.35,0.25) {$j$};
    \end{tikzpicture}}} = \sum_{t_i\in \Irr(\mathcal{C})}\sum_{\zeta_i\in \mathrm{ONB}(s_ij,t_i)}\sqrt{\frac{d_{t_i}}{d_{s_i}d_j}}\vcenter{\hbox{\begin{tikzpicture}
    \begin{scope}
            \VeryThinLine[midarrow] (120:0.5) -- (0,0);
            \VeryThinLine[midarrow] (30:0.5) -- (0,0);
            \node at (120:0.75) {$s_i$};
            \node at (30:0.75) {$j$};
            \draw[fill=black] (0:0) ellipse (0.05 and 0.05);
            \node at (-0.25,-0.2) {$\zeta_i$};
            \VeryThinLine(0,0) -- (0,-0.5);
        \end{scope}
        \begin{scope}[xshift=-1.5cm]
            \VeryThinLine[midarrow]  (0,0) -- (150:0.5);
            \node at (150:0.75) {$j$};
            \VeryThinLine[midarrow] (60:0.5) -- (0,0);
            \node at (60:0.75) {$s_i$};
            \draw[fill=black] (0:0) ellipse (0.05 and 0.05);
            \node at (0.25,-0.2) {$\widehat{\zeta_i}$};
            \VeryThinLine (0,0) -- (0,-0.5);
        \end{scope}
        \VeryThinLine[midarrow] (0,-0.5) arc (0:-180:0.75 and 0.75);
        \node at (-0.75,-1.5) {$t_i$};
    \end{tikzpicture}}}, 
    \end{equation}
    which is obtained by sliding $\zeta^*_i$ close to the morphism $\psi_i$ along the $s_i$ string, and by $\widehat{\zeta}_i = \overline{\zeta}_i$. 
    Now suppose $s_i$ is self-dual and $f_i$ is directed from $v_{i}$ to $v_{i+1}$. 
    When $t_i$ is not self-dual, we have 
    \begin{equation}
        \begin{aligned}
            \vcenter{\hbox{\begin{tikzpicture}
    \VeryThinLine[midarrow] (-0.35,0) -- (-0.35,-1) arc (0:-180:0.25 and 0.25) -- (-0.85,-0.15);
    \VeryThinLine[midarrow] (0.35,0) -- (0.35,-1) arc (0:-180: 0.95 and 0.95) -- +(0,1);
    \node at (-0.35,0.25) {$s_i$};
    \node at (-0.85,0.25) {$s_i$};
    \node at (0.35,0.25) {$j$};
    \VeryThinLine[midarrow] (-0.85,-0.15) -- (-0.85,-0.45);
    \VeryThinLine (-0.85,-0.85) -- (-0.85,-0.55);
    \draw[draw=black, fill=yellow] (-0.85,-0.5) circle (2pt);
    \end{tikzpicture}}} &= \sum_{t_i\in \Irr(\mathcal{C})}\sum_{\zeta_i\in \mathrm{ONB}(s_ij,t_i)}\sqrt{\frac{d_{t_i}}{d_{s_i}d_j}}\vcenter{\hbox{\begin{tikzpicture}
    \VeryThinLine (0,0) -- (150:0.5);
    \VeryThinLine (0,0) -- (30:0.5);
    \node at (150:0.75) {$s_i$};
    \node at (30:0.75) {$j$};
    \draw[fill=black] (0:0) ellipse (0.05 and 0.05);
    \node at (-0.25,-0.2) {$\zeta_i$};
    \VeryThinLine[midarrow] (0,0) -- (0,-1);
    \node at (0.3,-0.5) {$t_i$};
    \begin{scope}[yshift=-1cm]
        \VeryThinLine (0,0) -- (-150:0.5);
        \VeryThinLine (0,0) -- (-30:0.5);
        \VeryThinLine[midarrow] (-150:0.5) arc (0:-180:0.25 and 0.25) -- (-0.93,0.6);
        \VeryThinLine[midarrow] (-0.93,1.5) -- (-0.93,0.7);
        \draw[draw=black, fill=yellow] (-0.93,0.65) circle (2pt);
        \VeryThinLine[midarrow] (-30:0.5) arc (0:-180:1 and 1) -- +(0,1.5);
        \node at (-125:0.75) {$s_i$};
        \node at (-30:0.75) {$j$};
        \draw[fill=black] (0:0) ellipse (0.05 and 0.05);
        \node at (-0.25,0.2) {$\zeta^*_i$};
    \end{scope}
    \end{tikzpicture}}}\\
    &= \sum_{t_i\in \Irr(\mathcal{C})}\sum_{\zeta_i\in \mathrm{ONB}(s_ij,t_i)}\sqrt{\frac{d_{t_i}}{d_{s_i}d_j}} \vcenter{\hbox{\begin{tikzpicture}
        \begin{scope}
            \VeryThinLine[midarrow] (120:0.5) -- (0,0);
            \VeryThinLine[midarrow] (30:0.5) -- (0,0);
            \node at (120:0.75) {$s_i$};
            \node at (30:0.75) {$j$};
            \draw[fill=black] (0:0) ellipse (0.05 and 0.05);
            \node at (-0.25,-0.2) {$\zeta_i$};
            \VeryThinLine(0,0) -- (0,-0.5);
        \end{scope}
        \begin{scope}[xshift=-1.5cm]
            \VeryThinLine[midarrow]  (0,0) -- (150:0.5);
            \node at (150:0.75) {$j$};
            \VeryThinLine[midarrow] (60:0.5) -- (0,0);
            \node at (60:0.75) {$s_i$};
            \draw[fill=black] (0:0) ellipse (0.05 and 0.05);
            \node at (0.25,-0.2) {$\widehat{\zeta_i}$};
            \VeryThinLine (0,0) -- (0,-0.5);
        \end{scope}
        \VeryThinLine[midarrow] (0,-0.5) arc (0:-180:0.75 and 0.75);
        \node at (-0.75,-1.5) {$t_i$};
    \end{tikzpicture}}}. 
        \end{aligned}
    \end{equation}
    When $t_i$ is also self-dual, we rewrite the summands as 
    \begin{equation}
        \vcenter{\hbox{\begin{tikzpicture}
        \begin{scope}
            \VeryThinLine[midarrow] (120:0.5) -- (0,0);
            \VeryThinLine[midarrow] (30:0.5) -- (0,0);
            \node at (120:0.75) {$s_i$};
            \node at (30:0.75) {$j$};
            \draw[fill=black] (0:0) ellipse (0.05 and 0.05);
            \node at (-0.25,-0.2) {$\zeta_i$};
            \VeryThinLine(0,0) -- (0,-0.5);
        \end{scope}
        \begin{scope}[xshift=-1.5cm]
            \VeryThinLine[midarrow]  (0,0) -- (150:0.5);
            \node at (150:0.75) {$j$};
            \VeryThinLine[midarrow] (60:0.5) -- (0,0);
            \node at (60:0.75) {$s_i$};
            \draw[fill=black] (0:0) ellipse (0.05 and 0.05);
            \node at (0.25,-0.1) {$\widehat{\zeta_i}$};
            \VeryThinLine (0,0) -- (0,-0.5);
        \end{scope}
        \VeryThinLine[midarrow] (0,-0.5) arc (0:-180:0.75 and 0.75);
        \node at (-0.75,-1.5) {$t_i$};
        \VeryThinLine[midarrow] (-1.5,-0.3) -- (-1.5,-0.6);
        \draw[draw=black, fill=green] (-1.5,-0.65) circle (2pt);
    \end{tikzpicture}}}
    \end{equation}
    Resolving back into the tensor-product Hilbert space gives the term $\sigma(\vec{t})_i\sigma(\vec{s})_iB_{v_{i+1},p}(-, \zeta_i)\otimes B_{v_i,p}(\zeta_i,-)$ precomposed with $\mathrm{eval}_D$. 

    Now suppose that $s_i$ is self-dual and $f_i$ is oriented from $v_{i+1}$ to $v_{i}$.
    Then 
    \begin{equation}
        \begin{aligned}
            \vcenter{\hbox{\begin{tikzpicture}
    \VeryThinLine[midarrow] (-0.85,0) -- (-0.85,-1) arc (180:360:0.25 and 0.25) -- (-0.35,0);
    \VeryThinLine[midarrow] (0.35,0) -- (0.35,-1) arc (0:-180: 0.95 and 0.95) -- +(0,1);
    \node at (-0.35,0.25) {$s_i$};
    \node at (-0.85,0.25) {$s_i$};
    \node at (0.35,0.25) {$j$};
    \VeryThinLine[midarrow] (-0.35,-0.15) -- (-0.35,-0.45);
    \draw[draw=black, fill=yellow] (-0.35,-0.5) circle (2pt);
    \end{tikzpicture}}} &= \sum_{t_i\in \Irr(\mathcal{C})}\sum_{\zeta_i\in \mathrm{ONB}(s_ij,t_i)}\sqrt{\frac{d_{t_i}}{d_{s_i}d_j}}\vcenter{\hbox{\begin{tikzpicture}
    \VeryThinLine (0,0) -- (150:0.5);
    \VeryThinLine (0,0) -- (30:0.5);
    \node at (150:0.75) {$s_i$};
    \node at (30:0.75) {$j$};
    \draw[fill=black] (0:0) ellipse (0.05 and 0.05);
    \node at (-0.25,-0.2) {$\zeta_i$};
    \VeryThinLine[midarrow] (0,0) -- (0,-1);
    \node at (0.3,-0.5) {$t_i$};
    \begin{scope}[yshift=-1cm]
    \VeryThinLine (0,0) -- (-150:0.5);
    \VeryThinLine (0,0) -- (-30:0.5);
    \draw[fill=black] (0:0) ellipse (0.05 and 0.05);
    \node at (-0.25,0.2) {$\zeta^*_i$};
    \end{scope}
    \begin{scope}[xshift= -0.683cm ,yshift = -1.25cm]
        \VeryThinLine[midarrow] (-0.25,0) -- (-0.25,-1) arc (180:360:0.25 and 0.25) -- (0.25,0);
    \VeryThinLine[midarrow] (1.12,0) -- (1.12,-1) arc (0:-180: 1.12 and 0.75) -- +(0,1);
    \node at (-0.25,0.25) {$s_i$};
    \node at (-1.12,0.25) {$j$};
    \VeryThinLine[midarrow] (0.25,-0.15) -- (0.25,-0.45);
    \draw[draw=black, fill=yellow] (0.25,-0.5) circle (2pt);
    \end{scope}
    \end{tikzpicture}}}.
        \end{aligned}
    \end{equation}
    When $t_i$ is self-dual, inserting a pair of $\phi_{t_i}$ and its inverse gives
    \begin{equation}
        \vcenter{\hbox{\begin{tikzpicture}
    \VeryThinLine (0,0) -- (150:0.5);
    \VeryThinLine (0,0) -- (30:0.5);
    \node at (150:0.75) {$s_i$};
    \node at (30:0.75) {$j$};
    \draw[fill=black] (0:0) ellipse (0.05 and 0.05);
    \node at (-0.25,-0.2) {$\zeta_i$};
    \VeryThinLine[midarrow] (0,0) -- (0,-1);
    \node at (0.3,-0.5) {$t_i$};
    \begin{scope}[yshift=-1cm]
    \VeryThinLine (0,0) -- (-150:0.5);
    \VeryThinLine (0,0) -- (-30:0.5);
    \draw[fill=black] (0:0) ellipse (0.05 and 0.05);
    \node at (-0.25,0.2) {$\zeta^*_i$};
    \end{scope}
    \begin{scope}[xshift= -0.683cm ,yshift = -1.25cm]
        \VeryThinLine[midarrow] (-0.25,0) -- (-0.25,-1) arc (180:360:0.25 and 0.25) -- (0.25,0);
    \VeryThinLine[midarrow] (1.12,0) -- (1.12,-1) arc (0:-180: 1.12 and 0.75) -- +(0,1);
    \node at (-0.25,0.25) {$s_i$};
    \node at (-1.12,0.25) {$j$};
    \VeryThinLine[midarrow] (0.25,-0.15) -- (0.25,-0.45);
    \draw[draw=black, fill=yellow] (0.25,-0.5) circle (2pt);
    \end{scope}
    \end{tikzpicture}}} = \vcenter{\hbox{\begin{tikzpicture}
    \VeryThinLine (0,0) -- (150:0.5);
    \VeryThinLine (0,0) -- (30:0.5);
    \node at (150:0.75) {$s_i$};
    \node at (30:0.75) {$j$};
    \draw[fill=black] (0:0) ellipse (0.05 and 0.05);
    \node at (-0.25,-0.15) {$\zeta_i$};
    \VeryThinLine[midarrow] (0,0) -- (0,-0.65);
    \VeryThinLine[midarrow] (0,-1.3) -- (0,-0.65);
    \VeryThinLine[midarrow] (0,-1.3) -- (0,-1.95);
    \node at (0.3,-0.95) {$t_i$};
    \draw[draw=black, fill=green] (0,-0.65) circle (2pt);
    \draw[draw=black, fill=green] (0,-1.3) circle (2pt);
    \begin{scope}[yshift=-2cm]
    \VeryThinLine (0,0) -- (-150:0.5);
    \VeryThinLine (0,0) -- (-30:0.5);
    \draw[fill=black] (0:0) ellipse (0.05 and 0.05);
    \node at (-0.25,0.15) {$\zeta^*_i$};
    \end{scope}
    \begin{scope}[xshift= -0.683cm ,yshift = -2.25cm]
        \VeryThinLine[midarrow] (-0.25,0) -- (-0.25,-1) arc (180:360:0.25 and 0.25) -- (0.25,0);
    \VeryThinLine[midarrow] (1.12,0) -- (1.12,-1) arc (0:-180: 1.12 and 0.75) -- +(0,1);
    \node at (-0.25,0.25) {$s_i$};
    \node at (-1.12,0.25) {$j$};
    \VeryThinLine[midarrow] (0.25,-0.15) -- (0.25,-0.45);
    \VeryThinLine (0.25,-0.85) -- (0.25,-0.55);
    \draw[draw=black, fill=yellow] (0.25,-0.5) circle (2pt);
    \end{scope}
    \end{tikzpicture}}} = \nu_{t_i}\nu_{s_i}\vcenter{\hbox{\begin{tikzpicture}
        \begin{scope}
            \VeryThinLine[midarrow] (120:0.5) -- (0,0);
            \VeryThinLine[midarrow] (30:0.5) -- (0,0);
            \node at (120:0.75) {$s_i$};
            \node at (30:0.75) {$j$};
            \draw[fill=black] (0:0) ellipse (0.05 and 0.05);
            \node at (-0.25,-0.2) {$\zeta_i$};
            \VeryThinLine(0,0) -- (0,-0.5);
        \end{scope}
        \begin{scope}[xshift=-1.5cm]
            \VeryThinLine[midarrow]  (0,0) -- (150:0.5);
            \node at (150:0.75) {$j$};
            \VeryThinLine[midarrow] (60:0.5) -- (0,0);
            \node at (60:0.75) {$s_i$};
            \draw[fill=black] (0:0) ellipse (0.05 and 0.05);
            \node at (0.25,-0.2) {$\widehat{\zeta_i}$};
            \VeryThinLine (0,0) -- (0,-0.5);
        \end{scope}
        \VeryThinLine[midarrow] (-1.5,-0.5) arc (180:360:0.75 and 0.75) -- (0,-0.7);
        \node at (-0.75,-1.65) {$t_i$};
        \draw[draw=black, fill=green] (0,-0.7) circle (2pt);
    \end{tikzpicture}}}
    \end{equation}
    where in the last equality we used the property of Frobenius-Schur indicator \cite{FuchsGanchev1999FSindicator}:
    \begin{equation}
        \vcenter{\hbox{\begin{tikzpicture}[scale=0.75]
    \VeryThinLine[midarrow] (0,0.75) -- (0,0.05);
  \VeryThinLine (0,0.75) arc (180:0:0.5 and 0.25);
  \VeryThinLine[midarrow] (1,-0.75) -- (1,0.75);
  \node at (1,-1) {$b$};
  \VeryThinLine[midarrow] (0,-0.75) -- (0,-0.05);
  \VeryThinLine (0,-0.75) arc (0:-180:0.5 and 0.25);
  \VeryThinLine[midarrow]  (-1,0.75) --(-1,-0.75);
  \node at (-1,1) {$b$};
  \draw[draw=black, fill=yellow] (0,0) circle (2pt);
  \end{tikzpicture}}} = \nu_b \vcenter{\hbox{\begin{tikzpicture}
    \VeryThinLine[midarrow] (0,-0.75) -- (0,-0.05);
    \VeryThinLine[midarrow] (0,0.75) -- (0,0.05);
    \draw[draw=black, fill=yellow] (0,0) circle (2pt);
  \end{tikzpicture}}}. 
    \end{equation}
    Again, resolving back into the tensor-product Hilbert space gives $\sigma(\vec{t})_i\sigma(\vec{s})_iB_{v_{i+1},p}(-, \zeta_i)\otimes B_{v_i,p}(\zeta_i,-)$ precomposed with $\mathrm{eval}_D$.
    Thus, Equation~\eqref{eqn:: evaluation map and plaquette operator} holds for any orientation of the edges $f_i$. 
\end{proof}

\begin{remark}
    A consequence of Lemma ~\ref{lemma:: eval intertwines partial Bp} is that the plaquette operators mutually commute with each other. 
    By Appendix~B of \cite{LinLevinBurnell2021}, under the present conventions, including the Frobenius--Schur factors, each $B_p$ is a self-adjoint projection and commutes with every compatible edge projection $Q_e$.
\end{remark}

\begin{lemma}\label{lemma:: ONB of 4-point hom space}
    For simple objects $s_1,s_2,t_1,t_2\in \Irr(\mathcal{C})$ and any $g\in \Hom(s_1s_2, t_1 t_2)$, we have 
    \begin{equation}
        g = \sum_{j\in \Irr(\mathcal{C})}\sum_{\xi\in \mathrm{ONB}(s_1 j,t_1)} \sum_{\eta\in \mathrm{ONB}(s^{-}_2 j,t^{-}_2)} \frac{1}{\sqrt{d_{s_1}d_{s_2}d_{t_1}d_{t_2}}}\tr_{\mathcal{C}}\left( f(\xi,\eta)^* g \right) f(\xi,\eta),
    \end{equation}
    where $f(\xi,\eta)$ is defined as in Equation~\eqref{equation:: f(xi,eta)}.
\end{lemma}
\begin{proof}
    The morphisms $f(\xi,\eta)$ span $\Hom(s_1s_2,t_1t_2)$ as $\xi$ and $\eta$ range over $\Hom(s_1 j,t_1)$ and $\Hom(s^{-}_2j,t^{-}_2)$ and $j$ ranges over $\Irr(\mathcal{C})$.
    Orthogonality of the morphisms $f(\xi,\eta)^*$ can be verified using the following local relation: 
    \begin{equation}\label{eqn:: local relation I}
        \vcenter{\hbox{\begin{tikzpicture}
        \coordinate (u2) at (60:1.3);
        \coordinate (u3) at (30:1.5);
        \coordinate (d2) at (-60:1.3);
        \coordinate (d3) at (-30:1.5);
        \draw[fill=black] (u3) ellipse (0.05 and 0.05);
        \draw[fill=black] (d3) ellipse (0.05 and 0.05);
        \VeryThinLine (u2) -- (u3) -- (d3) -- (d2) -- cycle;
        \VeryThinLine[midarrow] (d2) -- (u2);
        \VeryThinLine[midarrow] (u3) -- (d3);
        \VeryThinLine[midarrow] (30:2.25) -- (u3);
        \VeryThinLine[midarrow] (d3) -- (-30:2.25);
        \node at (0.25,0) {$a$};
        \node at (30:2.5) {$b$};
        \node at (-30:2.5) {$b$};
        \node at (1.5,0) {$c$};
        \node[xshift = -6pt, yshift = -2pt] at (u3) {$\xi$};
        \node[xshift = -6pt, yshift = 2pt] at (d3) {$\eta^*$};
    \end{tikzpicture}}} = \sqrt{\frac{d_ad_c}{d_b}}\braket{\eta,\xi} \mathrm{Id}_b,\quad a,b,c\in \Irr(\mathcal{C}). 
    \end{equation}
    Therefore, the morphisms $f(\xi,\eta)$ form a complete orthonormal basis of $\Hom(s_1s_2,t_1t_2)$. 
\end{proof}

\begin{proposition}\label{prop:: eval is an isometry}
    Let $D$ be a simply connected region in $\mathbb{S}^2$ such that $\Gamma\cap D$ is connected, $\partial D$ is transverse to $\Gamma$ and avoids its vertices, and $\partial D$ crosses each edge of $\Gamma$ at most once and crosses exactly $m$ edges.
    Then $\mu^{-\vert F_D\vert}\mathrm{eval}_D$ restricts to an isometry from the ground state subspace of $H_D$ onto $\Hom(\mathbb{1},A^m)$, where $\vert F_D\vert$ is the number of plaquettes in the interior of $D$.
    More precisely,
    \begin{equation}
        \prod_{\substack{e\in E\\ \operatorname{supp}(Q_e)\subset D}} Q_e
        \prod_{p\subset \operatorname{int}(D)} B_p
        = \mu^{-2\vert F_D\vert}\mathrm{eval}_D^{\dagger}\mathrm{eval}_D.
    \end{equation}
\end{proposition}
\begin{proof}
    We prove the statement by induction on $\vert F_D\vert$.
    In the base case, $\vert F_D\vert = 0$, so $\Gamma\cap D$ is a tree.
    The statement then follows from repeated applications of Equation~\eqref{eqn:: local relation I}.
    Suppose the statement holds for every simply connected region with $n-1$ plaquettes in its interior, and take $\vert F_D\vert = n$.
    By transversality, we may push a connected arc of $\partial D$ across a boundary plaquette to obtain a simply connected subregion $D'\subset D$ such that $\Gamma\cap D'$ is connected, $\partial D'$ satisfies the same transversality conditions as $\partial D$, and $D'$ has $n-1$ plaquettes in its interior.
    Let $p$ be the only plaquette in $D$ not contained in $D'$.

    Set $R=D\setminus D'$ and $T=\Gamma\cap R$.
    By construction, $R$ is simply connected, and $T$ is a tree with exactly two leaves on $\partial D'$.
    Label its remaining $N$ leaves by $k_1,\ldots,k_N$ and order the boundary of $R$ counterclockwise so that these $N$ leaves come first and the two leaves on $\partial D'$ come last.
    Denote these two edges by $c_1,c_2$, retain their orientations inherited from $\Gamma$, and order the boundary of $D'$ so that $c_1,c_2$ come first.
    Consider $\ket{\psi_1},\ket{\varphi_1}\in \mathcal{H}_{R}$ and $\ket{\psi_2},\ket{\varphi_2}\in \mathcal{H}_{D'}$ such that
    \begin{equation}
        \begin{aligned}
            \mathrm{eval}_{R}\ket{\psi_1}&\in \Hom(\mathbb{1},k_1\cdots k_Ns^{-}_2s^{-}_1),&
            \mathrm{eval}_{R}\ket{\varphi_1}&\in \Hom(\mathbb{1},k_1\cdots k_Nt^{-}_2t^{-}_1),\\
            \mathrm{eval}_{D'}\ket{\psi_2}&\in \Hom(\mathbb{1},s_1s_2l_1\cdots l_{r-2}),&
            \mathrm{eval}_{D'}\ket{\varphi_2}&\in \Hom(\mathbb{1},t_1t_2l_1\cdots l_{r-2}),
        \end{aligned}
    \end{equation}
    where $s_1,s_2,t_1,t_2,k_1,\ldots,k_N,l_1,\ldots,l_{r-2}\in\Irr(\mathcal{C})$ and $r=\vert\Gamma\cap\partial D'\vert$, so $m=N+r-2$.
    If either vector has mismatched edge labels, then both sides of the desired identity vanish, so we henceforth assume that both vectors lie in the common image of all $Q_e$ supported in $D$.
    Let $\sigma_c(a)$ denote the factor assigned to the edge $c$ with label $a$ in the definition of $B_p$, and set $\varepsilon_{c_1,c_2}(s_1,s_2;t_1,t_2)=\prod_{i=1}^{2}\sigma_{c_i}(s_i)\sigma_{c_i}(t_i)$. 
    Define 
    \begin{equation}
        \overline{f(\xi,\eta)}^{\,c_1,c_2}\coloneqq
        \varepsilon_{c_1,c_2}(s_1,s_2;t_1,t_2)(\iota_{t_2}^{-1}\otimes\iota_{t_1}^{-1})\circ\overline{f(\xi,\eta)}\circ(\iota_{s_2}\otimes\iota_{s_1})
        \in\Hom(s^{-}_2s^{-}_1,t^{-}_2t^{-}_1),
    \end{equation}
    which is the morphism obtained by applying Lemma ~\ref{lemma:: eval intertwines partial Bp} to the two complementary paths in $\partial p$. 
    Applying the induction hypothesis to $D'$, the base case to $R$, for which $\Gamma\cap R=T$ is a tree, and Lemma ~\ref{lemma:: eval intertwines partial Bp} to the two complementary paths in $\partial p$, we obtain
    \begin{equation}
        \begin{aligned}
            &\braket{\varphi_1\otimes\varphi_2\vert
            \prod_{\substack{e\in E\\ \operatorname{supp}(Q_e)\subset D}}Q_e
            \prod_{q\subset\operatorname{int}(D)}B_q
            \vert\psi_1\otimes\psi_2}
            =\braket{\varphi_1\otimes\varphi_2\vert B_p\prod_{q\subset D'}B_q\vert\psi_1\otimes\psi_2}\\
            &=\frac{1}{\mu^{2n}}\sum_{j\in\Irr(\mathcal{C})}\sum_{\xi\in\mathrm{ONB}(s_1j,t_1)}\sum_{\eta\in\mathrm{ONB}(s^{-}_2j,t^{-}_2)}\sqrt{\frac{d_{t_1}d_{t_2}}{d_{s_1}d_{s_2}}}\\
            &\quad\braket{\mathrm{eval}_{R}(\varphi_1),(\mathrm{Id}_{A^N}\otimes\overline{f(\xi,\eta)}^{\,c_1,c_2})\mathrm{eval}_{R}(\psi_1)}
            \braket{\mathrm{eval}_{D'}(\varphi_2),(f(\xi,\eta)\otimes\mathrm{Id}_{A^{r-2}})\mathrm{eval}_{D'}(\psi_2)}.
        \end{aligned}
    \end{equation}
    For each $\xi\in\Hom(s_1j,t_1)$ and $\eta\in\Hom(s_2^-j,t_2^-)$, the inner product in Equation~\eqref{eqn:: skein module inner product} gives
    \begin{equation}
        \begin{aligned}
            &\sqrt{\frac{d_{t_1}d_{t_2}}{d_{s_1}d_{s_2}}}
            \braket{\mathrm{eval}_{R}(\varphi_1),(\mathrm{Id}_{A^N}\otimes\overline{f(\xi,\eta)}^{\,c_1,c_2})\mathrm{eval}_{R}(\psi_1)}
            \braket{\mathrm{eval}_{D'}(\varphi_2),(f(\xi,\eta)\otimes\mathrm{Id}_{A^{r-2}})\mathrm{eval}_{D'}(\psi_2)}\\
            &=\frac{1}{\sqrt{d_{s_1}d_{s_2}d_{t_1}d_{t_2}}}
            \frac{1}{\sqrt{\prod_{a=1}^{N}d_{k_a}\prod_{b=1}^{r-2}d_{l_b}}}
            \tr_{\mathcal{C}}\left(\mathrm{eval}_{R}(\psi_1)\mathrm{eval}_{R}(\varphi_1)^*(\mathrm{Id}_{A^N}\otimes\overline{f(\xi,\eta)}^{\,c_1,c_2})\right)\\
            &\quad\times\tr_{\mathcal{C}}\left(\mathrm{eval}_{D'}(\varphi_2)^*(f(\xi,\eta)\otimes\mathrm{Id}_{A^{r-2}})\mathrm{eval}_{D'}(\psi_2)\right), 
        \end{aligned}
    \end{equation}
    where we used the cyclicity of the trace. 
    The factor $\varepsilon_{c_1,c_2}(s_1,s_2;t_1,t_2)$ supplies Frobenius--Schur indicators that arise when labels of $c_1$ and $c_2$ are self-dual. 
    Summing over $j,\xi,\eta$ and by Lemma ~\ref{lemma:: ONB of 4-point hom space}, we obtain
    \begin{equation}
        \begin{aligned}
            &\frac{1}{\mu^{2n}}\sum_{j\in\Irr(\mathcal{C})}\sum_{\xi\in\mathrm{ONB}(s_1j,t_1)}\sum_{\eta\in\mathrm{ONB}(s_2^-j,t_2^-)}
            \frac{1}{\sqrt{d_{s_1}d_{s_2}d_{t_1}d_{t_2}}}
            \frac{1}{\sqrt{\prod_{a=1}^{N}d_{k_a}\prod_{b=1}^{r-2}d_{l_b}}}\\
            &\quad\times\tr_{\mathcal{C}}\left(\mathrm{eval}_{R}(\varphi_1)^*(\mathrm{Id}_{A^N}\otimes\overline{f(\xi,\eta)}^{\,c_1,c_2})\mathrm{eval}_{R}(\psi_1)\right)\\
            &\quad\times\tr_{\mathcal{C}}\left(\mathrm{eval}_{D'}(\varphi_2)^*(f(\xi,\eta)\otimes\mathrm{Id}_{A^{r-2}})\mathrm{eval}_{D'}(\psi_2)\right)\\
            &=\frac{1}{\mu^{2n}}\frac{1}{\sqrt{\prod_{a=1}^{N}d_{k_a}\prod_{b=1}^{r-2}d_{l_b}}} \tr_{\mathcal{C}}\left(\mathrm{eval}_{D}(\varphi_1\otimes\varphi_2)^*\mathrm{eval}_{D}(\psi_1\otimes\psi_2)\right)\\
            &=\frac{1}{\mu^{2n}}\braket{\mathrm{eval}_{D}(\varphi_1\otimes\varphi_2),\mathrm{eval}_{D}(\psi_1\otimes\psi_2)}.
        \end{aligned}
    \end{equation}
    This proves that $\mu^{-2n}\mathrm{eval}^\dagger_D\mathrm{eval}_D$ equals the product of all $Q_e$ and $B_p$ supported in $D$.
    Since the product on the left-hand side is the ground-state projection of $H_D$, it follows that $\mu^{-n}\mathrm{eval}_D$ restricts to an isometry on the ground-state subspace.
\end{proof}

\begin{corollary}\label{corollary:: ground state on sphere as evaluation}
    Let $\Gamma$ be a finite directed connected graph embedded in $\mathbb{S}^2$. 
    Then the ground state subspace of the Hamiltonian $H_{\Gamma}$ on $\mathbb{S}^2$ is one-dimensional. 
    Moreover, there exists a ground state $\ket{\Psi_{\Gamma}}$ whose wave function is given by 
    \begin{equation}
        \frac{1}{\mu^{\vert F_D\vert}}\mathrm{eval}_{D}\ket{\otimes_{v\in V}\varphi_v} = \braket{\Psi_{\Gamma}|\otimes_{v\in V}\varphi_v}\mathrm{Id}_{\mathbb{1}},
    \end{equation}
    where $D$ is any simply connected region in $\mathbb{S}^2$ that contains all but one plaquette in its interior. 
\end{corollary}
\begin{proof}
    Let $p_{0}$ be the only plaquette in $\mathbb{S}^2$ not contained in $D$, and assume that $p_0$ contains the north pole in its interior.
    Partition the vertices of $\Gamma$ into sets $V_1$ and $V_2$, where $V_1 = \{v_i\}^m_{i=1}$ consists of the vertices on the boundary of $p_0$.
    Under stereographic projection from the north pole, we have
    \begin{equation}
        \mathrm{eval}_D\ket{\otimes_{v\in V}\varphi_v} = \vcenter{\hbox{\begin{tikzpicture}
            \draw (0,0) circle (1.5);
            \VeryThinLine (-0.5,0) -- (120:1.5);
            \VeryThinLine (-0.5,0) -- (-120:1.5);
            \VeryThinLine (0.5,0) -- (60:1.5);
            \VeryThinLine (0.5,0) -- (-60:1.5);
            \node [draw, minimum width=0.5cm, minimum height=0.25cm, fill = white] at (0,0) {$\mathrm{eval}_{D'}(\varphi')$};
            \draw[fill=black] (120:1.5) ellipse (0.05 and 0.05);
            \draw[fill=black] (-120:1.5) ellipse (0.05 and 0.05);
            \draw[fill=black] (60:1.5) ellipse (0.05 and 0.05);
            \draw[fill=black] (-60:1.5) ellipse (0.05 and 0.05);
            \node at (-110:1.25) {$\varphi_{v_1}$};
            \node at (-70:1.25) {$\varphi_{v_2}$};
            \node at (70:1.25) {$\varphi_{v_3}$};
            \node at (110:1.25) {$\varphi_{v_m}$};
            \node at (0,0.85) {$\cdots$};
        \end{tikzpicture}}} 
    \end{equation}
    where $D'$ is the subregion of $D$ that contains all vertices in $V_2$, and $\ket{\varphi'} = \otimes_{v\in V_2}\ket{\varphi_v}$. 
    The intertwining relation in Lemma ~\ref{lemma:: eval intertwines partial Bp}, followed by applications of Equation~\eqref{eqn:: partition of unity}, gives
    \begin{equation}
        \mathrm{eval}_D B_{p_0}\ket{\bigotimes_{v\in V}\varphi_v} = \frac{1}{\mu^2}\sum_{j\in \Irr(\mathcal{C})}d_j\vcenter{\hbox{\begin{tikzpicture}
            \draw (0,0) circle (1.5);
            \VeryThinLine (-0.5,0) -- (120:1.5);
            \VeryThinLine (-0.5,0) -- (-120:1.5);
            \VeryThinLine (0.5,0) -- (60:1.5);
            \VeryThinLine (0.5,0) -- (-60:1.5);
            \node [draw, minimum width=0.5cm, minimum height=0.25cm, fill = white] at (0,0) {$\mathrm{eval}_{D'}(\varphi')$};
            \draw[fill=black] (120:1.5) ellipse (0.05 and 0.05);
            \draw[fill=black] (-120:1.5) ellipse (0.05 and 0.05);
            \draw[fill=black] (60:1.5) ellipse (0.05 and 0.05);
            \draw[fill=black] (-60:1.5) ellipse (0.05 and 0.05);
            \node at (-110:1.25) {$\varphi_{v_1}$};
            \node at (-70:1.25) {$\varphi_{v_2}$};
            \node at (70:1.25) {$\varphi_{v_3}$};
            \node at (110:1.25) {$\varphi_{v_m}$};
            \node at (0,0.85) {$\cdots$};
            \VeryThinLine[midarrow] (0,2) arc (90:-90:2);
            \VeryThinLine (0,-2) arc (-90:-270:2);
            \node at (2.25,0) {$j$};
        \end{tikzpicture}}} = \mathrm{eval}_D\ket{\otimes_{v\in V}\varphi_v}. 
    \end{equation}
    Thus, by Proposition ~\ref{prop:: eval is an isometry}, for any $\ket{\varphi} = \ket{\otimes_{v\in V}\varphi_v}$, we have
    \begin{equation}
        \braket{\varphi|\prod_{p}B_p|\varphi} = \frac{1}{\mu^{2\vert F_D\vert}}\braket{\varphi|\mathrm{eval}^\dagger_D\mathrm{eval}_DB_{p_0}|\varphi} = \frac{1}{\mu^{2\vert F_D\vert}} \braket{\varphi|\mathrm{eval}^\dagger_D\mathrm{eval}_D|\varphi}.
    \end{equation}
    Since $\mathrm{eval}_D$ maps $\mathcal{H}_{\Gamma}$ to the one-dimensional space $\Hom(\mathbb{1},\mathbb{1})$, the ground-state subspace of $H_{\Gamma}$ is one-dimensional and is spanned by the state $\ket{\Psi_{\Gamma}}$ defined in the statement.
    The independence of the choice of $D$ follows from the sphericality of $\mathcal{C}$. 
\end{proof}

\section{\texorpdfstring{$\mathbb{S}^2$}{PDFstring}-functional}\label{sec:: S^2-functional}

In this section, we recall the functional integral approach to TQFT and topological order proposed in \cite{Liu2024} by the first author, then we specialize to the two-dimensional case. 
The general framework of \cite{Liu2024} uses regular stratified piecewise-linear (PL) manifolds to describe higher-dimensional lattice models. 
Although our application is two-dimensional, we will nevertheless state definitions for generic $n$ and retain the PL conditions. 
When specializing to $n=2$, we suppress the PL condition, as $2$-dimensional manifolds have unique PL structures up to PL homeomorphism.

\subsection{Labelled stratifications and \texorpdfstring{$\mathbb{S}^n$}{PDFstring}-functional}

\begin{definition}[Stratified manifold]
    Let $M$ be a closed piecewise-linear $n$-dimensional manifold. 
    A stratification $\cM$ with support $|\cM| = M$, is a filtration of closed subspaces 
    \begin{equation}
        \emptyset = \cM^{-1} \subseteq \cM^0\subseteq \cM^1\subseteq \cdots \subseteq \cM^n = M
    \end{equation}
    such that for any $0\leq k\leq n$, $\cM^k\setminus \cM^{k-1}$ is an open piecewise-linear $k$-dimensional manifold with finitely many connected components, whose closure is $\cM^k$. 
    Components of $\cM^k\setminus \cM^{k-1}$ are called the $k$-stratas of $\cM$. 
    The trivial stratification of $M$ is defined by $M^{k} = \emptyset$ for all $0\leq k\leq n-1$. 
\end{definition}
By a homeomorphism between stratifications $\cM$ and $\cN$, we mean a piecewise-linear homeomorphism $f: M\to N$ such that the restriction of $f$ on each $\cM^k\setminus \cM^{k-1}$ is a homeomorphism onto $\cN^k\setminus \cN^{k-1}$ for all $0\leq k\leq n$. 

When $M$ has a non-empty boundary $\partial M$, a stratification $\cM$ of $M$ is defined similarly by requiring $M^{\circ} \cap (\cM^{k}\setminus \cM^{k-1})$ to be an open piecewise-linear $k$-dimensional manifold with finitely many connected components, whose closure is $M^{\circ} \cap \cM^{k}$ for all $0\leq k\leq n$ (here $M^{\circ} = M\backslash \partial M$ is the interior of $M$).
Moreover, we assume that $\partial M$ intersects each $\cM^k$ transversally. 
A stratification $\cM$ of $M$ induces a stratification $\partial \cM$ of $\partial M$ by taking the transversal intersection: 
\begin{equation}
    (\partial \cM)^{k-1} := \partial M \cap \cM^{k},\quad 0\leq k\leq n.
\end{equation}
We require $\partial \cM$ to be a closed stratified piecewise-linear $(n-1)$-manifold. 


\begin{definition}[One-point suspension]
    Let $\mathcal{S}$ be a stratified piecewise-linear manifold with support $\mathbb{S}^{n-1}$. 
    The one-point suspension $\Lambda \mathcal{S}$ is a stratification of the $n$-dimensional disk $\mathbb{D}^n$, centered at the origin $O_n$, defined as
    \begin{equation}
        (\Lambda \mathcal{S})^0 = \{O_n\}\cup\mathcal{S}^0,\qquad
        (\Lambda \mathcal{S})^k = \Lambda \mathcal{S}^{k-1}\cup \mathcal{S}^k,\quad 1\leq k\leq n-1,\qquad
        (\Lambda \mathcal{S})^n = \mathbb{D}^n.
    \end{equation}
\end{definition}
The one-point suspension $\Lambda \mathcal{S}$ is the simplest stratification of $\mathbb{D}^n$ such that $\partial \Lambda \mathcal{S} = \mathcal{S}$.  
For example, consider the following stratification $\mathcal{S}$ of $\mathbb{S}^1$:
\begin{equation}
    \emptyset \subset \underbrace{\vcenter{\hbox{\begin{tikzpicture}[scale=0.65]
        \draw[dashed] (1,0) arc (0:360:1);
        \draw[fill=black] (0:1) circle (1.5pt);
        \draw[fill=black] (120:1) circle (1.5pt);
        \draw[fill=black] (-120:1) circle (1.5pt);
    \end{tikzpicture}}}}_{\mathcal{S}^0}\subset 
    \underbrace{\vcenter{\hbox{\begin{tikzpicture}[scale=0.65]
        \draw (1,0) arc (0:360:1);
        \draw[fill=black] (0:1) circle (1.5pt);
        \draw[fill=black] (120:1) circle (1.5pt);
        \draw[fill=black] (-120:1) circle (1.5pt);
    \end{tikzpicture}}}}_{\mathcal{S}^1} = \mathbb{S}^1. 
\end{equation}
Then the one-point suspension $\Lambda \mathcal{S}$ is the following stratification of $\mathbb{D}^2$: 
\begin{equation}
    \emptyset\subset \vcenter{\hbox{\begin{tikzpicture}[scale=0.65]
        \draw[dashed] (1,0) arc (0:360:1);
        \draw[fill=black] (0:1) circle (1.5pt);
        \draw[fill=black] (120:1) circle (1.5pt);
        \draw[fill=black] (-120:1) circle (1.5pt);
        \draw[fill=black] (0,0) circle (1.5pt);
    \end{tikzpicture}}}\subset \vcenter{\hbox{\begin{tikzpicture}[scale=0.65]
        \draw (1,0) arc (0:360:1);
        \draw[fill=black] (0:1) circle (1.5pt);
        \draw[fill=black] (120:1) circle (1.5pt);
        \draw[fill=black] (-120:1) circle (1.5pt);
        \draw[fill=black] (0,0) circle (1.5pt);
        \VeryThinLine (0,0) -- (0:1);
        \VeryThinLine (0,0) -- (120:1);
        \VeryThinLine (0,0) -- (-120:1);
    \end{tikzpicture}}} \subset \mathbb{D}^2. 
\end{equation}

\begin{definition}[Regular chart and link boundary]
    Let $\cM$ be a stratification of $M$. 
    For a point $p\in M^{\circ}\cap (\cM^k\setminus \cM^{k-1})$, a \emph{regular chart} is a triple $(U,\phi,\mathcal{S})$, where $U$ is a closed neighborhood of $p$ in $M$, $\mathcal{S}$ is a stratified piecewise-linear manifold with support $\mathbb{S}^{n-k-1}$, and $\phi: \cM\vert_U \to \Lambda\mathcal{S} \times \mathbb{D}^k$ is a piecewise-linear homeomorphism. 
    The triple $(U,\phi,\mathcal{S})$ is subject to the following conditions:
    \begin{enumerate}
        \item $\phi(U\cap \cM^{k}) = O_{n-k}\times \mathbb{D}^k$;
        \item $\phi(p) = O_n$;
        \item for all $0\leq j\leq n$, the image $\phi\left( U\cap (\cM^j\setminus \cM^{j-1}) \right)$ has finitely many connected components;
        \item each component of $\phi\left( U\cap (\cM^j\setminus \cM^{j-1}) \right)$ is contained in a $j$-dimensional subspace of $\mathbb{R}^n$. 
    \end{enumerate}
    We call $\mathcal{S}$ the \emph{link boundary} of $p$. 
    A point $p\in \partial M\cap (\cM^k\setminus \cM^{k-1})$ is regular if it has a neighborhood $U$ such that $\partial M\cap U$ has a regular chart in $\partial\cM$ and $\cM\vert_U$ is piecewise-linearly homeomorphic to $(\partial\cM\vert_{\partial M\cap U})\times [0,1]$. 
\end{definition}

\begin{definition}[Regular stratification]\label{def:: regular stratification}
    A stratified piecewise-linear $0$-manifold is regular. 
    A stratified piecewise-linear $n$-manifold $\cM$ is called \emph{regular} if every interior point has a regular chart whose link boundary is regular and every boundary point is regular in the preceding sense. 
\end{definition}

\begin{remark}
    Regular stratifications of a two-dimensional surface $M$ correspond to embedded graphs in $M$. 
\end{remark}

In what follows all stratifications are regular. 
The link boundary $\mathcal{S}$ of $p$ captures the local lattice shape around $p$. 
For instance, consider a point $p$ sitting on an edge of a lattice in some surface $M$: 
\begin{equation}
    \phi: \vcenter{\hbox{\begin{tikzpicture}
        \draw (0.75,0) arc (0:360:0.75); 
        \draw[blue,thick] (-135:0.75) -- (45:0.75);
        \draw[fill = black] (0,0) circle (1.5pt);
        \node at (0:0.2) {$p$};
        \node at (135:1.25) {$U$};
    \end{tikzpicture}}}\rightarrow \vcenter{\hbox{\begin{tikzpicture}
        \draw (0,0) -- (2,0) -- (2,2) -- (0,2) -- (0,0);
        \draw[blue,thick] (1,0) -- (1,2);
        \draw[fill = black] (1,1) circle (1.5pt);
        \node at (0.45,1) {$\phi(p)$};
        \node at (2.5,1) {$\mathbb{D}^1$};
    \end{tikzpicture}}}
\end{equation}
Then we see that the link boundary $\mathcal{S}$ of $p$ is the trivial (and unique) stratification of $\mathbb{S}^0$. 
Alternatively, if $p$ is a vertex of the lattice with valence $m$, then the link boundary $\mathcal{S}$ of $p$ is the stratification of $\mathbb{S}^1$ with $S^0$ consisting of $m$ points. 

\begin{definition}[Local shape]
    For $0\leq k\leq n$, let $LS_k$ be a set of stratifications of $\mathbb{S}^{n-k-1}$, and set $LS_{\bullet} = \bigcup_{k=0}^n LS_k$. 
    We say a stratification $\cM$ has \emph{local shape} $LS_{\bullet}$ if for any interior point $p\in \cM^k\setminus \cM^{k-1}$, there is a regular chart $(U,\phi,\mathcal{S})$ of $p$ such that $\mathcal{S}\in LS_k$. 
    We call $LS_{\bullet}$ the \emph{local shape set}. 
\end{definition}

For $n=2$, a local shape set $LS_{\bullet}$ is determined by $LS_0$. 
This is because the stratifications of $\mathbb{S}^{2-1-1} = \mathbb{S}^0$ and $\mathbb{S}^{2-2-1} = \mathbb{S}^{-1} = \emptyset$ are unique. 
    Stratifications of $\mathbb{S}^1$ are indexed by non-negative integers, which count the number of connected components of the $0$-dimensional skeleton. 
Then for a stratification $\cM$, a point in $M^0$ with local shape $\mathbb{S}^1$ with $m$ points corresponds to a vertex of valence $m$. 

\begin{definition}[Label space]\label{def:: label space}
    Given a local shape set $LS_{\bullet}$, a label space $\mathcal{H}_{\bullet}=\{\mathcal{H}_{\mathcal{S}}\}_{\mathcal{S}\in LS_{\bullet}}$ is a family of finite-dimensional vector spaces over $\mathbb{C}$. 
    Moreover, for each $\mathcal{S}\in LS_{\bullet}$, we have a representation $\pi_{\mathcal{S}}$ of the orientation-preserving mapping class group $\mathrm{MGS}(\Lambda\mathcal{S})$ on $\mathcal{H}_{\mathcal{S}}$. 
\end{definition}

From now on, we take all label spaces to be Hilbert spaces, and the representations $\pi_{\mathcal{S}}$ of mapping class groups to be unitary representations. 

\begin{definition}[Marked stratification]
   A stratification $\cM$ is called marked by a local shape set $LS_{\bullet}$ if there is a finite set of marked points $P\subset M^{\circ}$ such that 
   \begin{enumerate}
    \item every $p\in P$ is contained in some $k$-strata for $0\leq k\leq n$;
    \item every $p\in P$ has a regular chart $(U_p,\phi_p,\mathcal{S}_p)$, with $\mathcal{S}_p\in LS_{\bullet}$.
   \end{enumerate} 
   The neighborhoods $\{U_p\}_{p\in P}$ are required to be pairwise disjoint and disjoint from $\partial M$, and each $\phi_p$ is orientation-preserving. 
\end{definition}

Let $M$ be a manifold, and $\cM$ be a marked stratification of $M$ with marked points $P$. 
For $p\in P$ with a local chart $(U_p,\phi_p,\mathcal{S}_p)$, define the local Hilbert space $\mathcal{H}_{(U_p,\phi_p,\mathcal{S}_p)}$ to be $\mathcal{H}_{\mathcal{S}_p}$. 
A simple tensor $\otimes_{p\in P}\ket{\psi_p}$ then uniquely determines a labelled stratification $\mathcal{M}[\otimes_{p\in P}\ket{\psi_p}]$ with base $M$. 
We therefore define the local Hilbert space of $\cM$ to be 
\begin{equation}
    \mathcal{H}_{\cM} = \bigotimes_{p\in P} \mathcal{H}_{(U_p,\phi_p,\mathcal{S}_p)}. 
\end{equation}
We require that the simple tensors transform covariantly when we re-parameterize the local charts. 
Let $g: \Lambda \mathcal{S}_{p_0}\rightarrow \Lambda\mathcal{S}_{p_0}$ be the one-point suspension of an orientation-preserving homeomorphism of $\mathcal{S}_{p_0}$, and let $\cM'$ be the marked stratification obtained from $\cM$ by replacing the chart $(U_{p_0},\phi_{p_0},\mathcal{S}_{p_0})$ by $(U_{p_0},g\circ \phi_{p_0},\mathcal{S}_{p_0})$.
Then we have
\begin{equation}
    \begin{aligned}
        \cM'[\pi_{\mathcal{S}_{p_0}}(g)\ket{\psi_{p_0}}\otimes_{p\in P,p\neq p_0}\ket{\psi_p}] = \cM[\ket{\psi_{p_0}}\otimes_{p\in P,p\neq p_0}\ket{\psi_p}],\quad \otimes_{p\in P}\ket{\psi_p}\in \mathcal{H}_{\cM}. 
    \end{aligned}
\end{equation}

\begin{definition}[Configuration space]
    Given a manifold $M$, a regular stratification $\mathcal{S}$ of $\partial M$, a local shape set $LS_{\bullet}$, and a label space $\mathcal{H}_{\bullet}$ indexed by $LS_{\bullet}$, the configuration space on $M$ is defined as: 
    \begin{equation}
        M_{\mathcal{S}}(\mathcal{H}_{\bullet}) = \bigoplus_{\cM,\partial \cM = \mathcal{S}} \mathcal{H}_{\cM} 
    \end{equation}
    where the direct sum is taken over all stratifications $\cM$ of $M$ with chosen local charts $(U_p,\phi_p,\mathcal{S}_p)$ for every marked point $p$. 
    When $M$ is closed, we write $M(\mathcal{H}_{\bullet})$ for short. 
\end{definition}
Elements in $M_{\mathcal{S}}(\mathcal{H}_{\bullet})$ are formal linear sums of labelled stratifications of $M$ with boundary $\mathcal{S}$. 
It describes the space of configurations on all lattices in $M$ with given local shapes, that are compatible with the boundary stratification $\mathcal{S}$. 

The theory developed in \cite{Liu2024} is based on a sphere function defined as follows. 

\begin{definition}[\texorpdfstring{$\mathbb{S}^n$}{}-functional]\label{def:: S^n-functional}
    Consider label spaces $\mathcal{H}_{\bullet} = \{\mathcal{H}_{\mathcal{S}}\}_{\mathcal{S}\in LS_{\bullet}}$ indexed by a local shape set $LS_{\bullet}$. 
    A $\mathbb{S}^n$-functional is a non-zero linear functional $Z: \mathbb{S}^n(\mathcal{H}_{\bullet})\to \mathbb{C}$. 
\end{definition}

\begin{definition}\label{def:: homeomorphism between marked stratifications}
    Consider a marked stratification $\cM$ with local charts $\{(U_p,\phi_p,\mathcal{S}_p)\}_{p\in P}$, and $\cN$ with local charts $\{(U_q,\phi_q,\mathcal{S}_q)\}_{q\in Q}$. 
    A map $f:\cM\rightarrow \cN$ is called a homeomorphism between marked stratifications, if $f$ is a homeomorphism between stratifications, $f(P) = Q$, and for each $p\in P$, we have 
    \begin{equation}
        \phi_{f(p)} = \phi_p\circ f^{-1}. 
    \end{equation}
    In this case, we will write $\cN = f(\cM)$. 
\end{definition}

\begin{definition}[Homeomorphism invariance]\label{def:: HI}
    A $\mathbb{S}^n$-functional $Z$ is \emph{homeomorphism invariant}, if for any orientation-preserving homeomorphism $f:\cM\rightarrow \cN$ between marked stratifications, we have 
    \begin{equation}
        Z(\cN[\cdot ]) = Z(\cM[\cdot])
    \end{equation}
    as functionals on $\mathcal{H}_{\cM}$. 
\end{definition}

Since $\mathbb{S}^2(\mathcal{H}_{\bullet})$ is the direct sum of these local Hilbert spaces over all marked stratifications $\cM$ of $\mathbb{S}^2$, a $\mathbb{S}^2$-functional $Z$ restricts, for each $\cM$, to a linear functional on $\mathcal{H}_{\cM}$. 
By Riesz representation, this restriction is represented by a unique vector $\ket{\tilde{\Psi}_{\cM}}\in\mathcal{H}_{\cM}$, characterized by
\begin{equation}\label{eq:: wave function from functional}
    \braket{\tilde{\Psi}_{\cM}|\otimes_{p}\psi_p} = Z(\cM[\otimes_{p}\psi_p]),\quad \otimes_{p}\psi_p\in \mathcal{H}_{\cM}. 
\end{equation}
Therefore, a single $\mathbb{S}^2$-functional determines the wave functions of a family of states $\{\ket{\tilde{\Psi}_\cM}\}_{\cM}$, parameterized by stratifications of $\mathbb{S}^2$. 
Note that these states are generally not normalized. 

\subsection{\texorpdfstring{$\mathbb{S}^2$}{PDFstring}-functional from string-net ground states}\label{subsection:: S^2-functional from LW model}

Our primary examples of $\mathbb{S}^2$-functionals come from the ground state wave functions of string-net models. 
Fixing an input UFC $\mathcal{C}$, we describe how to obtain a $\mathbb{S}^2$-functional $Z_{\mathcal{C}}$ from the Levin--Wen wave functions defined on lattices in $\mathbb{S}^2$. 

For simplicity we focus on the case of trivalent lattices. 
This means the local shape set $LS_{\bullet}$ is given by $LS_0 = \{\mathcal{S}(3)\}$.
For every $m\geq 0$, let $\mathcal{S}(m)$ denote the standard stratification of $\mathbb{S}^1$ with $m$ marked points.
For $m\geq 4$, it is represented schematically by
\begin{equation}
    \mathcal{S}(m) = \vcenter{\hbox{\begin{tikzpicture}
        \draw (1,0) arc (0:360:1);
        \draw[fill=black] (-135:1) circle (1.5pt);
        \draw[fill=black] (-45:1) circle (1.5pt);
        \draw[fill=black] (45:1) circle (1.5pt);
        \draw[fill=black] (135:1) circle (1.5pt);
        \node at (-135:1.5) {$1$};
        \node at (-45:1.5) {$2$};
        \node at (45:1.5) {$m-1$};
        \node at (135:1.5) {$m$};
        \node at (1.25,0) {$\vdots$};
    \end{tikzpicture}}}
\end{equation} 
For a vertex $p$ with regular chart $(U_p,\phi_p,\mathcal{S}(3))$, the restriction of $\phi_p$ identifies the edges incident to $p$ with these numbered $0$-cells and thereby induces the ordering used in Section ~\ref{section:: String-net wave function}. 
The label space $\mathcal{H}_{\mathcal{S}(3)}$ is  $\Hom(\mathbb{1},A^3)$. 
Edge orientations are auxiliary background data, not physical degrees of freedom.
For a stratification $\cM$ and a choice $o$ of edge orientations, let $\Gamma$ be the directed embedded graph whose vertices and edges are the components of $\cM^0$ and $\cM^1\setminus\cM^0$, respectively.
The regular charts at the vertices identify
\begin{equation}
    \mathcal{H}_{\cM}\cong\bigotimes_{p\in\cM^0}\Hom(\mathbb{1},A^{3})=\mathcal{H}_{\Gamma},
\end{equation}
where the ordering of the factors in $A^{3}$ is induced by the numbered $0$-cells of $\mathcal{S}(3)$.

Given a labelling $\displaystyle \psi=\bigotimes_{p\in \cM^0} \ket{\psi_p}$ in $\mathcal{H}_{\cM}$, where $\cM$ is a stratification of $\mathbb{S}^2$ with $\cM^1$ connected, we define its value under $Z_{\mathcal{C}}$ as follows. 
Choose a disk $D$ in $\mathbb{S}^2$ such that $\cM^1\subset D$ and $\cM^0\cap \partial D = \emptyset$. 
For each $p\in \cM^0$ with regular chart $(U_p,\phi_p,\mathcal{S}(3))$, we fill in a disk $U_p$ around $p$ and evaluate the label $\psi_p$ using the ordering induced by $\phi_p$. 
Together with the specified edge orientations, these vertex labels define a morphism $\mathrm{eval}_D(\otimes_{p\in \cM^0}\psi_p)\in\End(\mathbb{1})$ as in \ref{def:: Levin-Wen evaluation map}. 
Since $\End_{\mathcal{C}}(\mathbb{1})\cong\mathbb{C}$, there is a unique scalar $Z_{\mathcal{C}}(\psi)$ such that

\begin{equation}\label{def:: S^2 functional from Levin-Wen wave function}
    Z_{\mathcal{C}}(\psi)\mathrm{Id}_{\mathbb{1}} = \mathrm{eval}_D\left(\bigotimes_{p\in \cM^0}\psi_p\right).  
\end{equation}

This is the unnormalized categorical evaluation: no face-dependent prefactor is included in the definition.
In particular, the empty diagram evaluates to $\mathrm{Id}_{\mathbb{1}}$.
Changing $o$ only changes the local identifications used to interpret the vertex morphisms and leaves this functional unchanged.
The functional $Z_{\mathcal{C}}$ is independent of the choice of the disk $D$ and is homeomorphism invariant by the sphericality of $\mathcal{C}$.
Moreover, it follows from Corollary ~\ref{corollary:: ground state on sphere as evaluation} that 
\begin{equation}
    Z_{\mathcal{C}}(\psi) = \mu^{F-1}\braket{\Psi_{\Gamma}|\otimes_{p\in \cM^0} \psi_p},
\end{equation}
where $F$ is the number of connected components of $\cM^2\setminus \cM^1$, namely the faces of the embedded graph. 
Equivalently, under the identification $\mathcal{H}_{\cM}\cong\mathcal{H}_{\Gamma}$, one has the equality of linear functionals
\begin{equation}
    Z_{\mathcal{C}}=\mu^{F-1}\bra{\Psi_{\Gamma}}.
\end{equation}
For disconnected $\cM^1$, define $Z_{\mathcal{C}}$ by the same categorical evaluation.
Equivalently, because the tensor unit $\mathbb{1}$ of $\mathcal{C}$ is simple, this evaluation factors multiplicatively over the connected components of $\cM^1$.


\section{Axioms of Levin--Wen Wave Function}\label{sec:: axioms}

In this section, we propose axioms characterizing $\mathbb{S}^2$-functionals $Z$ whose corresponding wave functions $\tilde{\Psi}_{\cM}$ are ground-state wave functions of the Levin--Wen model associated with some UFC $\mathcal{C}$ whenever $\cM^1$ is connected.
In the rest of this paper, we assume $LS_0=\{\mathcal{S}(3)\}$. 
Thus for any stratification $\cM$, the subspace $\cM^1$ is a trivalent graph, and labels are only attached to the vertices of that graph. 
We denote the Hilbert space label at $p$ by $\mathcal{H}_p$. 
The directions of the edges are treated as auxiliary data: they do not count as degrees of freedom in the local Hilbert spaces. 

\subsection{Axioms of the \texorpdfstring{$\mathbb{S}^2$}{}-functional}

Define $\mathbb{D}^2_+$ to be the subset of $\mathbb{S}^2$ with non-negative $x$-coordinate,  and $\mathbb{D}^2_-$ to be the subset with non-positive $x$-coordinate. 
Define $\theta(x,y,z) = (-x,y,z)$. 
Then $\theta$ induces orientation-reversing homeomorphisms from stratifications of $\mathbb{D}^2_+$ to those of $\mathbb{D}^2_-$, fixing their boundaries. 

Let $\mathcal{D}$ be a stratification of $\mathbb{D}^2_+$. 
For each $p\in \mathcal{D}^0$, let 
\begin{equation}
    \hat{\theta}_p: \mathcal{H}_p\rightarrow \mathcal{H}_{\theta(p)}
\end{equation}
be an anti-unitary such that $\hat{\theta}_{\theta(p)}\circ \hat{\theta}_p$ and $\hat{\theta}_{p}\circ \hat{\theta}_{\theta(p)}$ are the identity maps on $\mathcal{H}_p$ and $\mathcal{H}_{\theta(p)}$, respectively. 
We obtain an anti-unitary map $\hat{\theta}: \mathcal{H}_{\mathcal{D}} \to \mathcal{H}_{\theta(\mathcal{D})}$ defined by 

\begin{equation}
    \hat{\theta}\left(\bigotimes_{p\in\mathcal{D}^0}\ket{\psi_p}\right) = \bigotimes_{p\in \mathcal{D}^0} \hat{\theta}_p \ket{\psi_p}. 
\end{equation}

Given stratifications $\mathcal{D}_1,\mathcal{D}_2$ of $\mathbb{D}^2_+$ with $\partial\mathcal{D}_1 = \partial\mathcal{D}_2$ and vectors $\ket{\psi_1}\in \mathcal{H}_{\mathcal{D}_1}$ and $\ket{\psi_2}\in \mathcal{H}_{\mathcal{D}_2}$, gluing $\theta(\mathcal{D}_2)$ and $\mathcal{D}_1$ along their common boundary produces an element
\begin{equation}
    \hat{\theta}(\psi_2)\otimes \psi_1\in \mathcal{H}_{\theta(\mathcal{D}_2)\cup \mathcal{D}_1} = \mathcal{H}_{\theta(\mathcal{D}_2)}\otimes \mathcal{H}_{\mathcal{D}_1}.
\end{equation}
From now on, we shall omit the stratification map $\cM[\cdot]$ when there is no confusion. 

\begin{definition}[Superposed Reflection Positivity]\label{def:: reflection positivity}
    A $\mathbb{S}^2$-functional $Z$ has \emph{superposed (s-) reflection positivity} with respect to $\hat{\theta}$, if for every $r\geq0$ and every $\xi\in\mathbb{D}^2_{\mathcal{S}(r)}(\mathcal{H}_{\bullet})$, one has
    \begin{equation}
        Z(\hat{\theta}(\xi)\otimes \xi)\geq 0. 
    \end{equation}
\end{definition}

\begin{remark}
    Superposed reflection positivity was called ``reflection positivity'' in \cite{Liu2024}. 
    The terminology here is to emphasize that the positivity condition holds not only for configurations on reflection-symmetric lattices, but also for superpositions of them. 
\end{remark}

From a $\mathbb{S}^2$-functional $Z$ with homeomorphism invariance and s-reflection positivity, we can construct a $\mathbb{S}^2$-algebra, or equivalently a spherical $C^*$ tensor category $\mathcal{C}$, as outlined in \cite{Liu2024}. 
In $2$ dimensions, the construction closely parallels the idempotent completion of a planar algebra \cite{PlanarAlgebra2021}. 
For the stratification $\mathcal{S}(m)$, a vector $\ket{\xi}\in \mathbb{D}^2_{\mathcal{S}(m)}(\mathcal{H}_{\bullet})$ is called a \emph{null vector} if 

\begin{equation}
    Z(\hat{\theta}(\xi)\otimes \eta) = 0,\quad \forall \ket{\eta}\in \mathbb{D}^2_{\mathcal{S}(m)}(\mathcal{H}_{\bullet}).
\end{equation}

Note that since $Z$ has s-reflection positivity, this is equivalent to $Z(\hat{\theta}(\xi)\otimes \xi) = 0$. 
Define $\ker_m(Z)$ to be the kernel of the quadratic form $\xi\mapsto Z(\hat{\theta}(\xi)\otimes \xi)$ on $\mathbb{D}^2_{\mathcal{S}(m)}(\mathcal{H}_{\bullet})$, and set $\tilde{V}_m = \mathbb{D}^2_{\mathcal{S}(m)}(\mathcal{H}_{\bullet})/ \ker_m(Z)$.
For a configuration $\psi\in \mathbb{D}^2_{\mathcal{S}(m)}(\mathcal{H}_{\bullet})$, we denote by $[\psi]$ its image under the quotient map. 
We say that $Z$ is \emph{locally finite} if $\tilde{V}_m$ is finite dimensional for all $m\geq 0$. 

\begin{definition}\label{def:: multiplicativity}
    A $\mathbb{S}^2$-functional $Z$ is called \emph{multiplicative} if, for every labelled stratification $\cM$ of $\mathbb{S}^2$ and every closed disk $D\subset\mathbb{S}^2$ satisfying $\partial D\cap\cM^1=\varnothing$, we have 
    \begin{equation}
        Z(\cM) = Z(\cM|_{D^c})\cdot Z(\hat{\cM|_{D^\circ}}),
    \end{equation}
    where $\hat{\cdot}$ represents the one-point compactification of a stratification of an open disk to the one of $\mathbb{S}^2$. 
    The condition on $\partial D$ ensures that both restrictions have empty boundary stratification.
\end{definition}

For a multiplicative, homeomorphism-invariant $\mathbb{S}^n$-functional $Z$, choose a labelled stratification $\cM$ with $Z(\cM)\neq 0$ and cut an unstratified top-dimensional cell by a disk.
One capped piece is homeomorphic to $\cM$, while the other is $\mathbb{S}^n_{\emptyset}$, so homeomorphism invariance and multiplicativity give $Z(\cM)=Z(\cM)Z(\mathbb{S}^n_{\emptyset})$.
Thus
\begin{equation*}
    Z(\mathbb{S}^2_{\emptyset}) = 1. 
\end{equation*}
In particular, a multiplicative, homeomorphism-invariant, s-reflection-positive $\mathbb{S}^2$-functional satisfies $\tilde{V}_0 \cong \mathbb{C}$.
Conversely, if $\tilde{V}_0 \cong \mathbb{C}$ and $Z(\mathbb{S}^2_{\emptyset})=1$, then $Z$ is multiplicative.
Suppose first that $Z$ is multiplicative.
For a labelled stratification $\xi$ of a disk with empty boundary, set $\lambda_{\xi}=Z(\hat\xi)$, where $\hat\xi$ is its one-point compactification.
Multiplicativity implies that, for every labelled stratification $\eta$ of a disk with empty boundary,
\begin{equation}
    \braket{[\eta],[\xi]-\lambda_{\xi}[\mathbb D^2]}
    =Z(\hat{\theta}(\eta)\otimes\xi)-\lambda_{\xi}Z(\hat{\theta}(\eta)\otimes\mathbb D^2)
    =0.
\end{equation}
So $[\xi]=\lambda_{\xi}[\mathbb D^2]$, and $Z(\mathbb{S}^2_{\emptyset})=1$ implies that $[\mathbb D^2]\neq 0$.
Hence $\tilde V_0\cong\mathbb C$.

Conversely, suppose that $\tilde V_0\cong\mathbb C$ and $Z(\mathbb{S}^2_{\emptyset})=1$.
Every connected component $\cM_i$ of $\cM^1$ is contained in a disk region $R_i\subset \mathbb{S}^2$.
Since $\tilde{V}_0$ is one-dimensional, we have $[\cM_i] = \lambda_i [\mathbb{D}^2]$.
The empty-sphere condition gives $\lambda_i = Z(\hat{\cM}_i)$, and replacing the components successively gives
\begin{equation}
    Z(\cM)=\prod_i Z(\hat{\cM}_i),
\end{equation}
showing $Z$ is multiplicative.

For the rest of the paper, all $\mathbb{S}^2$-functionals $Z$ are assumed to be homeomorphism invariant, multiplicative, and s-reflection positive. 
When $Z$ is also locally finite, the construction of \cite{Liu2024} produces a $C^*$-tensor category $\mathcal{C}(Z)$. 
We recall this construction and fix the notation in Appendix~\ref{app:: reconstruction from S2 functional}. 

\begin{definition}[Complete Finiteness]\label{def:: complete finiteness}
    We say that a $\mathbb{S}^2$-functional $Z$ is \emph{completely finite} if $Z$ is locally finite and the resulting $C^*$-tensor category $\mathcal{C}(Z)$ has finitely many isomorphism classes of simple objects. 
\end{definition}

With complete finiteness, the reconstructed category is a unitary spherical fusion category. 
In summary, a multiplicative $\mathbb{S}^2$-functional $Z$ with homeomorphism invariance, s-reflection positivity, and complete finiteness gives rise to a UFC $\mathcal{C}(Z)$ together with a distinguished object $A$ in $\mathcal{C}(Z)$ satisfying 
\begin{equation}
    \End_{\mathcal{C}(Z)}(A^m) \cong \mathcal{A}_m,\quad m\geq 0. 
\end{equation} 

\subsection{Axioms of the wave functions}

In this section, we introduce additional conditions that together with homeomorphism invariance, s-reflection positivity, and multiplicativity, characterize the Levin--Wen wave functions. 
To see why we need these additional constraints, note that a spherical unitary fusion category $\mathcal{C}$ with a distinguished object $A$ determines a $\mathbb{S}^2$-functional $Z_{\mathcal{C},A}$ on stratifications whose $1$-strata are embedded trivalent graphs. 
Take the local shape to be $\mathcal{S}(3)$ and the local label space to be $\Hom(\mathbb{1},A^3)$; for each such stratification $\cM$, the evaluation map $\mathrm{eval}_D$ defines $Z_{\mathcal{C},A}$ on $\mathbb{S}^2(\mathcal{H}_{\bullet})$.
However, this construction does not determine the states $\ket{\tilde{\Psi}_{\cM}}$, since the inner products on the local label spaces have not been specified. 

For a marked stratification $\cM$ of $\mathbb{S}^2$ such that $\cM^1$ intersects the boundary of $\mathbb{D}^2_+$ transversally, we write
\begin{equation}\label{eqn:: unnormalized reduced density matrix}
    \widetilde{\rho}_{\mathbb{D}^2_+}(\cM)
    := \Tr_{\mathbb{D}^2_-}\ket{\tilde{\Psi}_{\cM}}\bra{\tilde{\Psi}_{\cM}}
\end{equation}
for the unnormalized reduced density matrix on $\mathbb{D}^2_+$.

\begin{definition}[Local non-degeneracy]\label{def:: local non degeneracy}
    Let $\cM(3) = \Lambda \mathcal{S}(3)_- \cup\Lambda \mathcal{S}(3)_+$ be the stratification of $\mathbb{S}^2$.
    We say that $Z$ is \emph{locally non-degenerate} if $\widetilde{\rho}_{\mathbb{D}^2_+}(\cM(3))$ is strictly positive. 
    Equivalently, this means that the quotient map $\ket{\psi} \mapsto [\psi]$ is injective on $\mathcal{H}_{\Lambda\mathcal{S}(3)}$. 
\end{definition}
\begin{remark}
    From the perspective of wave function renormalization \cite{QImeetsQM2019}, the $0$-strata of a stratification are supersites in a renormalized lattice. 
    If the reduced density matrix on such a supersite has a nontrivial kernel, then a one-site projection can be used to remove the unentangled degrees of freedom. 
\end{remark}

\begin{definition}[Topological connectedness]\label{def:: topological connectedness}
    We say that a $\mathbb{S}^2$-functional $Z$ satisfies \emph{topological connectedness} if there exists $\tau>0 $ such that, for any stratifications $\cN$ and $\cM$ of $\mathbb{D}^2_+$ sharing the same boundary stratification with $\cN^1$ connected, the following two equalities hold: 
    \begin{itemize}
        \item[TC1:]
        \begin{equation}\label{eqn:: TC1}
        \begin{aligned}
            \widetilde{\rho}_{\mathbb{D}^2_+}(\theta(\cN)\cup \cM)
            = \tau^{\lvert F_\cN \rvert-\lvert F_\cM \rvert}\,
            \widetilde{\rho}_{\mathbb{D}^2_+}(\theta(\cM)\cup \cM).
        \end{aligned}
    \end{equation}
        \item[TC2:] Let $\Pi\in\mathcal{B}(\mathcal{H}_{\cN})$ be the range projection of $\widetilde{\rho}_{\mathbb{D}^2_+}(\theta(\cM)\cup \cN)$, then we have
    \begin{equation}\label{eqn:: TC2}
        \widetilde{\rho}_{\mathbb{D}^2_+}(\theta(\cM)\cup \cN)
            = \tau^{\lvert F_\cM \rvert-\lvert F_\cN \rvert}\,
            \Pi\widetilde{\rho}_{\mathbb{D}^2_+}(\theta(\cN)\cup \cN) \Pi.
    \end{equation}
    \end{itemize}
    Here $\lvert F_\cN \rvert$ and $\lvert F_\cM \rvert$ are the numbers of plaquettes contained in the interiors of $\cN$ and $\cM$, respectively; for a disconnected stratification, the face count is the sum of the corresponding interior plaquette counts of its connected components.
\end{definition}

\begin{definition}[Commutativity]\label{def:: commutativity}
    We say that the $\mathbb{S}^2$-functional $Z$ satisfies \emph{commutativity} if for any labelled stratifications $\ket{\psi_1},\ket{\psi_2},\ket{\psi_3}$ with boundary $\mathcal{S}(1,1)$,
    \begin{equation}
        Z\left( \vcenter{\hbox{\begin{tikzpicture}
            \VeryThinLine[blue] (0,0) circle (0.75);
            \node[circle, draw, fill=white,minimum size=0.75cm, inner sep=1mm] at (30:0.75) {$\psi_1$};
            \node[circle, draw, fill=white, minimum size=0.75cm, inner sep=1mm] at (150:0.75) {$\psi_2$};
            \node[circle, draw, fill=white, minimum size=0.75cm, inner sep=1mm] at (-90:0.75) {$\psi_3$};
        \end{tikzpicture}}} \right) = Z\left( \vcenter{\hbox{\begin{tikzpicture}
            \VeryThinLine[blue] (0,0) circle (0.75);
            \node[circle, draw, fill=white, minimum size=0.75cm, inner sep=1mm] at (30:0.75) {$\psi_2$};
            \node[circle, draw, fill=white, minimum size=0.75cm, inner sep=1mm] at (150:0.75) {$\psi_1$};
            \node[circle, draw, fill=white, minimum size=0.75cm, inner sep=1mm] at (-90:0.75) {$\psi_3$};
        \end{tikzpicture}}} \right). 
    \end{equation}
\end{definition}

\subsection{Verification of the axioms for the Levin--Wen wave function}

In this section, we verify that the $\mathbb{S}^2$-functionals constructed in Section ~\ref{subsection:: S^2-functional from LW model} do satisfy the axioms proposed in Section ~\ref{sec:: axioms}. 
To begin with, note that the $\mathbb{S}^2$-functional $Z_{\mathcal{C}}$ is multiplicative by construction and satisfies local non-degeneracy~\ref{def:: local non degeneracy} by definition of the local Hilbert space. 
Commutativity~\ref{def:: commutativity} follows from the fact that the distinguished object reconstructed from $Z_{\mathcal{C}}$ is $A \cong \bigoplus_{x\in \Irr(\mathcal{C})}x$. 

Thus we only need to prove s-reflection positivity ~\ref{def:: reflection positivity} and the topological connectedness~\ref{def:: topological connectedness}. 

The reflection operator $\hat{\theta}_v$ depends on the directions of edges incident to $v$. 
More specifically, the reflection is defined as follows: attach $\phi_c$ \emph{after} taking the modular conjugation if the edge labelled by a self-dual object $c$ is outgoing from the vertex, and \emph{before} if the edge is incoming. 
For instance, consider a trivalent vertex $v$ of the form $\vcenter{\hbox{\begin{tikzpicture}
    \draw[fill=black] (0,0) ellipse (0.06 and 0.06);
    \VeryThinLine[midarrow] (0,0) -- (-90:0.5);
    \VeryThinLine (0,0) -- (30:0.5);
    \VeryThinLine (0,0) -- (150:0.5);
    \node at (-0.25,-0.20) {$v$};
    \node at (0.25,-0.25) {$\mathdollar$};
\end{tikzpicture}}}$ and $\psi_v\in \Hom(\mathbb{1},abc)$ with $\overline{c}\cong c$, then 
\begin{equation}
    \hat{\theta}\ket{\psi_v} = \vcenter{\hbox{\begin{tikzpicture}
        \draw[fill=black] (0,0.5) ellipse (0.06 and 0.06);
        \VeryThinLine[midarrow] (-0.5,0) -- (0,0.5);
        \VeryThinLine[midarrow] (-0.5,0) -- (-0.5,-0.5);
        \VeryThinLine[midarrow] (0,0) -- (0,0.5);
        \VeryThinLine[midarrow] (0.5,0) -- (0,0.5);
        \draw[draw=black, fill=yellow] (-0.5,0) circle (2pt);
        \node at (0.25,1) {$\overline{\psi_v}$};
        \node at (-0.5,-0.75) {$c$};
        \node at (0,-0.25) {$b$};
        \node at (0.5,-0.25) {$a$};
    \end{tikzpicture}}}
\end{equation}
If the edge labelled by $c$ is incoming, we will have:
\begin{equation}
    \hat{\theta}\ket{\psi_v} = \nu_c\times \vcenter{\hbox{\begin{tikzpicture}
        \draw[fill=black] (0,0.5) ellipse (0.06 and 0.06);
        \VeryThinLine[midarrow] (-0.5,0) -- (0,0.5);
        \VeryThinLine[midarrow] (-0.5,0) -- (-0.5,-0.5);
        \VeryThinLine[midarrow] (0,0) -- (0,0.5);
        \VeryThinLine[midarrow] (0.5,0) -- (0,0.5);
        \draw[draw=black, fill=yellow] (-0.5,0) circle (2pt);
        \node at (0.25,1) {$\overline{\psi_v}$};
        \node at (-0.5,-0.75) {$c$};
        \node at (0,-0.25) {$b$};
        \node at (0.5,-0.25) {$a$};
    \end{tikzpicture}}}
\end{equation}
with $\nu_c$ due to the $180$-degree rotation of $\phi_c$. 
As reflection also reverses the directions of the edges, it is straightforward to verify that $\hat{\theta}_{\theta(p)}\circ \hat{\theta}_p$ and $\hat{\theta}_p\circ \hat{\theta}_{\theta(p)}$ are identities on $\mathcal{H}_p$ and $\mathcal{H}_{\theta(p)}$. 

\begin{proposition}
    The $\mathbb{S}^2$-functional defined by the Levin--Wen ground state associated to a unitary fusion category $\mathcal{C}$ as in Equation~\eqref{def:: S^2 functional from Levin-Wen wave function} satisfies s-reflection positivity. 
\end{proposition}
\begin{proof}
    Fix $m\geq 0$, let $\cM_1,\dots,\cM_L$ be stratifications of $\mathbb{D}^2_+$ with common boundary $\mathcal{S}(m)$, and, for each $l$, let $\ket{\xi_l}\in \mathcal{H}_{\cM_l}$ be a labelled stratification. 
    Consider an edge $e$ of $\theta(\cM_{l'})\cup \cM_l$ that meets both $\mathbb{D}^2_+$ and $\mathbb{D}^2_-$, and let $\theta(v)$ and $w$ denote its endpoints, where $v\in \cM_{l'}^0$ and $w\in \cM_l^0$, respectively. 
    Let $b$ be the simple object assigned to $e$ in the evaluation of $\hat{\theta}\ket{\xi_{l'}}\otimes \ket{\xi_l}$.
    When $b$ is self-dual and $e$ is directed from $w$ to $\theta(v)$ (so the edge labelled by $b$ is incoming to $v$), the resulting diagram around $e$ is 
    \begin{equation}
        \vcenter{\hbox{\begin{tikzpicture}
            \VeryThinLine[dashed] (0,1) -- (0,-1.75);
            \VeryThinLine[midarrow] (1,-0.5) arc (0:-180:1 and 0.75);
            \draw[fill=black] (-1.5,0.5) ellipse (0.06 and 0.06);
            \VeryThinLine[midarrow] (-1,0.1) -- (-1.5,0.5); 
            \VeryThinLine[midarrow] (-1,0) -- (-1,-0.5);
            \draw[draw=black, fill=yellow] (-1,0.05) circle (2pt);
            \draw[draw=black, fill=yellow] (-1,-0.55) circle (2pt);
            \VeryThinLine (-1.5,0.5) -- (-1.5,0);
            \VeryThinLine (-1.5,0.5) -- (-2,0);
            \node at (-1.25,0.8) {$\overline{\xi_{l'}(v)}$};
            \draw[fill=black] (1.5,0.5) ellipse (0.06 and 0.06);
            \VeryThinLine(1.5,0.5) -- (1,0) -- (1,-0.5);
            \VeryThinLine (1.5,0.5) -- (1.5,0);
            \VeryThinLine (1.5,0.5) -- (2,0);
            \node at (1.25,0.75) {$\xi_l(w)$};
            \node at (0.75,0) {$b$};
        \end{tikzpicture}}} = \vcenter{\hbox{\begin{tikzpicture}
            \VeryThinLine[dashed] (0,1) -- (0,-1.75);
            \VeryThinLine[midarrow] (1,-0.5) arc (0:-180:1 and 0.75);
            \draw[fill=black] (-1.5,0.5) ellipse (0.06 and 0.06);
            \VeryThinLine(-1,0) -- (-1.5,0.5); 
            \VeryThinLine(-1,0) -- (-1,-0.5);
            \VeryThinLine (-1.5,0.5) -- (-1.5,0);
            \VeryThinLine (-1.5,0.5) -- (-2,0);
            \node at (-1.25,0.8) {$\overline{\xi_{l'}(v)}$};
            \draw[fill=black] (1.5,0.5) ellipse (0.06 and 0.06);
            \VeryThinLine(1.5,0.5) -- (1,0) -- (1,-0.5);
            \VeryThinLine (1.5,0.5) -- (1.5,0);
            \VeryThinLine (1.5,0.5) -- (2,0);
            \node at (1.25,0.75) {$\xi_l(w)$};
            \node at (0.75,0) {$b$};
        \end{tikzpicture}}}
    \end{equation}
    When $e$ is directed from $\theta(v)$ to $w$, we have
    \begin{equation}
        \nu_b\times \vcenter{\hbox{\begin{tikzpicture}
            \VeryThinLine[dashed] (0,1) -- (0,-1.75);
            \VeryThinLine[midarrow] (-1,-0.5) arc (180:360:1 and 0.75);
            \draw[fill=black] (-1.5,0.5) ellipse (0.06 and 0.06);
            \VeryThinLine[midarrow] (-1,0.1) -- (-1.5,0.5); 
            \VeryThinLine (-1,0) -- (-1,-0.5);
            \draw[draw=black, fill=yellow] (-1,0.05) circle (2pt);
            \VeryThinLine (-1.5,0.5) -- (-1.5,0);
            \VeryThinLine (-1.5,0.5) -- (-2,0);
            \node at (-1.25,0.8) {$\overline{\xi_{l'}(v)}$};
            \draw[fill=black] (1.5,0.5) ellipse (0.06 and 0.06);
            \VeryThinLine[midarrow] (1.5,0.5) -- (1,0);
            \VeryThinLine(1,0) -- (1,-0.5);
            \draw[draw=black, fill=yellow] (1,-0.05) circle (2pt);
            \VeryThinLine (1.5,0.5) -- (1.5,0);
            \VeryThinLine (1.5,0.5) -- (2,0);
            \node at (1.25,0.75) {$\xi_{l}(w)$};
            \node at (0.75,0) {$b$};
        \end{tikzpicture}}} = \vcenter{\hbox{\begin{tikzpicture}
            \VeryThinLine[dashed] (0,1) -- (0,-1.75);
            \VeryThinLine[midarrow] (1,-0.5) arc (0:-180:1 and 0.75);
            \draw[fill=black] (-1.5,0.5) ellipse (0.06 and 0.06);
            \VeryThinLine[midarrow] (-1,0) -- (-1.5,0.5); 
            \VeryThinLine (-1,0) -- (-1,-0.5);
            \VeryThinLine (-1.5,0.5) -- (-1.5,0);
            \VeryThinLine (-1.5,0.5) -- (-2,0);
            \node at (-1.25,0.8) {$\overline{\xi_{l'}(v)}$};
            \draw[fill=black] (1.5,0.5) ellipse (0.06 and 0.06);
            \VeryThinLine[midarrow] (1.5,0.5) -- (1,0);
            \VeryThinLine(1,0) -- (1,-0.5);
            \VeryThinLine (1.5,0.5) -- (1.5,0);
            \VeryThinLine (1.5,0.5) -- (2,0);
            \node at (1.25,0.75) {$\xi_{l}(w)$};
            \node at (0.75,0) {$b$};
        \end{tikzpicture}}}
    \end{equation}
    The two graphical identities show that, irrespective of orientation, the isomorphisms $\phi_b$ and $\phi_b^{-1}$ cancel along every self-dual strand crossing the reflection axis, while they do not occur on non-self-dual strands. 
    Therefore, evaluating the glued diagram gives precisely the categorical inner product of the evaluations of its two halves. 
    Summing over $l$ and $l'$, we obtain
    \begin{equation}
        \sum^L_{l,l'=1}Z(\hat{\theta}(\xi_{l'})\otimes \xi_l) = \Braket{\sum_{l'}\mathrm{eval}_{\mathbb{D}^2_+}(\xi_{l'}),\sum_{l}\mathrm{eval}_{\mathbb{D}^2_+}(\xi_{l})}_{\mathrm{tr}_{\mathcal{C}}}\geq 0.
    \end{equation}
    The final inequality follows from the positive definiteness of the categorical inner product in the unitary fusion category $\mathcal{C}$. 
\end{proof}

Finally, we verify the topological connectedness~\ref{def:: topological connectedness}. 

\begin{proposition}
    The $\mathbb{S}^2$-functional $Z_{\mathcal{C}}$ defined by the Levin--Wen ground state associated to a unitary fusion category $\mathcal{C}$ satisfies topological connectedness, with $\tau = \mu^2 = \sum_{i\in \Irr(\mathcal{C})} d^2_i$. 
\end{proposition}
\begin{proof}
Let the common boundary of $\cN$ and $\cM$ be $\mathcal S(m)$, and set $V_m=\Hom(\mathbb 1,A^m)$.
Equip $V_m$ with the unweighted categorical-trace inner product
\begin{equation}
    (g,f):=\tr_{\mathcal C}(g^*f).
\end{equation}
Let $D_m$ be the positive operator on $V_m$ whose restriction to each boundary-label summand is
\begin{equation}
    \left.D_m\right|_{\Hom(\mathbb 1,a_1\cdots a_m)}
    =\sqrt{d_{a_1}\cdots d_{a_m}}\,\mathrm{Id}.
\end{equation}
Then the skein-module inner product in Equation~\eqref{eqn:: skein module inner product} and the trace inner product are related by
\begin{equation}
    \braket{g,f}
    = (g,D_m^{-1}f).
\end{equation}
For $\mathcal X\in\{\cN,\cM\}$, write $\mathrm{eval}_{\mathcal X}^{\ddagger}: V_m\rightarrow \mathcal{H}_{\mathcal{X}}$ for the adjoint of $\mathrm{eval}_{\mathcal X}$ when $V_m$ is equipped with the trace inner-product, and retain $\mathrm{eval}_{\mathcal X}^{\dagger}$ for the adjoint with respect to the skein-module inner product.
The preceding relation between the two inner products gives
\begin{equation}\label{eqn:: relation between evaluation adjoints}
    \mathrm{eval}_{\mathcal X}^{\ddagger}
    =\mathrm{eval}_{\mathcal X}^{\dagger}D_m.
\end{equation}
The definition of $Z_{\mathcal{C}}$ gives
\begin{equation}
    Z_{\mathcal C}(\hat\theta(\eta)\otimes\psi)
    =(\mathrm{eval}_{\cN}\eta,\mathrm{eval}_{\cM}\psi),
    \qquad
    \eta\in\mathcal H_{\cN},\quad \psi\in\mathcal H_{\cM}.
\end{equation}
Consequently, the definition of the partial trace gives
\begin{equation}\label{eqn:: connect eval and reduced density matrix}
    \widetilde{\rho}_{\mathbb D^2_+}(\theta(\mathcal X)\cup\mathcal Y)
    =\mathrm{eval}_{\mathcal Y}^{\ddagger}\left( \mathrm{eval}_{\mathcal X}
    \mathrm{eval}_{\mathcal X}^{\ddagger} \right)\mathrm{eval}_{\mathcal Y},
    \qquad
    \mathcal X,\mathcal Y\in\{\cN,\cM\}.
\end{equation}

By Proposition ~\ref{prop:: eval is an isometry}, $\mu^{-\lvert F_{\mathcal X}\rvert}\mathrm{eval}_{\mathcal X}$ is a partial isometry with respect to $\braket{\cdot,\cdot}$ when $\mathcal X^1$ is connected.
When $\mathcal X^1$ is disconnected, multiplicativity factors the evaluation over its connected components, and the definition of $\lvert F_{\mathcal X}\rvert$ makes the corresponding powers of $\mu$ additive.
Applying the same result to the connected components therefore defines the range projection $P_{\mathcal X}$ by
\begin{equation}
    P_{\mathcal X}
    :=\mu^{-2\lvert F_{\mathcal X}\rvert}
    \mathrm{eval}_{\mathcal X}\mathrm{eval}_{\mathcal X}^{\dagger},
    \qquad
    \mathcal X\in\{\cN,\cM\}.
\end{equation}
Since the evaluation maps preserve summands with fixed boundary labels, each $P_{\mathcal X}$ does as well.
Since $D_m$ is scalar on each such summand, $P_{\mathcal X}D_m=D_mP_{\mathcal X}$.
Equation~\eqref{eqn:: relation between evaluation adjoints} now gives $\mathrm{eval}_{\mathcal X}\mathrm{eval}_{\mathcal X}^{\ddagger}
    =\mu^{2\lvert F_{\mathcal X}\rvert}P_{\mathcal X}D_m$, and hence
\begin{equation}\label{eqn:: reduced density matrix from range projection}
    \widetilde{\rho}_{\mathbb D^2_+}(\theta(\mathcal X)\cup\mathcal Y)
    =\mu^{2\lvert F_{\mathcal X}\rvert}
    \mathrm{eval}_{\mathcal Y}^{\ddagger}
    P_{\mathcal X}D_m\mathrm{eval}_{\mathcal Y}.
\end{equation}

Since $\cN^1$ is connected, repeated use of the orthonormal-basis resolution in Lemma ~\ref{lemma:: ONB of 4-point hom space} represents every element of $V_m$ as the evaluation of a state in $\mathcal H_{\cN}$.
Thus $\mathrm{eval}_{\cN}$ is surjective and $P_{\cN}=\mathrm{Id}_{V_m}$.
By $P_{\cM}D_m\mathrm{eval}_{\cM}=D_m\mathrm{eval}_{\cM}$, we have
\begin{equation}
    \widetilde{\rho}_{\mathbb D^2_+}(\theta(\cN)\cup\cM)
    =\mu^{2(\lvert F_{\cN}\rvert-\lvert F_{\cM}\rvert)}
    \widetilde{\rho}_{\mathbb D^2_+}(\theta(\cM)\cup\cM).
\end{equation}
This proves TC1 as in Equation~\eqref{eqn:: TC1}.

For TC2, note that by Equation~\eqref{eqn:: reduced density matrix from range projection} the range projection of $\rho_{\mathbb{D}^2_+}(\theta(\cM)\cup \cN)$ is 
\begin{equation}
    \Pi =\mu^{-2\lvert F_{\cN}\rvert}
    \mathrm{eval}_{\cN}^{\dagger}P_{\cM}\mathrm{eval}_{\cN}
    \in\mathcal B(\mathcal H_{\cN}).
\end{equation}
Using $\mathrm{eval}_{\cN}^{\ddagger}=\mathrm{eval}_{\cN}^{\dagger}D_m$, Equation~\eqref{eqn:: reduced density matrix from range projection} with $\mathcal X=\mathcal Y=\cN$ spells 
\begin{equation}
    \widetilde{\rho}_{\mathbb D^2_+}(\theta(\cN)\cup\cN)
    =\mu^{2\lvert F_{\cN}\rvert}
    \mathrm{eval}_{\cN}^{\dagger}D_m^2\mathrm{eval}_{\cN}.
\end{equation}
By $P_{\cM}D_m=D_mP_{\cM}$, $\Pi$ commutes with $\widetilde{\rho}_{\mathbb D^2_+}(\theta(\cN)\cup\cN)$. 
Taking $(\mathcal X,\mathcal Y)=(\cM,\cN)$ in Equation~\eqref{eqn:: reduced density matrix from range projection} yields
\begin{equation}
\begin{aligned}
    \widetilde{\rho}_{\mathbb D^2_+}(\theta(\cM)\cup\cN)
    &=\mu^{2\lvert F_{\cM}\rvert}
    \mathrm{eval}_{\cN}^{\dagger}D_mP_{\cM}D_m\mathrm{eval}_{\cN} =\mu^{2(\lvert F_{\cM}\rvert-\lvert F_{\cN}\rvert)}\, \Pi
    \widetilde{\rho}_{\mathbb D^2_+}(\theta(\cN)\cup\cN)\Pi.
\end{aligned}
\end{equation}
Since $\tau=\mu^2$, this proves Equation~\eqref{eqn:: TC2} and completes the proof.
\end{proof}

Thus, for the raw categorical evaluation, the topological-connectedness scalar is $\tau=\mu^2$.
Equivalently, the global dimension is recovered from the functional and its induced reduced density matrices as $\mu=\sqrt{\tau}$.

\begin{remark}
    We notice that TC1 also follows from the fact that Levin--Wen wave functions on different lattices are related by local partial isometries. 
    It was shown in \cite{KoenigReichardtVidal2009,HSW2012LevinWenGSD} that the Levin--Wen wave function is invariant under local mutations of the underlying lattice. 
\end{remark}


\section{Characterizing the Levin--Wen Wave Function}\label{sec:: uniqueness of Levin-Wen wave function}

In this section, we prove that a multiplicative, s-reflection-positive, and homeomorphism-invariant $\mathbb{S}^2$-functional $Z$ determines a unitary fusion category $\mathcal{C}(Z)$ and reproduces its Levin--Wen wave functions on every stratification $\cM$ such that $\cM^1$ is connected, provided that $Z$ also satisfies:
\begin{itemize}
    \item Local non-degeneracy~\ref{def:: local non degeneracy},
    \item Topological connectedness~\ref{def:: topological connectedness},
    \item Commutativity~\ref{def:: commutativity}. 
\end{itemize}


\subsection{Consequences of the topological connectedness condition}
We now derive several consequences of topological connectedness~\ref{def:: topological connectedness}. 
If $\cM$ and $\cN$ are stratifications of $\mathbb{D}^2_+$ such that $\partial\cM = \partial\cN$. 
Then $\ket{\tilde{\Psi}_{\theta(\cN)\cup \cM}} \in \mathcal{H}_{\theta(\cN)}\otimes \mathcal{H}_{\cM}$ determines a linear map $K_{\cN,\cM}: \mathcal{H}_{\cM}\rightarrow \mathcal{H}_{\cN}$, defined by 
\begin{equation}\label{eqn:: matrix coefficient of the kernel}
    \braket{\eta_{\cN}|K_{\cN,\cM}|\psi_{\cM}} = \braket{\tilde{\Psi}_{\theta(\cN)\cup \cM}| \hat{\theta}(\eta_{\cN})\otimes \psi_{\cM}}. 
\end{equation}
The linear maps $K_{\cN,\cM}$ form the matrix entries of the inner-product matrix defined by $Z$. 
Let $\cM_1,\dots,\cM_L$ be stratifications of $\mathbb{D}^2_+$ sharing the same boundary, and write $K_{l,l'}$ for $K_{\cM_{l},\cM_{l'}}$. 
Then the kernel of the inner product induced by $Z$ on $\bigoplus^L_{l=1}\mathcal{H}_{\cM_l}$ is given by 
\begin{equation}\label{eqn:: inner-product matrix}
    \begin{bmatrix}
        K_{1,1} & K_{1,2} & \cdots & K_{1,L}\\
        K_{2,1} & K_{2,2} & \cdots & K_{2,L}\\
        \vdots & & \ddots& \vdots\\
        K_{L,1} & K_{L,2} & \cdots & K_{L,L}
    \end{bmatrix}
\end{equation}
In particular, the functional $Z$ is s-reflection positive if and only if Equation~\eqref{eqn:: inner-product matrix} is positive semidefinite for every choice of stratifications $\{\cM_l\}^L_{l=1}$. 

\begin{lemma}\label{lemma:: reduced density matrices in terms of the kernel}
    Let $Z$ be a $\mathbb{S}^2$-functional with s-reflection positivity. 
    Then for stratifications $\cM,\cN$ of $\mathbb{D}^2_+$ with the same boundary, we have $K_{\cM,\cN} = K_{\cN,\cM}^\dagger$, and 
    \begin{equation}
        K_{\cM,\cN}K_{\cN,\cM} = \widetilde{\rho}_{\mathbb{D}^2_+}(\theta(\cN)\cup \cM). 
    \end{equation}
\end{lemma}
\begin{proof}
Apply Equation~\eqref{eqn:: inner-product matrix} to the pair $\{\cM,\cN\}$.
By s-reflection positivity, the resulting $2\times 2$ block operator is positive semidefinite, hence self-adjoint.
Its two off-diagonal blocks therefore satisfy
\begin{equation}
    K_{\cM,\cN}=K_{\cN,\cM}^{\dagger}.
\end{equation}

Let $\{\ket{e_i}\}_i$ be an orthonormal basis of $\mathcal H_{\cN}$.
Since $\hat\theta$ is anti-unitary, $\{\hat\theta\ket{e_i}\}_i$ is an orthonormal basis of $\mathcal H_{\theta(\cN)}$.
Thus, for any $\ket{\varphi},\ket{\psi}\in\mathcal H_{\cM}$, the definition of the partial trace and Equation~\eqref{eqn:: matrix coefficient of the kernel} give
\begin{equation}
\begin{aligned}
    &\braket{\varphi|\widetilde{\rho}_{\mathbb{D}^2_+}(\theta(\cN)\cup \cM)|\psi}\\
    &=\sum_i
    \braket{\hat\theta(e_i)\otimes\varphi|\tilde{\Psi}_{\theta(\cN)\cup \cM}}
    \braket{\tilde{\Psi}_{\theta(\cN)\cup \cM}|\hat\theta(e_i)\otimes\psi}\\
    &=\sum_i
    \overline{\braket{e_i|K_{\cN,\cM}|\varphi}}
    \braket{e_i|K_{\cN,\cM}|\psi}\\
    &=\braket{\varphi|K_{\cN,\cM}^{\dagger}K_{\cN,\cM}|\psi}
    =\braket{\varphi|K_{\cM,\cN}K_{\cN,\cM}|\psi}.
\end{aligned}
\end{equation}
Since $\varphi$ and $\psi$ are arbitrary, the claim follows.
\end{proof}

When $Z$ has s-reflection positivity, topological connectedness~\ref{def:: topological connectedness} can be recast in terms of the operators $K_{\cN,\cM}$.
Indeed, Lemma~\ref{lemma:: reduced density matrices in terms of the kernel} shows TC1~\eqref{eqn:: TC1} is equivalent to 
\begin{equation}\label{eqn:: TC in terms of the kernel}
    K_{\cM,\cN}K_{\cN,\cM} = \tau^{\lvert F_\cN \rvert-\lvert F_\cM \rvert} K_{\cM,\cM}^2,
\end{equation}
whenever $\cN^1$ is connected, and TC2~\eqref{eqn:: TC2} is equivalent to
\begin{equation}\label{eqn:: projected TC in terms of the kernel}
    K_{\cN,\cM}K_{\cM,\cN}
    =\tau^{\lvert F_\cM \rvert-\lvert F_\cN \rvert}\Pi K_{\cN,\cN}^2\Pi,
\end{equation}
where $\Pi\in\mathcal{B}(\mathcal{H}_{\cN})$ is the range projection specified in TC2.

For a stratification $\cM$ of $\mathbb{D}^2_+$ with boundary $\mathcal S(m)$, let $q_{\cM}:\mathcal H_{\cM}\rightarrow\tilde V_m$ be the quotient map, and equip $\tilde{V}_m$ with the inner product $\braket{\cdot,\cdot }_{\mathcal{C}(Z)}$ induced by the categorical trace of $\mathcal{C}(Z)$. 
By construction of $\mathcal C(Z)$, we have
\begin{equation}
    \braket{q_{\cM}(\eta),q_{\cM}(\xi)}_{\mathcal{C}(Z)} = Z(\hat{\theta}(\eta)\otimes \xi),\quad \forall \eta,\xi\in \mathcal H_{\cM}.
\end{equation}
The definition of $K$ then reads 
\begin{equation}\label{eq:: K as pullback of quotient inner product}
    K_{\cN,\cM}=q_{\cN}^{\dagger}q_{\cM}. 
\end{equation}

\begin{theorem}\label{Thm:: key consequence of TC}
    Let $Z$ be a $\mathbb{S}^2$-functional with s-reflection positivity. 
    Suppose that $Z$ satisfies the topological connectedness~\ref{def:: topological connectedness}.
    Then for any stratification $\cN$ of $\mathbb{D}^2_+$ such that $\cN^1$ is connected, the quotient map $q_{\cN}:\mathcal H_{\cN}\rightarrow\tilde V_m$ is surjective. 
\end{theorem}
\begin{proof}
Let $\cN$ and $\cM$ be stratifications of $\mathbb{D}^2_+$ with the same boundary such that $\cN^1$ and $\cM^1$ are connected.
Set $n_{\cN}=\lvert F_\cN \rvert$ and $n_{\cM}=\lvert F_\cM \rvert$, and write
\begin{equation}
    K_{\cN,\cM}=u\lvert K_{\cN,\cM}\rvert
\end{equation}
for the polar decomposition, where $u:\mathcal{H}_{\cM}\rightarrow\mathcal{H}_{\cN}$ is the associated partial isometry.
By Equation~\eqref{eqn:: TC in terms of the kernel} and the positivity of $K_{\cM,\cM}$, we have $\lvert K_{\cN,\cM}\rvert
    = \tau^{(n_{\cN}-n_{\cM})/2}K_{\cM,\cM}$. 
Interchanging $\cM$ and $\cN$ gives
\begin{equation}
    \lvert K_{\cM,\cN}\rvert
    = \tau^{(n_{\cM}-n_{\cN})/2}K_{\cN,\cN}.
\end{equation}
The identities $u^{\dagger}K_{\cN,\cM}=K_{\cM,\cN}u=\lvert K_{\cN,\cM}\rvert$ and $u^{\dagger}\lvert K_{\cM,\cN}\rvert u=\lvert K_{\cN,\cM}\rvert$ therefore imply
\begin{equation}
    \begin{aligned}
        u^{\dagger}K_{\cN,\cM}
        &=K_{\cM,\cN}u
        =\tau^{(n_{\cN}-n_{\cM})/2}K_{\cM,\cM},\\
        u^{\dagger}K_{\cN,\cN}u
        &=\tau^{n_{\cN}-n_{\cM}}K_{\cM,\cM}.
    \end{aligned}
\end{equation}
For $\ket{\psi}\in\mathcal{H}_{\cM}$, define
\begin{equation}\label{eq:: constructive eta for connected stratifications}
    \ket{\eta}
    =\tau^{(n_{\cM}-n_{\cN})/2}u\ket{\psi}
    \in\mathcal{H}_{\cN}.
\end{equation}
Using Equation~\eqref{eq:: K as pullback of quotient inner product} and the preceding identities, we obtain
\begin{equation}
    \begin{aligned}
        \left\lVert q_{\cM}(\psi)-q_{\cN}(\eta)\right\rVert^2 &=
        \braket{\psi|K_{\cM,\cM}|\psi}
        -\tau^{(n_{\cM}-n_{\cN})/2}\braket{\psi|K_{\cM,\cN}u|\psi}\\
        &\qquad
        -\tau^{(n_{\cM}-n_{\cN})/2}\braket{\psi|u^{\dagger}K_{\cN,\cM}|\psi}
        +\tau^{n_{\cM}-n_{\cN}}\braket{\psi|u^{\dagger}K_{\cN,\cN}u|\psi}\\
        &=\left(1-1-1+1\right)\braket{\psi|K_{\cM,\cM}|\psi}
        =0.
    \end{aligned}
\end{equation}
Thus $q_{\cM}(\psi)=q_{\cN}(\eta)$, so $\image(q_{\cM})\subseteq\image(q_{\cN})$.
Interchanging $\cM$ and $\cN$ gives the reverse inclusion.
Hence $\image(q_{\cM})=\image(q_{\cN})$ whenever $\cM^1,\cN^1$ are connected.
Denote this common image by $W\subseteq\tilde V_m$.

Now let $\mathcal P$ be an arbitrary stratification of $\mathbb{D}^2_+$ with $\partial\mathcal P=\mathcal S(m)$, and write $n_{\mathcal{P}}$ for $\lvert F_{\mathcal P}\rvert$. 
By s-reflection positivity,
\begin{equation}
    \begin{bmatrix}
        K_{\cN,\cN} & K_{\cN,\mathcal P}\\
        K_{\mathcal P,\cN} & K_{\mathcal P,\mathcal P}
    \end{bmatrix}
    \geq 0.
\end{equation}
Write $K_{\cN,\mathcal P}=v\lvert K_{\cN,\mathcal P}\rvert$ for the polar decomposition.
Equations~\eqref{eqn:: TC in terms of the kernel} and~\eqref{eqn:: projected TC in terms of the kernel} give
\begin{equation}
    \begin{aligned}
        K_{\mathcal P,\cN}K_{\cN,\mathcal P}
        &=\tau^{n_{\cN}-n_{\mathcal{P}}}K_{\mathcal P,\mathcal P}^{2},\\
        K_{\cN,\mathcal P}K_{\mathcal P,\cN}
        &=\tau^{n_{\mathcal{P}}-n_{\cN}}\Pi K_{\cN,\cN}^{2}\Pi.
    \end{aligned}
\end{equation}
Consequently,
\begin{equation}
    K_{\cN,\mathcal P}
    =\tau^{(n_{\cN}-n_{\mathcal{P}})/2}vK_{\mathcal P,\mathcal P}.
\end{equation}
Let $p=vv^{\dagger}$.
Then $p$ is the support projection of $K_{\cN,\mathcal P}K_{\mathcal P,\cN}$.
Since $\Pi$ is the range projection of the same operator, $p=\Pi$, and
\begin{equation}
    pK_{\cN,\cN}^{2}p
    =\Pi K_{\cN,\cN}^{2}\Pi
    =\tau^{2(n_{\cN}-n_{\mathcal{P}})}
    vK_{\mathcal P,\mathcal P}^{2}v^{\dagger}.
\end{equation}
Therefore,
\begin{equation}\label{eqn:: compressed square root}
    \left(pK_{\cN,\cN}^{2}p\right)^{1/2}
    =\tau^{n_{\cN}-n_{\mathcal{P}}}
    vK_{\mathcal P,\mathcal P}v^{\dagger}.
\end{equation}

In the remainder of the proof, the inverse of a positive operator is taken on its range subspace.
Applying the generalized Schur-complement criterion to the lower-right block $K_{\mathcal P,\mathcal P}$ of the block matrix above, which also gives the required range inclusion, we obtain
\begin{equation}
    \begin{aligned}
        K_{\cN,\cN} \geq K_{\cN,\mathcal P}K_{\mathcal P,\mathcal P}^{-1}K_{\mathcal P,\cN} =\tau^{n_{\cN}-n_{\mathcal{P}}}
        vK_{\mathcal P,\mathcal P}v^{\dagger}.
    \end{aligned}
\end{equation}
Compressing this inequality by $p$ and using Equation~\eqref{eqn:: compressed square root} gives
\begin{equation}
    pK_{\cN,\cN}p
    \geq\left(pK_{\cN,\cN}^{2}p\right)^{1/2}.
\end{equation}
On the other hand,
\begin{equation}
    pK_{\cN,\cN}^{2}p-\left(pK_{\cN,\cN}p\right)^{2}
    =pK_{\cN,\cN}(1-p)K_{\cN,\cN}p
    \geq0.
\end{equation}
The operator monotonicity of the function $t\mapsto \sqrt{t}$ gives the reverse inequality.
Hence
\begin{equation}
    pK_{\cN,\cN}p
    =\left(pK_{\cN,\cN}^{2}p\right)^{1/2}
    =\tau^{n_{\cN}-n_{\mathcal{P}}}
    vK_{\mathcal P,\mathcal P}v^{\dagger}.
\end{equation}
Equality also implies
\begin{equation}
    pK_{\cN,\cN}(1-p)K_{\cN,\cN}p=0,
\end{equation}
and hence $p$ commutes with $K_{\cN,\cN}$. 
Thus we have $pK_{\cN,\cN}^{-1}p
    =\tau^{n_{\mathcal{P}}-n_{\cN}}
    vK_{\mathcal P,\mathcal P}^{-1}v^{\dagger}$. 
Since $v^{\dagger}v$ is the support projection of $K_{\mathcal P,\mathcal P}$, it follows that
\begin{equation}
    \begin{aligned}
        K_{\mathcal P,\cN}K_{\cN,\cN}^{-1}K_{\cN,\mathcal P}
        &=\tau^{n_{\cN}-n_{\mathcal{P}}}
        K_{\mathcal P,\mathcal P}v^{\dagger}
        \left(pK_{\cN,\cN}^{-1}p\right)
        vK_{\mathcal P,\mathcal P}\\
        &=K_{\mathcal P,\mathcal P}K_{\mathcal P,\mathcal P}^{-1}K_{\mathcal P,\mathcal P} =K_{\mathcal P,\mathcal P}.
    \end{aligned}
\end{equation}
By Equation~\eqref{eq:: K as pullback of quotient inner product}, this identity is equivalent to
\begin{equation}
    q_{\mathcal P}^{\dagger}
    \left(1-q_{\cN}K_{\cN,\cN}^{-1}q_{\cN}^{\dagger}\right)
    q_{\mathcal P}
    =0.
\end{equation}
Since the operator $q_{\cN}K_{\cN,\cN}^{-1}q_{\cN}^{\dagger}$ is the orthogonal projection onto $\image(q_{\cN})=W$, it follows that $\image(q_{\mathcal P})\subseteq W$.
By definition, $\tilde V_m$ is spanned by the images of the quotient maps $q_{\mathcal P}$ over all stratifications $\mathcal P$ with boundary $\mathcal S(m)$.
Hence $W=\tilde V_m$, so $q_{\cN}$ is surjective. 
\end{proof}

\begin{corollary}
    A multiplicative, s-reflection-positive $\mathbb{S}^2$-functional $Z$ that satisfies topological connectedness~\ref{def:: topological connectedness} is automatically locally finite. 
\end{corollary}
\begin{proof}
For each $m\geq 0$, choose a stratification $\cM$ of $\mathbb{D}^2_+$ with boundary $\mathcal S(m)$ such that $\cM^1$ is connected.
By Theorem ~\ref{Thm:: key consequence of TC}, the quotient map $q_{\cM}:\mathcal H_{\cM}\rightarrow\tilde V_m$ is surjective.
Since the local label spaces are finite-dimensional, $\mathcal H_{\cM}$ is finite-dimensional.
Consequently, $\tilde V_m$ is finite-dimensional.
As this holds for every $m\geq0$, $Z$ is locally finite.
\end{proof}

\begin{lemma}\label{lemma:: A contains all simple objects}
    Assume that $Z$ is multiplicative, s-reflection-positive, and satisfies topological connectedness~\ref{def:: topological connectedness}. 
    Then every simple object of $\mathcal C(Z)$ is isomorphic to a direct summand of $A=(1,\mathrm{Id}_1)$.
\end{lemma}
\begin{proof}
It suffices to prove that, for every $m\geq 0$, the following factorization holds: 
\begin{equation}\label{eq:: Am factors through A}
    \mathcal A_m=\mathcal A(1,m)\mathcal A(m,1). 
\end{equation}
Choose a connected stratification $B_m$ of a rectangle with boundary $\mathcal S(m,m)$ that has a horizontal cut meeting $B_m^1$ in one point. 
For $m\geq 1$, take a trivalent tree from the $m$ upper boundary points to one intermediate boundary point and stack its reflected copy below it.
For $m=0$, take a connected closed trivalent graph with a separating edge and cut that edge.
Then every labelled vector in $\mathcal H_{B_m}$ is a finite sum of tensors $\eta\otimes \xi$, where $\xi$ labels the upper piece and $\eta$ labels the lower piece.
For each simple tensor, vertical composition gives
\begin{equation}
    q_{B_m}(\eta\otimes \xi)=q_{\mathrm{low}}(\eta)q_{\mathrm{up}}(\xi)\in \mathcal A(1,m)\mathcal A(m,1).
\end{equation}
Since $B_m^1$ is connected, Theorem ~\ref{Thm:: key consequence of TC} implies that $q_{B_m}$ is surjective onto $\mathcal A_m\cong\tilde V_{2m}$.
Therefore, Equation~\eqref{eq:: Am factors through A} holds.
If a simple object occurs as a summand of $A^m$, the factorization provides a nonzero morphism to or from $A$.
Semisimplicity then implies that this simple object is a direct summand of $A$, proving the claim.
\end{proof}

\begin{corollary}\label{corollary:: A is sum of simple objects}
    Under the assumptions of Lemma ~\ref{lemma:: A contains all simple objects}, $Z$ has complete finiteness. 
    If $Z$ also satisfies commutativity ~\ref{def:: commutativity}, then
    \begin{equation}
        A \cong \bigoplus_{x\in \Irr(\mathcal{C}(Z))} x. 
    \end{equation}
    In particular, $Z$ is completely finite. 
\end{corollary}
\begin{proof}
By Lemma ~\ref{lemma:: A contains all simple objects}, every simple object of $\mathcal C(Z)$ occurs as a direct summand of $A$.
Since $\mathcal A_1$ is finite-dimensional, $A$ has finite length, so the lemma also implies that $\Irr(\mathcal C(Z))$ is finite and hence $Z$ is completely finite.
Write $A\cong\bigoplus_{x\in\Irr(\mathcal C(Z))}x^{\oplus n_x}$, where $n_x\geq 1$ for every $x$.
The graphical commutativity condition gives $\tr((xy-yx)z)=0$ for all $z\in\End(A)$.
Faithfulness of the categorical trace therefore implies that $\End(A)=\mathcal A_1$ is commutative, hence $n_x=1$ for every $x$, which proves the claim. 
\end{proof}

\begin{proposition}\label{prop:: canonical operator from connected stratifications}
    Let $Z$ be a multiplicative $\mathbb{S}^2$-functional with s-reflection positivity that satisfies topological connectedness.
    For any stratification $\cN$ of $\mathbb{D}^2_+$ with connected $\cN^1$ and $\partial\cN = \mathcal{S}(m)$, set $\mathfrak{D}_{\cN}=q_{\cN}q^\dagger_{\cN} \in \mathcal{B}(\tilde{V}_m)$.
    Then $\mathfrak{D}_{\cN}$ is strictly positive, and for any other stratification $\cN'$ with connected $(\cN')^1$ and $\partial\cN'=\mathcal S(m)$,
    \begin{equation}\label{eq:: scaling of canonical operators}
        \mathfrak{D}_{\cN'}
        =\tau^{\lvert F_{\cN'} \rvert-\lvert F_\cN \rvert}\mathfrak{D}_{\cN}.
    \end{equation}
    In particular, $\tau^{-\lvert F_\cN \rvert}\mathfrak{D}_{\cN}$ is independent of the choice of $\cN$.
\end{proposition}
\begin{proof}
By Theorem ~\ref{Thm:: key consequence of TC}, $q_{\cN}$ is surjective, so $q_{\cN}^{\dagger}$ is injective.
Therefore, for every nonzero $x\in\tilde V_m$,
\begin{equation}
    \braket{x,\mathfrak D_{\cN}x}
    =\lVert q_{\cN}^{\dagger}x\rVert^2
    >0,
\end{equation}
and hence $\mathfrak D_{\cN}$ is strictly positive.
Let $\cN'$ be another stratification of $\mathbb{D}^2_+$ with connected $(\cN')^1$ and $\partial\cN'=\mathcal S(m)$.
Applying Equation~\eqref{eqn:: TC in terms of the kernel} with $\cM=\cN$ and then using Equation~\eqref{eq:: K as pullback of quotient inner product} gives
\begin{equation}
    q_{\cN}^{\dagger}\mathfrak D_{\cN'}q_{\cN}
    =\tau^{\lvert F_{\cN'} \rvert-\lvert F_\cN \rvert}
    q_{\cN}^{\dagger}\mathfrak D_{\cN}q_{\cN}.
\end{equation}
Since $q_{\cN}$ is surjective and $q_{\cN}^{\dagger}$ is injective, it follows that
\begin{equation}
    \mathfrak D_{\cN'}
    =\tau^{\lvert F_{\cN'} \rvert-\lvert F_\cN \rvert}\mathfrak D_{\cN}. 
\end{equation}
\end{proof}

Set $\mathfrak A=A\boxtimes A^{\mathrm{op}}$.
For each $x\in\Irr(\mathcal C(Z))$, let $P_x$ denote the projection of $\mathfrak A$ onto the diagonal summand $x\boxtimes x^{\mathrm{op}}$.
Operators on $\mathcal{A}(m,n)$ can be interpreted as morphisms in the doubled category $\mathcal C(Z)\boxtimes\mathcal C(Z)^{\mathrm{op}}$.
For objects $x,y\in\mathcal C(Z)$, an operator $T$ on $\Hom(x,y)$ is identified with a doubled morphism as follows.
If $\{v_{\alpha}\}_{\alpha}$ is an orthonormal basis of $\Hom(x,y)$ with respect to the trace inner product, then $T$ corresponds to the morphism $\widehat{T}$ given by
\begin{equation}
    \widehat{T} = \sum_{\alpha,\beta}\braket{v_{\alpha},Tv_{\beta}}_{\mathcal{C}(Z)} v_{\alpha}\boxtimes v_{\beta}^{*}
    \in
    \Hom(x\boxtimes x^{\mathrm{op}},y\boxtimes y^{\mathrm{op}}).
\end{equation}
Here $v_{\beta}^{*}\in\Hom(y,x) = \Hom(x^{\mathrm{op}},y^{\mathrm{op}})$. 
This construction is independent of the chosen orthonormal basis. 

\begin{definition}\label{def:: F-positivity}
    Let $x,y$ be objects in $\mathcal C(Z)$. 
    A morphism $\alpha\in \Hom(x\boxtimes x^{\mathrm{op}},y\boxtimes y^{\mathrm{op}})$ is called (strictly) \emph{$\mathcal{F}$-positive} if the unique operator $T$ on $\Hom(x,y)$ that satisfies $\alpha = \widehat{T}$ is (strictly) positive.
\end{definition}

\begin{remark}
     Here $\mathcal{F}$-positivity is short for positivity of Fourier multiplier. 
     This notion appears frequently in the framework of quantum Fourier analysis \cite{JiangLiuWu2016,HJLW2023}. 
\end{remark}

For $\cM = \Lambda\mathcal{S}(2,1)$, write $\mathfrak{D}=\mathfrak{D}_{\Lambda\mathcal{S}(2,1)}$. 
The operator $\mathfrak{D}$ measures the discrepancy between the inner product on $\mathcal{H}_{\mathcal{S}(2,1)}$, namely the on-site Hilbert space, and the morphism space $\Hom_{\mathcal{C}(Z)}(A\otimes A,A)$. 
In the rest of this section, we deduce sufficient information that determines $\mathfrak{D}$ up to scalar multiples. 
Now consider 
\begin{equation}
    \mathbb{T}:= \widehat{\mathfrak{D}} \in \Hom(\mathfrak{A}^2,\mathfrak{A}). 
\end{equation}
By definition, $\mathbb{T}$ is strictly $\mathcal{F}$-positive.
For an orthonormal basis $\{\psi_i\}_i$ of $\mathcal H_{\mathcal S(2,1)}$, we have
\begin{equation}
    \mathbb{T} = \sum_i q_{\cM}(\psi_i)\boxtimes q_{\cM}(\psi_i)^* \in \Hom(\mathfrak A^2,\mathfrak A).
\end{equation}

\begin{remark}
    More generally, the quotient map $q^\dagger_{\cM}$ can be understood as a PEPS valued in the unitary fusion category $\mathcal{C}(Z)$. 
    Theorem \ref{Thm:: key consequence of TC} then says that the PEPS is injective if $\cM^1$ is connected. 
    Under this analogy, $\widehat{q^\dagger_{\cM}q_{\cM}}$ is the associated transfer matrix \cite{FannesNachtergaeleWerner1992,PerezGarciaVerstraeteCiracWolf2008}. 
\end{remark}

\begin{proposition}\label{prop:: Frobenius identity from TC}
Assume that $Z$ is multiplicative, s-reflection-positive, and satisfies topological connectedness~\ref{def:: topological connectedness}. 
Then the morphism $\mathbb{T}$ defined above satisfies the Frobenius identity
\begin{equation}\label{eqn:: Frobenius identity}
(\mathbb{T}\otimes \mathrm{Id}_{\mathfrak{A}})
\circ
(\mathrm{Id}_{\mathfrak{A}}\otimes \mathbb{T}^*)
= \mathbb{T}^*\mathbb{T}
=
(\mathrm{Id}_{\mathfrak{A}}\otimes \mathbb{T})
\circ
(\mathbb{T}^*\otimes \mathrm{Id}_{\mathfrak{A}}). 
\end{equation}
\end{proposition}
\begin{proof}
Let $H$ and $I$ denote the stratifications with boundary $\mathcal{S}(2,2)$ obtained by joining two trivalent vertices horizontally and vertically:
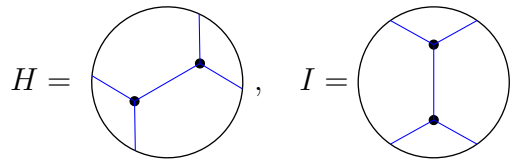
\begin{equation*}
    H = \vcenter{\hbox{\begin{tikzpicture}[rotate= 30]
          \VeryThinLine (0,0) circle (1);
          \draw[fill=black] (-0.5,0) ellipse (0.06 and 0.06);
          \draw[fill=black] (0.5,0) ellipse (0.06 and 0.06);
          \draw[blue] (-0.5,0) -- (0.5,0);
          \draw[blue] (-0.5,0) -- (-145:1);
          \draw[blue] (-0.5,0) -- (145:1);
          \draw[blue] (0.5,0) -- (35:1);
          \draw[blue] (0.5,0) -- (-35:1);
    \end{tikzpicture}}},\quad
    I = \vcenter{\hbox{\begin{tikzpicture}
          \VeryThinLine (0,0) circle (1);
          \draw[fill=black] (0,-0.5) ellipse (0.06 and 0.06);
          \draw[fill=black] (0,0.5) ellipse (0.06 and 0.06);
          \draw[blue] (0,0.5) -- (0,-0.5);
          \draw[blue] (0,0.5) -- (125:1);
          \draw[blue] (0,0.5) -- (55:1);
          \draw[blue] (0,-0.5) -- (-55:1);
          \draw[blue] (0,-0.5) -- (-125:1);
    \end{tikzpicture}}}
\end{equation*}
Both $H^1$ and $I^1$ are connected and $\partial H=\partial I=\mathcal S(2,2)$.
Equation~\eqref{eq:: scaling of canonical operators} gives
\begin{equation}
    \mathfrak D_I =\mathfrak D_H,
\end{equation}
since both stratifications have no plaquettes in the interior. 
Both $H$ and $I$ have two trivalent vertices, so $\mathcal H_H\cong\mathcal H_{\mathcal S(2,1)}\otimes \mathcal H_{\mathcal S(1,2)}\cong\mathcal H_I$.
Set $\psi_j^\vee:=\rho(\hat{\theta}(\psi_j))\in\mathcal H_{\mathcal S(1,2)}$.
Then $\{\psi_i\otimes\psi_j^\vee\}_{i,j}$ is an orthonormal basis of both local Hilbert spaces.
Set $a_i=[\psi_i]$ and
\begin{equation}
h_{ij}:=(a_i\otimes\mathrm{Id}_A)\circ(\mathrm{Id}_A\otimes a_j^*),
\qquad
\iota_{ij}:=a_j^*\circ a_i.
\end{equation}
The quotient maps satisfy $q_H(\psi_i\otimes\psi_j^\vee)=h_{ij}$ and $q_I(\psi_i\otimes\psi_j^\vee)=\iota_{ij}$.
Thus, under $\mathcal B(\mathcal{A}(2,2))\cong \mathcal{A}(2,2)\otimes \mathcal{A}(2,2)^*$ and the identification of $\mathcal{A}(2,2)$ with the opposite morphism space induced by adjoints, we obtain
\begin{equation}
\begin{aligned}
\widehat{\mathfrak{D}_{H}} &= 
\sum_{i,j}h_{ij}\boxtimes(h_{ij}^*)^{\mathrm{op}}
=
(\mathbb{T}\otimes\mathrm{Id}_{\mathfrak A})
\circ
(\mathrm{Id}_{\mathfrak A}\otimes\mathbb{T}^*),\\
\widehat{\mathfrak{D}_{I}} &= 
\sum_{i,j}\iota_{ij}\boxtimes(\iota_{ij}^*)^{\mathrm{op}}
=
\mathbb{T}^*\circ\mathbb{T}.
\end{aligned}
\end{equation}
The equality $\mathfrak{D}_{I}= \mathfrak{D}_{H}$ proves the first Frobenius equality.
Taking adjoints proves the second.
\end{proof}

\begin{lemma}
    Assume that $Z$ satisfies the assumptions of Proposition ~\ref{prop:: canonical operator from connected stratifications} and commutativity~\ref{def:: commutativity}.
    For any stratification $\cM$ of $\mathbb{D}^2_+$ with connected $\cM^1$ and $\partial\cM = \mathcal{S}(1,1)$, there is a scalar $c>0$ such that
    \begin{equation}
        \widehat{\mathfrak{D}_{\cM}} =c \sum_{x\in\Irr(\mathcal C(Z))}p_x\boxtimes p_x.
    \end{equation}
    where $p_x\in \End(A)$ is the projection onto the simple summand $x$ in $A$.
\end{lemma}
\begin{proof}
Let $\cM \#\cM$ be the stratification of $\mathbb{D}^2_+$ obtained by vertically concatenating two copies of $\cM$. 
Then $(\cM \#\cM)^1$ is  connected and $\partial(\cM \#\cM)=\mathcal S(1,1)$.
By Proposition ~\ref{prop:: canonical operator from connected stratifications}, there is a scalar $c>0$ such that
\begin{equation}
    \mathfrak{D}_{\cM \#\cM}=c\,\mathfrak{D}_{\cM}.
\end{equation}
This translates to the following identity for $\widehat{\mathfrak{D}_{\cM}}$:
\begin{equation}
    \widehat{\mathfrak{D}_{\cM}}^2
    =
    c\,\widehat{\mathfrak{D}_{\cM}}.
\end{equation}
Since $\mathcal A(1,1)$ is commutative by \ref{def:: commutativity}, $\{p_x\}_{x\in \Irr(\mathcal C(Z))}$ is an orthogonal basis of $\mathcal A(1,1)$, so we can write

\begin{equation}
    \widehat{\mathfrak{D}_{\cM}} =
    \sum_{x,y\in \Irr(\mathcal C(Z))}
    d_{x,y}\, p_x\boxtimes p_y.
\end{equation}

The condition $\widehat{\mathfrak{D}_{\cM}}^2=c\,\widehat{\mathfrak{D}_{\cM}}$ implies that $d_{x,y}^2 = c d_{x,y}$ for all $x,y\in \Irr(\mathcal C(Z))$, so either $d_{x,y}=0$ or $d_{x,y}=c$. 
Let $D=(d_{x,y})$ be the coefficient matrix of $\mathfrak{D}_{\cM}$ in this basis.
Since $\mathfrak{D}_{\cM}$ is strictly positive, $D$ is positive definite.
Thus $d_{x,x}>0$, and hence $d_{x,x}=c$, for all $x\in \Irr(\mathcal C(Z))$.
If $d_{x,y}=c$ for some $x\neq y$, then Hermitian symmetry gives $d_{y,x}=c$, and the corresponding $2\times 2$ principal minor is singular, contradicting positive definiteness of $D$.
Therefore $d_{x,y}=0$ when $x\neq y$, and
\begin{equation}
    \widehat{\mathfrak{D}_{\cM}}=c\sum_{x\in\Irr(\mathcal C(Z))}p_x\boxtimes p_x=cP.
\end{equation}
\end{proof}

\begin{corollary}\label{corollary:: T is block diagonal}
    Define $P_x = p_x\boxtimes p_x\in \End (A\boxtimes A^{\mathrm{op}})$. 
    For simple objects $i,j,k\in \Irr(\mathcal C(Z))$, define $\mathbb{T}^{ij}_k = P_k\mathbb{T}(P_i\otimes P_j)$. 
    Then we have 
    \begin{equation}
        \mathbb{T} = \sum_{i,j,k\in \Irr(\mathcal C(Z))} \mathbb{T}^{ij}_k. 
    \end{equation}
\end{corollary}
\begin{proof}
Let $\cN$ be a connected stratification of $\mathbb D^2_+$ with boundary $\mathcal S(1,1)$.
By the preceding lemma, there is a scalar $c>0$ such that
\begin{equation}
    \widehat{\mathfrak{D}_{\cN}} = c \sum_{x\in\Irr(\mathcal C(Z))} p_x\boxtimes p_x = c \sum_{x\in\Irr(\mathcal C(Z))} P_x.
\end{equation}
Set $P=\sum_{x\in\Irr(\mathcal C(Z))}P_x$.
Concatenate $\cN$ to each of the three edges of $\Lambda\mathcal S(2,1)$.
In each case, the resulting stratification is connected and has boundary $\mathcal S(2,1)$, so Proposition ~\ref{prop:: canonical operator from connected stratifications} implies the resulting operator on $\mathcal A(2,1)$ is a positive scalar multiple of $\mathfrak{D}$. 
Under the doubled-morphism identification, this gives scalars $b_0,b_1,b_2>0$ such that
\begin{equation}
\begin{aligned}
    P\mathbb T=
    b_0\mathbb T,\quad 
    \mathbb T(P\otimes\mathrm{Id}_{\mathfrak A})= b_1\mathbb T,\quad \mathbb T(\mathrm{Id}_{\mathfrak A}\otimes P) = b_2\mathbb T.
\end{aligned}
\end{equation}
Since $\mathfrak{D}$ is strictly positive, $\mathbb T\neq 0$. 
Since $P$ is a projection, composing each identity once more with the same projection gives $b_r^2\mathbb T=b_r\mathbb T$.
Hence $b_r=1$ for $r=0,1,2$ and $\mathbb T=P\circ\mathbb T\circ(P\otimes P)$.
Thus we obtain the result by expanding $P$. 
\end{proof}

\subsection{A uniqueness theorem}

Let $\gamma = \bigoplus_{x\in \Irr(\mathcal C(Z))} x\boxtimes x^{\mathrm{op}}$ be the diagonal subobject of $\mathfrak{A}$, and let $P$ be the projection onto $\gamma$.
Corollary ~\ref{corollary:: T is block diagonal} shows that $\mathbb{T}\in \Hom(\gamma^2,\gamma)$.
Therefore topological connectedness together with commutativity implies that $\mathbb{T}$ is a strictly $\mathcal{F}$-positive morphism in $\Hom(\gamma^2,\gamma)$ that satisfies the Frobenius condition. 
We now show that these assumptions together with rotation and modular-conjugation symmetry force $\mathbb{T}$ to be a scalar multiple of the multiplication of the canonical Frobenius algebra in $\mathcal{C}(Z)\boxtimes \mathcal C(Z)^{\mathrm{op}}$ \cite{Muger2003b}.

Let $\mathcal{C}$ be a spherical unitary fusion category with representatives of simple objects $\Irr(\mathcal{C})$, and set $\gamma = \bigoplus_{x\in \Irr(\mathcal{C})} x\boxtimes x^{\mathrm{op}}$. 
Let $\mathcal{T}\in \Hom(\gamma^2,\gamma)$ be a strictly $\mathcal{F}$-positive morphism that is invariant under $2\pi/3$-rotation and satisfies the Frobenius identity~\eqref{eqn:: Frobenius identity} with $(\mathbb T,\mathfrak A)$ replaced by $(\mathcal T,\gamma)$.
Let $\mathbf{D}$ be the strictly positive operator on $\bigoplus_{i,j,k\in \Irr(\mathcal{C})} \Hom(ij,k)$, equipped with the categorical trace inner product, such that $\mathcal{T} = \widehat{\mathbf{D}}$. 
Since $\mathcal{T}\in \Hom(\gamma^2,\gamma)$, $\mathbf{D}$ is block diagonal with respect to the direct sum decomposition. 
Denote the restriction of $\mathbf{D}$ to $\Hom(ij,k)$ by $\mathbf{D}^{ij}_k$ and set $\widehat{\mathbf{D}^{ij}_k} = \mathcal{T}^{ij}_k$; then $\mathbf{D}^{ij}_k$ is strictly positive on $\Hom(ij,k)$.
Similarly, we denote by $\mathbf{D}^k_{ij}$ the positive operator on $\Hom(k,ij)$ corresponding to $(\mathcal{T}^{ij}_k)^*$. 
For $i,j,k,l\in \Irr(\mathcal{C})$, let $\mathbf{E}^{ij}_{kl}$ be the positive operator on $\Hom(ij,kl)$ corresponding to the $(ij,kl)$-block of $\mathcal{T}^*\mathcal{T}$. 

\begin{lemma}\label{lemma:: key lemma on eigenvalues of D3}
    Let $i,j,k,l,p\in \Irr(\mathcal{C})$ such that $\Hom(ij,p)\neq 0$ and $\Hom(kl,p)\neq 0$. 
    Let $\alpha\in \Hom(ij,p)$ be an eigenvector of $\mathbf{D}^{ij}_p$ with eigenvalue $a>0$, and let $\beta\in \Hom(kl,p)$ be an eigenvector of $\mathbf{D}^{kl}_p$ with eigenvalue $b>0$. 
    Then $\beta^*\alpha$ is an eigenvector of $\mathbf{E}^{ij}_{kl}$ with eigenvalue $ab/d_p$.
\end{lemma}
\begin{proof}
    Without loss of generality, we assume that $\alpha,\beta$ are unit vectors. 
    Choose orthonormal eigenbases (with respect to inner product induced by $\tr_{\mathcal{C}}$) $\{\alpha_1,\dots,\alpha_n\}$ of $\mathbf{D}^{ij}_p$ and $\{\beta_1,\dots,\beta_m\}$ of $\mathbf{D}^{kl}_p$ with eigenvalues $a_1,\dots,a_n$ and $b_1,\dots,b_m$, such that $\alpha=\alpha_1$ and $\beta=\beta_1$.
    The vectors $\beta_1^*,\dots,\beta_m^*$ form an orthonormal eigenbasis of $\mathbf{D}^p_{kl}$ with eigenvalues $b_1,\dots,b_m$. 
    Thus
\begin{equation}
        \widehat{\mathbf{D}^p_{kl}}\widehat{\mathbf{D}^{ij}_p} = (\mathcal{T}^{kl}_p)^* \mathcal{T}^{ij}_p = \sum^n_{s=1}\sum^m_{t=1} a_s\cdot b_t \beta_t^*\alpha_s \boxtimes (\beta_t^*\alpha_s)^*. 
    \end{equation}
    Note that
    \begin{equation}
        \braket{\beta^*_{t'}\alpha_{s'},\beta^*_t\alpha_s} = \Braket{\vcenter{\hbox{\begin{tikzpicture}
            \begin{scope}
                \VeryThinLine (0:0) -- (-90:0.5);
                \VeryThinLine (0:0) -- (150:0.75);
                \VeryThinLine (0:0) -- (30:0.75);
                \draw[fill=black] (0:0) ellipse (0.05 and 0.05);
                \node at (150:1) {$i$};
                \node at (30:1) {$j$};
                \node at (-75:0.5) {$p$};
                \node at (-0.25,-0.125) {$\alpha_{s'}$};
            \end{scope}
            \begin{scope}[yshift=-1cm]
                \VeryThinLine (0:0) -- (90:0.5);
                \VeryThinLine (0:0) -- (-150:0.75);
                \VeryThinLine (0:0) -- (-30:0.75);
                \draw[fill=black] (0:0) ellipse (0.05 and 0.05);
                \node at (-150:1) {$k$};
                \node at (-30:1) {$l$};
                \node at (-0.25,0.125) {$\beta^*_{t'}$};
            \end{scope}
        \end{tikzpicture}}}, \vcenter{\hbox{\begin{tikzpicture}
            \begin{scope}
                \VeryThinLine (0:0) -- (-90:0.5);
                \VeryThinLine (0:0) -- (150:0.75);
                \VeryThinLine (0:0) -- (30:0.75);
                \draw[fill=black] (0:0) ellipse (0.05 and 0.05);
                \node at (150:1) {$i$};
                \node at (30:1) {$j$};
                \node at (-75:0.5) {$p$};
                \node at (-0.25,-0.125) {$\alpha_{s}$};
            \end{scope}
            \begin{scope}[yshift=-1cm]
                \VeryThinLine (0:0) -- (90:0.5);
                \VeryThinLine (0:0) -- (-150:0.75);
                \VeryThinLine (0:0) -- (-30:0.75);
                \draw[fill=black] (0:0) ellipse (0.05 and 0.05);
                \node at (-150:1) {$k$};
                \node at (-30:1) {$l$};
                \node at (-0.25,0.125) {$\beta^*_{t}$};
            \end{scope}
        \end{tikzpicture}}}}_{\mathcal{C}} = \delta_{s,s'}\delta_{t,t'}\frac{1}{d_p}. 
    \end{equation}
    Since $\widehat{\mathbf{E}^{ij}_{kl}} = \sum_{p\in \Irr(\mathcal{C})} (\mathcal{T}^{kl}_p)^* \mathcal{T}^{ij}_p$ gives an orthogonal decomposition of $\mathbf{E}^{ij}_{kl}$, we obtain:
\begin{equation}
        \mathbf{E}^{ij}_{kl}(\beta^*\alpha) = \sum_{s=1}^n\sum_{t=1}^m a_s\cdot b_t \braket{\beta^*_t\alpha_s,\beta^*\alpha}\beta^*_t\alpha_s = \frac{ab}{d_p}\beta^*\alpha. 
    \end{equation} 
    This proves the claim. 
\end{proof}

\begin{proposition}\label{prop:: key relation among eigenvalues}
    Let $i,j,k,l,p\in \Irr(\mathcal{C})$ such that $\Hom(ij,p)\neq 0$ and $\Hom(kl,p)\neq 0$. 
    Let $\alpha\in \Hom(ij,p)$ be an eigenvector of $\mathbf{D}^{ij}_p$ with eigenvalue $a>0$, and let $\beta\in \Hom(kl,p)$ be an eigenvector of $\mathbf{D}^{kl}_p$ with eigenvalue $b>0$. 
    Then we have
    \begin{equation}
        \frac{a}{b} = \sqrt{\frac{d_id_j}{d_kd_l}}. 
    \end{equation}
\end{proposition}
\begin{proof}
    By the previous lemma, $\beta^*\alpha$ is an eigenvector of $\mathbf{E}^{ij}_{kl}$ with eigenvalue $ab/d_p$. 
    The morphism space $\Hom(ij,kl)$ decomposes as
    \begin{equation}
        \bigoplus_{p\in \Irr(\mathcal{C})} \Hom(ij,p) \Hom(p,kl) = \Hom(ij,kl) = \bigoplus_{q\in \Irr(\mathcal{C})} \left( \Hom(iq,k)\otimes \mathrm{Id}_l \right)\left( \mathrm{Id}_i\otimes \Hom(j,ql) \right),
    \end{equation} 
    For each simple object $q$, choose orthonormal eigenbases of $\mathbf{D}^{iq}_k$ and $\mathbf{D}^{ql}_j$. 
    The morphisms $(\zeta\otimes \mathrm{Id}_l)(\mathrm{Id}_i\otimes \delta^*)$ obtained from these bases span $\Hom(ij,kl)$. 
    Since $\beta^*\alpha\neq 0$, there exists a simple object $q$, an eigenvector $\zeta\in \Hom(iq,k)$ of $\mathbf{D}^{iq}_k$ with eigenvalue $c>0$, and an eigenvector $\delta\in \Hom(ql,j)$ of $\mathbf{D}^{ql}_j$ with eigenvalue $d>0$, such that
    \begin{equation}\label{eqn:: non-zero overlap}
        0\neq \braket{\beta^*\alpha,(\zeta\otimes \mathrm{Id}_l)(\mathrm{Id}_i\otimes \delta^*)} = \vcenter{\hbox{\begin{tikzpicture}[scale=0.72, baseline=-0.5ex, every node/.style={font=\scriptsize, inner sep=1pt}]
    \coordinate (B) at (0,2.0);
    \coordinate (A) at (0,0.75);
    \coordinate (G) at (0,-1.55);
    \coordinate (D) at (1.35,-0.55);

    \VeryThinLine[midarrow]
        (G) .. controls (-0.35,-1.95) and (-0.85,-2.05) .. (-1.18,-1.5)
            .. controls (-1.85,-0.4) and (-1.95,0.45) .. (-1.35,1.1)
            .. controls (-0.9,1.6) and (-0.45,2.15) .. (B);

    \VeryThinLine[midarrow]
        (D) .. controls (2.05,-2.0) and (2.95,-1.35) .. (2.85,0.05)
            .. controls (2.75,1.85) and (0.85,2.85) .. (B);

    \VeryThinLine[midarrow] (B) -- (A);
    \VeryThinLine[midarrow]
        (A) .. controls (-0.55,0.25) and (-0.9,-0.15) .. (-0.9,-0.65)
            .. controls (-0.85,-1.15) and (-0.45,-1.35) .. (G);
    \VeryThinLine[midarrow]
        (A) .. controls (0.45,0.25) and (0.85,0.0) .. (D);
    \VeryThinLine[midarrow] (D) -- (G);

    \draw[fill=black] (B) circle (2pt);
    \draw[fill=black] (A) circle (2pt);
    \draw[fill=black] (G) circle (2pt);
    \draw[fill=black] (D) circle (2pt);

    \node at (-2.15,0.35) {$k$};
    \node at (3.12,0.35) {$l$};
    \node at (0.28,1.35) {$p$};
    \node at (-0.55,-0.35) {$i$};
    \node at (0.78,0.34) {$j$};
    \node at (0.94,-1.35) {$q$};

    \node at (-0.42,1.65) {$\beta$};
    \node at (-0.55,0.55) {$\alpha^{*}$};
    \node at (-0.42,-1.58) {$\zeta$};
    \node at (1.62,-0.35) {$\delta^{*}$};
\end{tikzpicture}}}. 
    \end{equation}
    As in Lemma ~\ref{lemma:: key lemma on eigenvalues of D3}, the morphism $(\zeta\otimes \mathrm{Id}_l)(\mathrm{Id}_i\otimes \delta^*)$ is an eigenvector with eigenvalue $cd/d_q$ of the operator on $\Hom(ij,kl)$ corresponding to the $(ij,kl)$-block of $(\mathcal{T}\otimes \mathrm{Id}_{\gamma})(\mathrm{Id}_{\gamma}\otimes\mathcal{T}^*)$.
    The Frobenius condition identifies this operator with $\mathbf{E}^{ij}_{kl}$.
    Since $\mathbf{E}^{ij}_{kl}$ is self-adjoint and the two eigenvectors above have nonzero inner product, their eigenvalues must be the same:
    \begin{equation}
        \frac{ab}{d_p} = \frac{cd}{d_q}.
    \end{equation}

    By spherical isotopy, the diagram in Equation~\eqref{eqn:: non-zero overlap} is also the overlap between morphisms $\vcenter{\hbox{\begin{tikzpicture}[scale = 0.75,rotate= 30]
          \draw[fill=black] (-0.5,0) ellipse (0.06 and 0.06);
          \draw[fill=black] (0.5,0) ellipse (0.06 and 0.06);
          \VeryThinLine (-0.5,0) -- (0.5,0);
          \VeryThinLine (-0.5,0) -- (-145:1);
          \VeryThinLine (-0.5,0) -- (145:1);
          \VeryThinLine (0.5,0) -- (35:1);
          \VeryThinLine (0.5,0) -- (-35:1);
    \end{tikzpicture}}}$ and $\vcenter{\hbox{\begin{tikzpicture}[scale = 0.75]
          \draw[fill=black] (0,-0.5) ellipse (0.06 and 0.06);
          \draw[fill=black] (0,0.5) ellipse (0.06 and 0.06);
          \VeryThinLine (0,0.5) -- (0,-0.5);
          \VeryThinLine (0,0.5) -- (125:1);
          \VeryThinLine (0,0.5) -- (55:1);
          \VeryThinLine (0,-0.5) -- (-55:1);
          \VeryThinLine (0,-0.5) -- (-125:1);
    \end{tikzpicture}}}$, with the inner strands labelled by $(j,k)$ and $(i,l)$ and the vertices labelled by the respective one-click rotations of the corresponding morphisms. 
    For $a,b,c\in\Irr(\mathcal{C})$, denote by $J$ the modular conjugation $J:\Hom(ab,c)\rightarrow \Hom(\overline{b}\overline{a},\overline{c})$, and denote by $\rho:\Hom(ab,c)\rightarrow \Hom(\overline{c}a,\overline{b})$ the one-click rotation. 
    Then we have 
    \begin{equation}
        \begin{aligned}
            0\neq \braket{\beta^*\alpha,(\zeta\otimes \mathrm{Id}_l)(\mathrm{Id}_i\otimes \delta^*)} &= \left\langle
(\rho J(\alpha))^{*}\rho J(\zeta),
(\beta\otimes\mathrm{Id}_{\overline{j}})
\bigl(\mathrm{Id}_k\otimes(\rho^{2}(\delta))^{*}\bigr)
\right\rangle\\
&= \left\langle
(\alpha\otimes\mathrm{Id}_{\overline{l}})
\bigl(\mathrm{Id}_i\otimes(\rho J(\delta))^{*}\bigr),
(\rho J(\beta))^{*}\zeta
\right\rangle. 
        \end{aligned}
    \end{equation}
    Rotation and modular-conjugation invariance of $\mathcal{T}$ thus gives
    \begin{equation}
        \begin{aligned}
            \frac{ac}{d_i} &= \frac{bd}{d_l}\\
            \frac{ad}{d_j} &= \frac{bc}{d_k}.
        \end{aligned}
    \end{equation}
    Solving these three equations over positive real numbers gives
    \begin{equation}
        (a,b,c,d) = (t\sqrt{d_id_jd_p},t\sqrt{d_kd_ld_p},t\sqrt{d_id_qd_k},t\sqrt{d_qd_ld_j}),\quad t>0. 
    \end{equation}
    Hence $a/b = \sqrt{d_id_j/(d_kd_l)}$.
\end{proof}

\begin{theorem}\label{thm:: unique positive Frobenius morphism}
    Let $\mathcal{C}$ be a unitary fusion category. 
    Set $\gamma = \bigoplus_{x\in \Irr(\mathcal{C})}x\boxtimes x^{\mathrm{op}}$, and let $\mathcal{T}\in \Hom(\gamma^2,\gamma)$.
    Suppose that $\mathcal{T}$ is invariant under the one-click rotation and modular conjugation, strictly $\mathcal{F}$-positive, and satisfies the Frobenius identity~\eqref{eqn:: Frobenius identity}.
    Then there exists $\lambda>0$ such that
    \begin{equation}\label{eqn:: the unique solution to the constraints}
        \mathcal{T} = \bigoplus_{i,j,k\in \Irr(\mathcal{C})}\lambda\sum_{t^{ij}_k} \sqrt{d_id_jd_k} P_k \left( t^{ij}_{k}\boxtimes (t^{ij}_{k})^* \right) P_i\otimes P_j,
    \end{equation}
    where $t^{ij}_{k}$ belongs to an orthonormal basis of $\Hom(ij,k)$ with respect to the trace inner product, and $P_i$ is the projection onto the diagonal component $i\boxtimes i^{\mathrm{op}}$.
\end{theorem}
\begin{proof}
    Since $\Hom(\mathbb{1}\otimes \mathbb{1},\mathbb{1})$ is one-dimensional, $\mathbf{D}^{\mathbb{1},\mathbb{1}}_{\mathbb{1}} = \lambda$ for some $\lambda>0$.
    For a simple object $x$, the space $\Hom(\overline{x}\otimes x,\mathbb{1})$ is also one-dimensional; write $\mathbf{D}^{\overline{x},x}_{\mathbb{1}} = \lambda_x$.
    Proposition ~\ref{prop:: key relation among eigenvalues} gives $\lambda_x/\lambda = \sqrt{d_{\overline{x}}d_x/(d_{\mathbb{1}}d_{\mathbb{1}})} = d_x$, and hence $\mathbf{D}^{\overline{x},x}_{\mathbb{1}} = \lambda\cdot d_x$.
    By the $2\pi/3$-rotation invariance of $\mathcal{T}$, the corresponding rotated vector in $\Hom(\mathbb{1}\otimes x,x)$ is an eigenvector of $\mathbf{D}^{\mathbb{1}x}_x$ with the same eigenvalue $\lambda\cdot d_x$.

    Now let $i,j,k$ be simple objects such that $\Hom(i\otimes j,k)$ is nonzero, and let $\alpha$ be an eigenvector of $\mathbf{D}^{ij}_k$ with eigenvalue $\lambda^{ij}_{k}$. 
    The morphism $\beta = \vcenter{\hbox{\begin{tikzpicture}
        \VeryThinLine[midarrow] (30:0.5) -- (0,0);
        \VeryThinLine[midarrow] (0,0) -- (-90:0.5);
        \draw[dashed] (0,0) -- (150:0.5);
        \node at (30:0.75) {$k$};
    \end{tikzpicture}}}$ is an eigenvector of $\mathbf{D}^{\mathbb{1}k}_k$ with eigenvalue $\lambda\cdot d_k$. 
    Therefore, Proposition ~\ref{prop:: key relation among eigenvalues} gives
    \begin{equation}
        \lambda^{ij}_k = \lambda\cdot d_k\cdot \sqrt{\frac{d_id_j}{d_k}} = \lambda\cdot\sqrt{d_id_jd_k}. 
    \end{equation}
    Since our choice of $\alpha$ is arbitrary, we obtain
    \begin{equation}
        \mathbf{D}^{ij}_k = \lambda\cdot \sqrt{d_id_jd_k}. 
    \end{equation}
    This completes the proof. 
\end{proof}

\begin{lemma}\label{lemma:: rotation invariance of T}
    Let $Z$ satisfy the assumptions of Proposition ~\ref{prop:: canonical operator from connected stratifications}, and let $\mathfrak D=\mathfrak D_{\Lambda\mathcal S(2,1)}$.
    The one-click rotation of the morphism space $\mathcal A(2,1)$ leaves $\mathfrak D$ invariant.
    Moreover, $\mathfrak D$ is invariant under categorical modular conjugation.
\end{lemma}
\begin{proof}
Let $\cM=\Lambda\mathcal S(2,1)$ be equipped with its marked point and the standard regular chart, so that $\mathcal H_{\cM}=\mathcal H_{\mathcal S(2,1)}$.
Let $r$ be the orientation-preserving one-click rotation of $\Lambda\mathcal S(2,1)$, and let $U:=\pi_{\mathcal S(2,1)}([r])\in\mathcal B(\mathcal H_{\cM})$ be the corresponding mapping-class-group action on the label space, which is unitary. 

Write $q=q_{\cM}:\mathcal H_{\cM}\rightarrow\mathcal A(2,1)$, and denote by $R$ the one-click rotation on $\mathcal A(2,1)$.
The covariance of labels under reparameterization of the regular chart implies that rotating the labelled marked stratification $\cM[\psi]$ by $r$ is represented by $\cM[U\psi]$.
It then follows that
\begin{equation}\label{eq:: rotation equivariance of local quotient map}
    qU=Rq.
\end{equation}
The spherical $C^*$-structure on $\mathcal C(Z)$ makes planar rotation unitary for the categorical trace inner product, so $R^\dagger=R^{-1}$.
Taking adjoints in Equation~\eqref{eq:: rotation equivariance of local quotient map} and using the unitarity of $U$ and $R$ gives $Uq^\dagger R^{-1}=q^\dagger$. 
Consequently,
\begin{equation}
\begin{aligned}
    R\mathfrak D R^{-1}=Rq q^\dagger R^{-1} =qUq^\dagger R^{-1} =qq^\dagger
    =\mathfrak D.
\end{aligned}
\end{equation}
On the local label space, reflection is the existing anti-unitary $\hat\theta$.
By s-reflection positivity it descends to the categorical modular conjugation $J$, and the quotient map satisfies $q\hat\theta=Jq$.
Therefore,
\begin{equation}
    J\mathfrak D J^{-1}
    =(q\hat\theta)(q\hat\theta)^\dagger
    =qq^\dagger
    =\mathfrak D.
\end{equation}
\end{proof}

\begin{corollary}\label{corollary:: local Hilbert space of reconstructed model}
    Let $Z$ be a homeomorphism-invariant, multiplicative, s-reflection-positive, locally non-degenerate $\mathbb{S}^2$-functional. 
    Suppose that $Z$ satisfies topological connectedness~\ref{def:: topological connectedness} and commutativity~\ref{def:: commutativity}.
    Then there is a vector-space isomorphism
    \begin{equation}
        \mathcal{H}_{\mathcal{S}(3)} \cong \bigoplus_{i,j,k\in \Irr(\mathcal{C}(Z))} \Hom(i\otimes j,k).
    \end{equation}
    Moreover, the summands on the right-hand side are mutually orthogonal, and there exists $\lambda>0$ such that the inner product on each summand is given by
    \begin{equation}
        \braket{\beta|\alpha} = \frac{1}{\lambda\sqrt{d_id_jd_k}}\tr(\beta^*\alpha),
        \quad \alpha,\beta\in \Hom(i\otimes j,k).
    \end{equation}
\end{corollary}
\begin{proof}
Identify $\mathcal{H}_{\mathcal{S}(3)}$ with $\mathcal{H}_{\Lambda\mathcal{S}(2,1)}$, and write $q=q_{\Lambda\mathcal{S}(2,1)}:\mathcal{H}_{\mathcal{S}(3)}\rightarrow \mathcal{A}(2,1)$.
By local non-degeneracy~\ref{def:: local non degeneracy}, $q$ is injective; by Theorem ~\ref{Thm:: key consequence of TC}, it is surjective because $\Lambda\mathcal{S}(2,1)^1$ is connected.
Hence $q$ is a linear isomorphism.

By Corollary ~\ref{corollary:: T is block diagonal}, $\mathbb{T}=\widehat{\mathfrak{D}}$ lies in $\Hom(\gamma^2,\gamma)$, where $\mathfrak{D}=q q^\dagger$ and $\gamma=\bigoplus_{x\in\Irr(\mathcal{C}(Z))}x\boxtimes x^{\mathrm{op}}$.
    Moreover, $\mathbb{T}$ is strictly $\mathcal{F}$-positive by construction, invariant under one-click rotation and modular conjugation by Lemma~\ref{lemma:: rotation invariance of T}, and satisfies the Frobenius identity by Proposition~\ref{prop:: Frobenius identity from TC}.
Corollary ~\ref{corollary:: A is sum of simple objects} gives complete finiteness, so $\mathcal{C}(Z)$ is a unitary fusion category.
Therefore, Theorem ~\ref{thm:: unique positive Frobenius morphism} applies to $\mathcal{C}(Z)$ and $\mathcal{T}=\mathbb{T}$.
By Corollary ~\ref{corollary:: A is sum of simple objects}, $A\cong\bigoplus_{x\in\Irr(\mathcal{C}(Z))}x$, so the underlying morphism space decomposes as
\begin{equation}
    \mathcal{A}(2,1)=\Hom(A^2,A)\cong\bigoplus_{i,j,k\in\Irr(\mathcal{C}(Z))}\Hom(i\otimes j,k),
\end{equation}
and Theorem ~\ref{thm:: unique positive Frobenius morphism} gives a scalar $\lambda>0$ such that the operator $\mathfrak{D}$ restricts to
\begin{equation}
    \mathfrak{D}\vert_{\Hom(i\otimes j,k)}
    =
    \lambda\sqrt{d_id_jd_k}\,\mathrm{Id}_{\Hom(i\otimes j,k)}.
\end{equation}
This proves the asserted vector-space decomposition of $\mathcal{H}_{\mathcal{S}(3)}$ via the isomorphism $q$.

It remains to identify the Hilbert space inner product. Since $q$ is invertible and $\mathfrak{D}=q q^\dagger$, we have for $x,y\in\mathcal{A}(2,1)$,
\begin{equation}
    \braket{q^{-1}y|q^{-1}x}_{\mathcal{H}_{\mathcal{S}(3)}}
    =
    \braket{y|\mathfrak{D}^{-1}x}_{\mathcal{C}(Z)}.
\end{equation}
Thus the above block formula for $\mathfrak{D}$ gives, for $\alpha,\beta\in\Hom(i\otimes j,k)$,
\begin{equation}
    \braket{q^{-1}\beta|q^{-1}\alpha}_{\mathcal{H}_{\mathcal{S}(3)}}
    =
    \frac{1}{\lambda\sqrt{d_id_jd_k}}\tr(\beta^*\alpha),
\end{equation}
and different simple summands are orthogonal. 
\end{proof}

\begin{theorem}\label{thm:: reconstructed state is Levin-Wen ground state}
    Let $Z$ be a homeomorphism-invariant, multiplicative, s-reflection-positive, and locally non-degenerate $\mathbb{S}^2$-functional. 
    Suppose that $Z$ satisfies topological connectedness~\ref{def:: topological connectedness} and commutativity~\ref{def:: commutativity}. 
    Then for any stratification $\cM$ of $\mathbb{S}^2$ with $\cM^1$ connected, the induced vector $\ket{\tilde{\Psi}_{\cM}}$ maps, under the local unitary identification of Corollary~\ref{corollary:: local Hilbert space of reconstructed model}, to a nonzero ground-state vector of the Levin--Wen model with input unitary fusion category $\mathcal{C}(Z)$. 
\end{theorem}
\begin{proof}
Fix an orientation $o$ of the edges of $\cM^1$, and let $\Gamma$ be the resulting directed embedded graph.
Write $V$ and $F$ for the numbers of vertices and faces of $\Gamma$, respectively.
By Corollary ~\ref{corollary:: A is sum of simple objects}, the reconstructed object $A=(1,\mathrm{Id}_1)$ is isomorphic to $\bigoplus_{x\in\Irr(\mathcal C(Z))}x$, which is the object used to define the Levin--Wen local Hilbert spaces in Section~\ref{section:: String-net wave function}.
    By Corollary ~\ref{corollary:: local Hilbert space of reconstructed model} and the skein-module inner product in Equation~\eqref{eqn:: skein module inner product}, the quotient map $q:\mathcal H_{\mathcal S(3)}\rightarrow \Hom(A^2,A)$ satisfies 
\begin{equation}
    \braket{q\eta,q\xi}_{\mathrm{LW}} = \lambda\braket{\eta,\xi}_{\mathcal H_{\mathcal S(3)}}
\end{equation}
for all $\xi,\eta\in\mathcal H_{\mathcal S(3)}$, where $\braket{\cdot,\cdot}_{\mathrm{LW}}$ denotes the local inner product of the Levin--Wen model in Section~\ref{section:: String-net wave function}.
At each vertex $v$, use the ordering supplied by its regular chart and the orientation $o$ to bend the three boundary strands into the convention of Section~\ref{section:: String-net wave function}.
The resulting map $b_v:\Hom(A^2,A)\rightarrow\mathcal H_v$ is unitary for the skein-module inner product, and hence $U_v:=\lambda^{-1/2}b_vq:\mathcal H_{\mathcal S(3)}\rightarrow\mathcal H_v$ is unitary.
Set $U_{\cM}:=\bigotimes_{v\in\cM^0}U_v:\mathcal H_{\cM}\rightarrow\mathcal H_{\Gamma}$.

Let $D\subset\mathbb S^2$ be a disk containing $\Gamma$ and all but one face in its interior, as in Corollary ~\ref{corollary:: ground state on sphere as evaluation}.
For a pure tensor $\psi=\bigotimes_{v\in\cM^0}\psi_v\in\mathcal H_{\cM}$, the definition of the reconstructed category by gluing labelled stratifications gives
\begin{equation}
Z(\psi)\,\mathrm{Id}_{\mathbb{1}}
=
\mathrm{eval}_D\left(\bigotimes_{v\in\cM^0}b_vq(\psi_v)\right)
=
\lambda^{V/2}\mathrm{eval}_D\bigl(U_{\cM}\psi\bigr).
\end{equation}
Indeed, composition and tensor product in $\mathcal C(Z)$ are induced by gluing, while closing the resulting morphism evaluates it by $Z$.
Since $D$ contains $F-1$ plaquettes, Corollary ~\ref{corollary:: ground state on sphere as evaluation} gives
\begin{equation}
\mathrm{eval}_D(\varphi)
=
\mu^{F-1}\braket{\Psi_{\Gamma}|\varphi}\,\mathrm{Id}_{\mathbb{1}},
\qquad
\varphi\in\mathcal H_{\Gamma},
\end{equation}
where $\ket{\Psi_{\Gamma}}$ is a Levin--Wen ground state.
Using the defining identity in Equation~\eqref{eq:: wave function from functional} for $\tilde{\Psi}_{\cM}$, we therefore obtain
\begin{equation}
\braket{U_{\cM}\tilde{\Psi}_{\cM}|U_{\cM}\psi}
=
\lambda^{V/2}\mu^{F-1}
\braket{\Psi_{\Gamma}|U_{\cM}\psi}.
\end{equation}
Pure tensors span $\mathcal H_{\cM}$ and $U_{\cM}$ is unitary, so
    \begin{equation}\label{eqn:: reconstructed state as Levin-Wen ground state}
    U_{\cM}\ket{\tilde{\Psi}_{\cM}} = \lambda^{V/2}\mu^{F-1}\ket{\Psi_{\Gamma}}.
\end{equation}
Thus $U_{\cM}\ket{\tilde{\Psi}_{\cM}}$ lies in the Levin--Wen ground-state subspace, which proves the claim.
\end{proof}

\begin{remark}
    The scalar $\lambda$ measures the overall discrepancy between the prescribed inner product on the local label space $\mathcal H_{\mathcal S(3)}$ and the skein-module inner product reconstructed from $Z$.
    Topological connectedness determines this scale: $\lambda=\sqrt{\tau}/\mu$.
    Indeed, let $\mathcal Y$ be the one-vertex trivalent tree in a disk with boundary $\mathcal S(3)$, so $\lvert F_{\mathcal Y}\rvert=0$, and let $\mathcal T$ be the triangle with one boundary leg at each vertex, so $\lvert F_{\mathcal T}\rvert=1$.
    TC1 gives
    \begin{equation}
        \widetilde\rho_{\mathbb D^2_+}\bigl(\theta(\mathcal T)\cup\mathcal Y\bigr)
        =\tau\,\widetilde\rho_{\mathbb D^2_+}\bigl(\theta(\mathcal Y)\cup\mathcal Y\bigr).
    \end{equation}
    Taking traces and using Equation~\eqref{eqn:: reconstructed state as Levin-Wen ground state}, the graph $\theta(\mathcal T)\cup\mathcal Y$ has $(V,F)=(4,4)$, whereas $\theta(\mathcal Y)\cup\mathcal Y$ has $(V,F)=(2,3)$.
    Therefore $\lambda^4\mu^6=\tau\lambda^2\mu^4$, so $\tau=\lambda^2\mu^2$ and hence $\lambda=\sqrt{\tau}/\mu$.
    For the canonical functional $Z_{\mathcal C}$ of Section~\ref{subsection:: S^2-functional from LW model}, $\tau=\mu^2$, so $b_vq$ is unitary for the skein-module inner product and $\lambda=1$. 
\end{remark}

\section{Discussion}

We end the paper with a brief discussion of the origin of some of the concepts in this paper. 
First of all, reflection positivity originally appeared in constructive quantum field theory \cite{OS1973}, and was later adapted to the study of classical/quantum statistical mechanics \cite{FILS1978}.
We emphasize that the notion of RP used here is more related to the second scenario, as the reflection positivity of the ground state of a Hamiltonian can be derived from that of the Hamiltonian. 
The reflection positivity of the Levin--Wen Hamiltonian is obtained in \cite{JL2017}. 
Second, although proved in a purely algebraic way, the topological connectedness of the Levin--Wen ground-state wave function admits a diagrammatic proof using the topological connectedness of the alterfold TQFT \cite{alterfolds2023I,alterfolds2023II,alterfolds2024}. 
More precisely, the topological connectedness is a direct consequence of invariance of the alterfold partition function under Move $1$ \cite[Section 3]{alterfolds2023I}, together with the presentation of the evaluation map \cite[Remark 4.25]{alterfolds2023II}. 

The definition of $\mathcal{F}$-positive morphisms draws its inspiration from the quantum Fourier analysis \cite{JaffeLiu2020,HJLW2023}, which extends Fourier analysis on group symmetries to generalized symmetries. 
Indeed, in one spatial dimension (that is, for $\mathbb{S}^1$-functionals), the correspondence $\mathfrak{D}\mapsto \widehat{\mathfrak{D}} $ is precisely the inverse Fourier transform of a $2$-box in a (possibly reducible) planar algebra. 
In that case, the topological connectedness translates to the condition that $\mathfrak{D}$ is a positive operator with inverse Fourier transform an idempotent. 
As shown in \cite{Liu2016ExchangeRP}, such objects are precisely biprojections \cite{bisch1994biprojection} when the underlying planar algebra is irreducible. 
Thus, topological connectedness generalizes the notion of biprojections in higher dimensions. 
On the other hand, biprojections satisfy the so-called exchange relation \cite{Landau2002Exchange}. 
Therefore, generalizing the exchange relation to higher dimensions may serve as the appropriate weakening of the topological connectedness condition. 
Finally, we expect close connections between the exchange relation and quantum Markov states \cite{HaydenPetz2004}. 
We leave a detailed investigation of this relation for future work. 

\appendix

\section{Reconstruction of the spherical unitary tensor category}\label{app:: reconstruction from S2 functional}

This appendix follows \cite{Liu2024}, which was inspired by the early work of Jones on subfactor planar algebras \cite{PlanarAlgebra2021}, and recalls the construction of the unitary spherical tensor category $\mathcal{C}(Z)$ from a locally finite, homeomorphism-invariant, s-reflection-positive, multiplicative $\mathbb{S}^2$-functional $Z$. 
The s-reflection positivity of $Z$ allows us to define a positive definite inner product on $\tilde{V}_m$ by 
\begin{equation}\label{eq:: inner product from RP}
    \braket{b,a}_Z = Z(\hat{\theta}(b)\otimes a),\quad a,b\in \tilde{V}_m. 
\end{equation}
Thus each $\tilde{V}_m$ becomes a finite-dimensional Hilbert space. 

For each $m\geq 0$, we endow $\tilde{V}_{2m}$ with a $C^*$-algebra structure $\mathcal{A}_m$ as follows. 
We view $\mathcal{S}(2m)$ as the boundary stratification $\partial(\mathcal{D}(m)\times \mathbb{D}^1)$, where $\mathcal{D}(m)$ is the stratification of $\mathbb{D}^1$ with $m$ points in the $0$-skeleton:
\begin{equation}
    \mathcal{D}(m) = \vcenter{\hbox{\begin{tikzpicture}
        \VeryThinLine (0,0) -- (3,0);
        \draw[fill=black] (0.5,0) circle (1.5pt);
        \draw[fill=black] (1.5,0) circle (1.5pt);
        \draw[fill=black] (2.5,0) circle (1.5pt);
        \node at (0.5,0.25) {$1$};
        \node at (1.5,0.25) {$2$};
        \node at (2,0.25) {$\cdots$};
        \node at (2.5,0.25) {$m$};
    \end{tikzpicture}}}
\end{equation}
In this way, every labelled stratification $x$ with boundary $\mathcal{S}(2m)$ gives rise to an operator $L_x\in \mathcal{B}(\tilde{V}_{2m})$ by vertical stacking along the last coordinate and then quotienting out null vectors. 
This action depends only on the equivalence class $[x]$, since by homeomorphism invariance, stacking with a null vector will always give a null vector. 
The identity $1_m\in\mathcal{A}_m$ is represented by $\mathcal{D}(m)\times\mathbb{D}^1$.
The $*$-structure on $\mathcal{A}_m$ is defined by imposing $L_{x^*} = L^\dagger_x$. 
Denote by $\rho$ the $180$-degree rotation of $\mathbb{S}^2$ about the $x$-axis, which induces an orientation-preserving homeomorphism from stratifications of $\mathbb{D}^2_+$ to those of $\mathbb{D}^2_-$. 
Then one has (\cite[Def. 5.15]{Liu2024}):
\begin{equation}
    x^* = \rho(\hat{\theta}(x)) = \hat{\theta}(\rho(x)),\quad x\in \mathcal{A}_m. 
\end{equation}
The algebras $\mathcal{A}_m$ describe endomorphisms of a fixed boundary type. 
There is a trace $\tr_m:\mathcal{A}_m\to\mathbb{C}$ defined by 
\begin{equation}
    \tr_m(xy^*) = Z(\hat{\theta}(y)\otimes x) = \braket{y,x}_Z,\quad x,y\in \mathcal{A}_m. 
\end{equation} 
Its traciality follows from homeomorphism invariance, while s-reflection positivity implies that it is positive and faithful. 
Moreover, $\operatorname{tr}_{\mathbf 1}(\operatorname{Id}_{\mathbf 1})=Z(\mathbb{S}^2_{\emptyset})=1$. 

To compare projections belonging to different algebras, we also need morphisms between distinct boundary types. 
For $m,n\geq 0$, let $\mathcal{S}(m,n)$ be the stratification of $\mathbb{S}^1$ that is homeomorphic to $\mathcal{S}(m+n)$ but  is displayed as the boundary of a rectangle, with $m$ and $n$ marked points on the upper and lower boundary intervals, respectively. 
In particular, $\mathcal{S}(2,1)$ is the rectangular presentation of $\mathcal{S}(3)$.
We regard the upper boundary as the input and the lower boundary as the output. 
Define
\begin{equation}\label{eq:: rectangular sector}
    \mathcal{A}(m,n)
    :=
    \mathbb{D}^2_{\mathcal{S}(m,n)}(\mathcal{H}_{\bullet})
    \big/ \ker_{m+n} (Z).
\end{equation}
Note that homeomorphism invariance and s-reflection positivity give a natural unitary equivalence $\mathcal{A}(m,n)\cong \tilde{V}_{m+n}$ of Hilbert spaces. 
In particular, $\mathcal{A}(m,m)=\mathcal{A}_m$, whereas $\mathcal{A}(m,n)$ is not an algebra when $m\neq n$.

Given $x\in\mathcal{A}(m,n)$ and $y\in\mathcal{A}(n,k)$, we define $yx\in\mathcal{A}(m,k)$ by vertically stacking a representative of $x$ above a representative of $y$ and gluing the two boundary intervals carrying $n$ marked points. 
Thus gluing induces bilinear composition maps
\begin{equation}\label{eq:: rectangular composition}
    \mathcal{A}(n,k)\otimes\mathcal{A}(m,n)
    \rightarrow
    \mathcal{A}(m,k),
    \qquad
    y\otimes x \mapsto yx.
\end{equation}
These maps are well-defined on the quotient because gluing a null vector into any larger labelled stratification again produces a null vector.
Associativity follows from homeomorphism invariance. 

Using the same reflection-rotation operation that defines the adjoint on $\mathcal{A}_m$, we obtain conjugate-linear maps $*:\mathcal{A}(m,n)\rightarrow\mathcal{A}(n,m)$ that satisfy
\begin{equation}
    (yx)^*=x^* y^*,
    \qquad
    (x^*)^*=x.
\end{equation}
The spaces $\mathcal{A}(m,n)$, with composition by vertical stacking and adjoint by reflection-rotation, form a strict $C^*$-category $\breve{\mathcal{C}}(Z)$ whose objects are non-negative integers and whose morphism spaces are 
\begin{equation}
    \Hom_{\breve{\mathcal{C}}(Z)}(m,n):=\mathcal{A}(m,n).
\end{equation}
The endomorphism algebra of the object $m$ is precisely $\End_{\breve{\mathcal{C}}(Z)}(m)=\mathcal{A}_m$. 
The induced $C^*$-norm on $\mathcal{A}(m,n)$ is defined by $\Vert x\Vert^2 = \Vert x^*x\Vert_{\mathcal{A}_m}$. 

We now take the idempotent completion $\mathcal{C}(Z)$ of $\breve{\mathcal{C}}(Z)$. 
Objects of $\mathcal{C}(Z)$ are formal direct sums of pairs $(m,p)$, where $p=p^2=p^*\in\mathcal{A}_m$ is a projection. 
For projections $p\in\mathcal{A}_m$ and $q\in\mathcal{A}_n$, define
\begin{equation}\label{eq:: hom in idempotent completion}
    \Hom((m,p),(n,q))
    :=
    q\,\mathcal{A}(m,n)\,p. 
\end{equation}
An object $(m,p)$ is simple if $p$ is a minimal projection in $\mathcal{A}_m$. 
Two idempotents $(m,p)$ and $(n,q)$ are equivalent if there exists $v\in\mathcal{A}(m,n)$ such that
\begin{equation}\label{eq:: partial isometry equivalence}
    v^*v=p,\quad v v^*=q.
\end{equation}
The isomorphism classes of simple objects are the equivalence classes of minimal projections in the family $\{\mathcal{A}_m\}_{m\geq 0}$.

Horizontal juxtaposition of labelled stratifications defines the tensor product of morphisms 
\begin{equation}
    \mathcal{A}(m,n)\otimes\mathcal{A}(m',n')
    \rightarrow
    \mathcal{A}(m+m',n+n'),
\end{equation}
which extends to the idempotent completion by $(m,p)\otimes(n,q):=(m+n,p\otimes q)$. 
Together with the duality induced by planar rotation, this construction gives the $C^*$ tensor category reconstructed from $Z$. 
Multiplicativity, namely $\mathcal{A}_0\cong\mathbb{C}$, implies that the tensor unit is simple. 
The traces $\{\tr_m\}_{m\geq 0}$ restrict to the corners $p\mathcal{A}_m p=\End_{\mathcal{C}(Z)}(m,p)$ and define a balanced trace on $\mathcal{C}(Z)$. 
Indeed, for $x\in q\mathcal{A}(m,n)p$ and $y\in p\mathcal{A}(n,m)q$, homeomorphism invariance identifies the two closures and gives
\begin{equation}
    \tr_m(yx)=\tr_n(xy).
\end{equation}
Together with the duality induced by planar rotation, this balanced trace makes $\mathcal{C}(Z)$ spherical. 

\bibliographystyle{alpha}
\bibliography{Reference}
\end{document}